\documentclass[11pt]{article}

\usepackage[top=1in, bottom=1in, left=1in, right=1in]{geometry}

\usepackage[style=alphabetic,natbib=true,maxbibnames=99,backend=biber]{biblatex}

\usepackage{microtype} 
\usepackage{graphicx}
\usepackage{subfigure}
\usepackage{booktabs} % 
\PassOptionsToPackage{hyphens}{url}
 
\usepackage{hyperref}

\usepackage{algcompatible}
\usepackage{algorithm}

\usepackage{amsmath}
\usepackage{amssymb}
\usepackage{mathtools}
\usepackage{amsthm}
\usepackage{bbm}

\usepackage[capitalize,noabbrev]{cleveref}

\theoremstyle{plain}
\newtheorem{theorem}{Theorem}[section]

\newtheorem{lemma}[theorem]{Lemma}

\theoremstyle{definition}
\newtheorem{definition}[theorem]{Definition}

\theoremstyle{remark}
\newtheorem{remark}[theorem]{Remark}

\newtheorem{example}[theorem]{Example}

\newtheorem{problem}{Problem}

\AfterEndEnvironment{definition}{\noindent\ignorespaces}
\AfterEndEnvironment{assumption}{\noindent\ignorespaces}
\AfterEndEnvironment{lemma}{\noindent\ignorespaces}
\AfterEndEnvironment{theorem}{\noindent\ignorespaces}
\AfterEndEnvironment{proposition}{\noindent\ignorespaces}
\AfterEndEnvironment{fact}{\noindent\ignorespaces}
\AfterEndEnvironment{question}{\noindent\ignorespaces}
\AfterEndEnvironment{corollary}{\noindent\ignorespaces}
\AfterEndEnvironment{model}{\noindent\ignorespaces}
\AfterEndEnvironment{remark}{\noindent\ignorespaces}
\AfterEndEnvironment{proof}{\noindent\ignorespaces}
\AfterEndEnvironment{fact}{\noindent\ignorespaces}

\usepackage[textsize=tiny]{todonotes}

\usepackage{mathrsfs}
\usepackage{mathtools} 
\usepackage{dsfont}
\usepackage{color}
\usepackage{color,fancybox,graphicx,subfigure,multicol} 
\usepackage[inline]{enumitem}

\newcommand{\white}[1]{\textcolor{white}{#1}}

\newcommand{\blue}[1]{\textcolor{black}{#1}}
 
\crefname{section}{Section}{Sections}
\crefname{theorem}{Theorem}{Theorems}
\crefname{assumption}{Assumption}{Assumptions}
\crefname{lemma}{Lemma}{Lemmas}
\crefname{definition}{Definition}{Definitions}
\crefname{conjecture}{Conjecture}{Conjectures}
\crefname{corollary}{Corollary}{Corollaries}
\crefname{construction}{Construction}{Constructions}
\crefname{claim}{Claim}{Claims}
\crefname{observation}{Observation}{Observations}
\crefname{proposition}{Proposition}{Propositions}
\crefname{fact}{Fact}{Facts}
\crefname{question}{Question}{Questions}
\crefname{problem}{Problem}{Problems}
\crefname{remark}{Remark}{Remarks}
\crefname{example}{Example}{Examples}
\crefname{equation}{Equation}{Equations}
\crefname{appendix}{Section}{Sections}
\crefname{algorithm}{Algorithm}{Algorithms}
\crefname{model}{Model}{Models}
\crefname{figure}{Figure}{Figures}

\newcommand{\Ex}{\mathbb{E}}
\newcommand{\R}{\mathbb{R}}
\newcommand{\N}{\mathbb{N}}
\newcommand{\Exp}{\mathbb{E}}
\newcommand{\OPT}{\mathrm{OPT}}
\newcommand{\pos}{\mathrm{pos}}
\newcommand{\score}{F}
\newcommand{\cs}{\mathrm{CS}}

\newcommand{\app}{\mathrm{App}}
\newcommand{\poly}{\mathrm{poly}}

\newcommand{\prefs}[1]{\mathcal{L}(#1)}

\def\abs#1{\left| #1 \right|}
\def\sabs#1{| #1 |}

\newcommand{\sinsquare}[1]{[#1]}
\newcommand{\inbrace}[1]{\left\{#1\right\}}

\newcommand{\inparen}[1]{\left(#1\right)}
\newcommand{\insquare}[1]{\left[#1\right]}

\newcommand{\eps}{\varepsilon}
\renewcommand{\epsilon}{\varepsilon}
\newcommand{\cdf}{\ensuremath{{\rm cdf}}}

\newcommand{\yesnum}{\addtocounter{equation}{1}\tag{\theequation}}
\newcommand{\tagnum}[1]{\addtocounter{equation}{1}{\tag{#1)\ \ (\theequation}}}
\makeatletter
\newcommand{\customlabel}[2]{%
\protected@write \@auxout {}{\string \newlabel {#1}{{#2}{\thepage}{#2}{#1}{}} }%
\hypertarget{#1}{}
}
\makeatother

\newcommand{\Stackrel}[2]{\stackrel{\mathmakebox[\widthof{\ensuremath{#2}}]{#1}}{#2}}
\newcommand{\eat}[1]{}

\newcommand{\evF}{\ensuremath{\mathscr{F}}}
\newcommand{\evE}{\ensuremath{\mathscr{E}}}

\newcommand{\wh}[1]{\widehat{#1}}

\newcommand{\hF}{\wh{F}}
\newcommand{\hR}{\wh{R}}
\newcommand{\hS}{\wh{S}}

\newcommand{\wt}[1]{\widetilde{#1}}

\newcommand{\cK}{\mathcal{K}}

\newcommand{\calK}{\mathcal{K}}
\newcommand{\calO}{\mathcal{O}}

\newcommand{\cW}{\mathcal{W}}
\newcommand{\err}{\ensuremath{{\tau \sqrt{nk\log{\frac{m}{\delta}} }}}}
\newcommand{\errm}{\ensuremath{{\tau \sqrt{nm\log{\frac{m}{\delta}} }}}}

\newcommand{\errmult}{\ensuremath{{\frac{1}{\alpha}\sqrt{\frac{k}{n}\log{\frac{m}{\delta}} }}}}

\makeatletter
\def\moverlay{\mathpalette\mov@rlay}
\def\mov@rlay#1#2{\leavevmode\vtop{%
   \baselineskip\z@skip \lineskiplimit-\maxdimen
   \ialign{\hfil$\m@th#1##$\hfil\cr#2\crcr}}}
\newcommand{\charfusion}[3][\mathord]{
    #1{\ifx#1\mathop\vphantom{#2}\fi
        \mathpalette\mov@rlay{#2\cr#3}
      }
    \ifx#1\mathop\expandafter\displaylimits\fi}
\makeatother

\newcommand{\sfrac}[2]{{#1/#2}}

\title{\mbox{\hspace{-2.5mm}Subset Selection Based On Multiple Rankings in the Presence of Bias:}\\ \mbox{\hspace{-12.5mm} Effectiveness of Fairness Constraints for Multiwinner Voting Score Functions}}

\author{Niclas Boehmer \\ TU Berlin \and L. Elisa Celis \\ Yale University  \and Lingxiao Huang\\ Nanjing University \and Anay Mehrotra \\ Yale University \and Nisheeth K. Vishnoi \\ Yale University}

\begin{document}

\maketitle

\begin{abstract}
We consider the problem of subset selection  where one is given multiple rankings of items and the goal is to select the highest ``quality'' subset. Score functions from the multiwinner voting literature have been used to aggregate rankings into quality scores for subsets. We study this setting of subset selection problems when, in addition, rankings may contain systemic or unconscious biases toward a  group of items. For a general model of input rankings and biases, we show that requiring the selected subset to satisfy group fairness constraints can improve the quality of the selection with respect to  unbiased rankings. Importantly, we show that for  fairness constraints to be effective, different multiwinner score functions may require a drastically different number of rankings: While for some functions, fairness constraints need an exponential number of rankings  to recover a close-to-optimal solution, for others, this dependency is only polynomial. This result relies on a novel notion of ``smoothness'' of  submodular functions in this setting that quantifies how well a function can ``correctly'' assess the quality of items in the presence of bias. The results in this paper can be used to guide the choice of multiwinner score functions for the subset selection setting considered here; we additionally provide a tool to empirically enable this.  

\end{abstract}

\newpage
\tableofcontents
\newpage

\section{Introduction}
\label{sec:intro}

The task of selecting a size-$k$ subset from a set of $m$ items is a basic problem in machine learning and computer science with applications arising in recommender systems \cite{mcsherry2002diversity,DBLP:conf/kdd/El-AriniVSG09}, feature selection \cite{DBLP:journals/jmlr/GuyonE03}, search engines \cite{DBLP:conf/wsdm/AgrawalGHI09}, data summarization \cite{DBLP:conf/nips/ElhamifarK17,DBLP:conf/icml/AngelidakisKSZ22}, and algorithmic hiring \cite{DBLP:conf/atal/SchumannCFD19,Raghavan2020Hiring,hiring}. 
The simplest formulation of this problem assumes that every item has some reported utility, which leads to the selection of  $k$ items with the highest utility. 
In several applications, however, items are evaluated from multiple points of view, e.g.,  different users evaluating items in a group recommendation or shortlisting setting \cite{jameson2007recommendation,Raghavan2020Hiring,DBLP:conf/prima/StreviniotisC22},  or a single user evaluating items from different perspectives (e.g., in personalized recommendations, a user might be interested in being recommended items from different categories, and separate orderings for categories are produced \cite{DBLP:conf/kdd/El-AriniVSG09,DBLP:conf/recsys/StreviniotisC22,DBLP:conf/eumas/GawronF22}).
Moreover, these different evaluations oftentimes come in the form of {\em rankings} of the $m$ items instead of numerical scores \cite{DBLP:conf/fat/ChakrabortyPGGL19,DBLP:conf/recsys/StreviniotisC22,DBLP:conf/prima/StreviniotisC22}.
Asking for rankings instead of numerical scores  for items is a popular approach in training large language models with human feedback \cite{chatGPT}, and in reinforcement learning \cite{NIPS2017_d5e2c0ad}) because numerical scores have a higher elicitation cost and can lead to  aggregation and calibration issues \cite{griffin2004perspectives,DBLP:conf/allerton/MitliagkasGCV11,DBLP:conf/atal/WangS19}.

A principled method to aggregate  $n$ rankings of $m$ items into a single quality score function for subsets of $\{1,\ldots,m\}$ is the usage of ``score functions'' studied in the multiwinner voting literature within social choice theory \cite{faliszewski2017multiwinner,DBLP:conf/fat/ChakrabortyPGGL19,DBLP:conf/aaai/MondalBSP21,DBLP:conf/recsys/StreviniotisC22,DBLP:conf/eumas/GawronF22,DBLP:conf/prima/StreviniotisC22,DBLP:series/sbis/LacknerS23}.
Here, many score functions have been proposed and their properties and merits have been extensively discussed \cite{DBLP:journals/scw/ElkindFSS17,DBLP:conf/aaai/ElkindFLSST17,DBLP:conf/ijcai/LacknerS19}. 
These functions, for instance, allow for specifying which part of each ranking to take into account, and how many of the selected items are relevant from the perspective of each ranking.
For instance, a score function may only take into account the top choice of each ranking, and the quality of a subset is then the number of rankings whose top choice it contains (this score function is referred to as single non-transferable vote, or SNTV).
More generally, a score function may take into account the entire ranking.
For instance, each ranking may add $m-1$ points to the quality of a subset if its top choice is present in it, another $m-2$ points if its second most preferred choice is present, and so on (the resulting score function is called the Borda count). 
While SNTV and Borda lead to modular set functions, more general score functions from the multiwinner voting literature give rise to general submodular functions (\Cref{def:score}), and a multitude of techniques have been developed to optimize the quality of the selection with respect to the given rankings \cite{DBLP:journals/scw/ProcacciaRZ08,DBLP:journals/ai/SkowronFS15,DBLP:conf/atal/AzizGGMMW15}.

However, there is growing evidence that rankings of items -- whether generated by humans or ML algorithms -- may contain implicit and/or other forms of systemic biases against socially-salient groups (e.g., those defined by race, gender, opinions) \cite{rooth2010automatic,kite2016psychology,regner2019committees}.
For instance, in the context of subset selection,  \citet{capers2017implicit} detected that, despite identical qualifications, members of admission committees rank White medical students above African Americans, and \citet{moss2012science} found that faculty members evaluate male applicants for a job as more qualified than female ones. 
The latter bias has also been observed in Amazon's former resume screening tool \cite{amazon}, as human biases can find their way into the output of machine learning algorithms in case they are trained on biased or systematically skewed data \cite{amazon,upturn_help_wanted,dressel2018accuracy,Raghavan2020Hiring,DBLP:conf/aistats/JiangN20}.
A question arises: \blue{in which situations can bias be mitigated in subset selection based on aggregating  multiple (biased) rankings using multiwinner score functions, and in these situations, how can the selection of a high-quality committee with respect to the unbiased rankings be guaranteed?}

A  model to incorporate such biases for the simplest (numerical) formulation of subset selection was considered by
\citet{kleinberg2018selection}:
Items are partitioned into an advantaged and a disadvantaged group and item $i$ has latent utility $W_i$. 
For members of the advantaged group, the observed utility is the same as the latent utility, whereas, for members of the disadvantaged group, the observed utility is $\theta\cdot W_i$ for some global bias parameter $\theta \in [0,1]$. 
\citet{kleinberg2018selection} show that when optimizing for the summed observed utility, the  latent utility of the selected subset drops sharply in the presence of implicit bias. 
These suboptimal selections as well as strong negative consequences for the affected groups have also been observed in practice \cite{rooth2010automatic,kang2011implicit,chapman2013physicians,kite2016psychology,greenwald2020implicit}.

To mitigate suboptimal selections in the presence of bias, one popular intervention is the use of  representational constraints, which are a type of group fairness constraint requiring that a certain fraction of the selected items need to be part of the disadvantaged group  (used, e.g., in the form of the Rooney Rule in shortlisting job applicants \cite{collins2007tackling,nfl}, or in primary elections \cite{DBLP:conf/eaamo/EvequozRKC22}).
\citet{kleinberg2018selection} studied the effectiveness of such  representational constraints in their bias model and  demonstrated that, while for worst-case $W_i$'s these constraints may not improve the latent utility of the selection, under certain conditions, e.g., assuming that the $W_i$'s are all drawn from the same distribution, representational constraints can increase the expected latent quality of the returned selection.
Subsequent works have confirmed the power of representational constraints to reduce the effect of bias in more involved settings such as online subset selection \cite{DBLP:conf/wine/SalemG20}, with intersectional groups \cite{mehrotra2022selection}, producing a ranking of items \cite{celis2020interventions}, and subset selection if biases come in the form of evaluation uncertainty  \cite{DBLP:journals/ai/EmelianovGGL22,DBLP:conf/sigecom/EmelianovGL22}.

Representational constraints have also already been considered in the context of subset selection based on score functions from multiwinner voting, albeit with a focus on the computational complexity of selecting subsets maximizing some goal function while respecting such constraints \cite{DBLP:conf/aaai/BredereckFILS18,DBLP:conf/ijcai/CelisHV18}.
In contrast to these works, following the work of \citet{kleinberg2018selection},  we analyze whether and when representational constraints can debias subset selection based on multiple input rankings and help to select high-quality subsets.
Notably, previous results on selection in the presence of bias do not apply to our setting, as we assume that the input comes from multiple sources, in the form of rankings (and not numerical values of items), which are aggregated using multiwinner score functions.

\paragraph{Our Contributions.}
We study the power of representational constraints for the following problem: 
Given $n$ potentially biased rankings over $m$ items, an integer $k$, and a score function $F$, select a size-$k$ subset of the items with a high score according to $F$ with respect to the latent rankings. 
As in the work of \citet{kleinberg2018selection}, without  distributional assumptions on the rankings and an explicit model of bias (that quantifies the relation between latent and observed rankings), no algorithm can guarantee that the returned subset has even a small fraction of optimum latent quality even with representational constraints  (see, e.g., \Cref{sec:smooth}). 
We consider a model where each latent ranking in the input is generated i.i.d. from a given distribution and then bias is introduced into each of these rankings in a ``controlled'' manner, inspired by the work of \citet{kleinberg2018selection}; our model is more general and is presented in Section \ref{sec:model}.  
We then give an algorithm and prove that there is a representational constraint such that, if the number of input rankings is large enough, the algorithm outputs an  ``approximately'' optimal solution with respect to the latent rankings; see \Cref{thm:main_algorithmic}.
We note that while the algorithm takes as input biased rankings and aims to find the subset with the maximum (biased) score subject to the representational constraint, its guarantee is with respect to the (unconstrained) optimal solution when the input rankings are unbiased.
Importantly, the {\em number of input rankings} the above result requires to hold depends on the chosen multiwinner score function.
To mathematically capture these differences between (submodular) multiwinner score functions, we introduce a new notion of ``smoothness'' (\Cref{def:smooth}), which quantifies how well the score function can ``correctly distinguish'' the strength of candidates among the same group under the latent and biased preferences.
A particular challenge here is that by considering submodular instead of modular functions, assessing the strength of a candidate becomes more difficult as a candidate's ``quality'' depends on the other candidates present in the subset, and that which of two candidates has a higher marginal contribution to some subset might change by  biasing the rankings.  
We complement our algorithmic result with an impossibility result that establishes the need for a large number of input rankings for certain multiwinner score functions; see \Cref{thm:main_impossibility}.

As a simple corollary of our results, we show a stark contrast between the two popular multiwinner score functions SNTV and Borda: Under a utility-based model generalizing the model of \citet{kleinberg2018selection} falling into our general class of generative models, for SNTV an exponential number of input rankings (in the number of items) is needed for representational constraints to recover a close-to-optimal solution, whereas for Borda the dependency is only polynomial (\cref{cor:algorithmic,thm:main_impossibility}). 
This difference is also intuitive, as under SNTV only a local snapshot of each ranking, i.e., who is ranked in the first position, is taken into account making fine-grained distinctions of candidates' quality harder, whereas for Borda the full ranking matters.

In summary, we contribute to the growing literature on designing algorithms in the presence of bias by showing the effectiveness of representational constraints in ranking-based subset selection. 
Distinguishing the effectiveness of representational constraints for different multiwinner score functions based on a new notion of smoothness, we contribute to a better understanding of multiwinner score functions, thereby providing further guidance for the selection of such a rule in the desired context.
Allowing for a convenient extension of our theoretical results, we give code that, given a bias model and scoring function as input, can be used as a tool to study the smoothness of the score function and the effectiveness of representational constraints with respect to the given bias model via simulations (\cref{sec:simulations_synthetic}).

\section{Other Related Work}
\label{sec:related}

    {A recent work \cite{mehrotra2023submodular} studies a variant of subset selection where the goal is to select a size-$k$ subset that maximizes the value of a submodular function.
  As in the current work, \cite{mehrotra2023submodular} also study the effectiveness of representational constraints for mitigating the adverse effects of bias.
    They, however, consider a family of submodular functions that is relevant to recommendation and web search, which is different from score functions considered in this work (\cref{def:score}).}
    {Specifically, in their setting, each item (such as websites, products, movies, or candidates) has $m$ attributes (such as topics, product-category, genres, or skills).
    Items and attributes in their setting map to candidates and voters respectively in the context of this work.
    Accordingly, we refer to items and attributes as candidates and voters respectively in the following.
    \cite{mehrotra2023submodular} specify a submodular function $F\colon 2^C\to \R_{\geq 0}$ by $n$ non-decreasing functions $g_1,g_2,\dots,g_n\colon \R\to \R_{\geq 0}$ (which measure the utility for each voter) and an $n\times m$ utility matrix $w$ (capturing the utility of candidates to each voter) as follows: $F(S) = g_1\inparen{\sum_{c\in S}W_{1,c}} + g_2\inparen{\sum_{c\in S}W_{2,c}} + \dots + g_n\inparen{\sum_{c\in S}W_{n,c}}$.}

    {The primary difference between the two families is that in their family each voter provides us with a cardinal utility for each candidate (scores $W_{1,c},W_{2,c},\dots,W_{n,c}$ for each candidate $c$) whereas we consider ordinal utilities specified by $n$ rankings of the $m$ candidates.
    Numerical scores are more accurate but can lead to serious aggregation and calibration issues \cite{griffin2004perspectives,DBLP:conf/allerton/MitliagkasGCV11,steck2018calibrated,DBLP:conf/atal/WangS19}.
    Moreover, while numerical scores are available in recommendation and web-search contexts, in contexts most relevant to this work, such as elections, numerical scores have a high elicitation cost \cite{griffin2004perspectives,DBLP:conf/allerton/MitliagkasGCV11,DBLP:conf/atal/WangS19}. 
    Due to this difference, the family of submodular functions considered by \cite{mehrotra2023submodular} is incomparable to the multiwinner scoring functions considered in this work (\cref{def:score}).}
    \begin{itemize}[itemsep=1pt]
        \item {On the one hand, in \cite{mehrotra2023submodular}'s model, voters can have the same utility for two or more candidates and can also have an arbitrarily large difference between the utilities of two candidates.
        However, because the functions we consider are defined by strict rankings over candidates, voters cannot be indifferent between two candidates.
        Moreover, as preferences are captured by rankings, we do not capture the magnitude of the difference in the utilities of two candidates in, say, adjacent positions in a ranking.} 
        \item {On the other hand, since in \cite{mehrotra2023submodular}'s model, the utility of a committee $S$ to a  voter $v$ is a function of the sum $\sum_{i\in S} W_{v,c}$, it cannot capture other utility functions such as the $\max$ utility a voter derives for a selected candidate or the sum of utilities for the best $t$ candidates included for some $1<t<k$ which are required, e.g., to capture the $\ell_1$-CC rule and its extensions.}
    \end{itemize}
    {That said, the SNTV rule, the Borda rule, and other committee scoring rules (\cref{sec:example}), can be captured by both the family of submodular functions in \cite{mehrotra2023submodular} and \cref{def:score}. 
    Here, \cite{mehrotra2023submodular}'s and our results complement each other.
    The bounds established by  \cite{mehrotra2023submodular} on the number of voters that are needed for representational constraints to recover a close-to-optimal utility degrades with $\poly(n/k)$\footnote{Concretely, their positive result degrades with $O\inparen{n^2 k^{-1/4}}$} and, since the number of voters, $n$, is much higher than the number of selected candidates, $k$, in real-world election contexts, \cite{mehrotra2023submodular}'s results lead to vacuous bounds in these contexts. 
    In contrast, our main algorithmic result (\cref{thm:main_algorithmic}) degrades with $\frac{1}{\poly(n/k)}$ and, hence, it gives a meaningful bound in contexts where $n\gg k$ but is vacuous if $k\gg n$.}   
    
    \noindent The reason why we are able to obtain results that degrade as $\poly\inparen{\frac{1}{n}}$ in our setting is because we require preference lists to be i.i.d., because of which a \textit{fixed} representational constraint is sufficient to recover high latent utility in our setting.
    In contrast, \cite{mehrotra2023submodular} allow the utility matrix (that captures preferences in their model) to be non-i.i.d. and show that, because of this, no fixed representational constraint is sufficient to recover any constant fraction of the optimal latent utility.
    Instead, they give an algorithm that, given functions $g_1,\dots,g_n$ and biased or observed utilities, determines the relevant representational constraint.
    Such algorithms may be reasonable in certain contexts (e.g., recommendation and web search) where the selection does not require direct human feedback but is not suitable in contexts that require direct human feedback (e.g.,  elections) and where \mbox{representational constraints must be fixed before human feedback (e.g. votes) are collected.}

  %  \paragraph{Further related work on the study of representational constraints, learning preferences in the presence of noise, and other works of subset selection.}
    Beyond the study of the ability of representational constraints to increase the utility of selection, their effect on the decision-maker's bias over the long term has also been studied \cite{celis2021effect,hoda2021allocating}.
    Moreover, apart from representational constraints, the power of various other interventions to debias decisions based on biased inputs  has also been studied, e.g., by \citet{DBLP:journals/corr/abs-2004-10846,DBLP:conf/fat/GargLM21} in the context of school and college admission and by \citet{DBLP:conf/forc/BlumS20} in the context of classification.
    
    The problem of selecting a subset maximizing a multiwinner score function is a special case of submodular maximization.
    There is rich literature on submodular maximization.
    In this literature, optimization in the presence of cardinality constraints has been extensively studied \cite{fujishige2005submodular,krause2014submodular} and there is a standard $(1-\frac{1}{e})$-approximation algorithm for finding a size-$k$ subset maximizing a submodular function \cite{nemhauser1978analysis}.

    Closer to our work from a methodological perspective, there are many works that want to learn  user preferences, by e.g., fitting some parameterized generative model to the data \cite{DBLP:conf/icml/GuiverS09,DBLP:journals/jmlr/LuB14,DBLP:journals/jmlr/VitelliSCFA17,DBLP:conf/uai/AlloucheLY22}.
    Particularly closely related to the present paper are works on the sample complexity of such learning algorithms   \cite{DBLP:conf/nips/AwasthiBSV14,DBLP:conf/icml/Busa-FeketeHS14,DBLP:conf/ijcai/CaragiannisM17,DBLP:conf/focs/LiuM18,DBLP:conf/icml/CollasI21}, and the capabilities of different rules to recover a ground truth \cite{DBLP:conf/uai/ProcacciaRS12,DBLP:journals/aamas/CaragiannisKKK22}. 
    Notably, the problem that we study also captures some of these learning problems by interpreting latent rankings as the underlying ground truth and biased rankings as voters' noisy estimates of the ground truth. 
    Thereby, we recover settings studied, e.g., by \citet{DBLP:conf/uai/ProcacciaRS12} and \citet{DBLP:journals/aamas/CaragiannisKKK22}. %
    
    Finally, there is also a large body of work studying the generalization of the simplest (numerical) formulation of subset selection from outputting a subset to outputting a ranking \cite{INR-016,burges2010ranknet}.
    Within this body of works, the problem of biases in the output rankings as well as different types of interventions, including representational constraints, to mitigate these biases have been studied \cite{KayMM15,overviewFairRanking,fair_ranking_survey2,criticalReviewFairRanking22}.
    This problem is different from the variant of subset selection we study in two ways: (1) the input is a single numerical score for each item, and (2) the output is a ranking instead of a subset.

%%%%%%%%%%%%%%%%%%%%%%%%%%%%%%%%%%%%%%%%%%%%%%%%%%%%%%%%%%%%%%%%%%
\section{Models of Score Functions, Bias, and Representational Constraints}
\label{sec:model}

In this section, we formally introduce our setting and the studied problem. 
For the sake of consistency with the literature, we use terminologies from the multiwinner voting literature. 
We start by introducing the family of score functions in \Cref{sub:score-functions}, then our bias model in \Cref{sec:bias}, and lastly representational constraints in \Cref{sec:intervention}.  
All used notations are summarized in \Cref{tab:notation} in \cref{sec:proof}.

\subsection{A Family of Score Functions} \label{sub:score-functions}

Given a set $C$ of candidates (also called items), let $\prefs{C}$ be the set of all strict and complete orders over $C$. 
We refer to elements of $\prefs{C}$ as preference lists  (or rankings) and to subsets of $C$ as committees.
Given a preference list $\succ\ \in\prefs{C}$, we use $\pos_{\succ}(c)$ to denote the position of $c\in C$ in $\succ$ and use $\pos_{\succ}(S) \coloneqq (\pos_{\succ}(c))_{c\in S}$ to denote the vector of positions of each $c\in S$.

In the classic multiwinner voting problem, there is \begin{enumerate*}[label=(\roman*)] 
\item a set $C$ of $m$ candidates, 
\item a set $V$ of $n$ voters together with a  preference profile $R= \left\{\succ_v \in \prefs{C}: v\in V\right\}$ where $\succ_v$ is the preference list of voter $v$,
\item a function $\score: 2^{C}\rightarrow \R_{\geq 0}$ with an evaluation oracle that depends on $R$, and 
\item a desired committee size $k\in [m]$.
\end{enumerate*}
The goal is then to select a committee $S$ that maximizes $\score(S)$ subject to having size-$k$.
Our results apply to a general class of submodular functions $F$.
For its definition, 
we use the following notion to compare the quality of committees.
\begin{definition}[\bf{Position-wise dominance}]\label{def:dominate}
	Given $\succ\ \in \prefs{C}$ and two committees $S = \left\{c_1,\ldots,c_k\right\}$ and $S' = \left\{c'_1,\ldots,c'_k\right\}$ of equal size, satisfying  $\pos_{\succ}(c_1)  > \cdots > \pos_{\succ}(c_k)$ and $\pos_{\succ}(c'_1) > \cdots > \pos_{\succ}(c'_k)$,
	we say that \emph{$S$ position-wise dominates $S'$ with respect to $\succ$} if for every $\ell\in [k]$, $\pos_{\succ}(c_\ell) \leq \pos_{\succ}(c'_\ell)$.
\end{definition}
We consider the following general family of score functions.
\begin{definition}[\bf{Multiwinner score function}]
    \label{def:score}
    Given a set $C$ of candidates and the voters' preference profile $R= \left\{\succ_v \in \prefs{C}: v\in V\right\}$, we say a function $\score: 2^C\rightarrow \R_{\geq 0}$ is a \emph{multiwinner score function} if $F$ satisfies the following: %
    \begin{itemize} %
        \item (Separability) there is a function $f\colon 2^C\times \prefs{C} \to \R_{\geq 0}$ mapping committees and preference lists to scores such that for every $S\subseteq C$, $\score(S) = \sum_{v\in V} f(S, \succ_v)$. Here, $f_v(S)\coloneqq f(S, \succ_v) = f(\pos_{\succ_v}(S))$\footnote{We only consider functions $f$ of the positions $\pos_{\succ_v}(S)$ of $S$ in $\succ_v$, imposing that $f$ is symmetric with respect to  candidate indices. Such functions are known as neutral functions in social choice.} is the score voter $v$ awards to $S$ that depends on their preference list $\succ_v$;
        \item (Monotone submodular) for every voter $v\in V$, $f_v$
        is a monotone submodular function;\footnote{Here, ``monotone'' means that for any $S\subsetneq C$ and candidate $c\in C\setminus S$, $f_v(S\cup \{c\}) \geq f_v(S)$; ``submodular'' means that for any $S\subsetneq T\subsetneq C$ and candidate $c\in C\setminus T$, $f_v(T\cup \{c\}) - f_v(T) \leq f_v(S\cup \{c\}) - f_v(S)$.} 
        \item (Domination sensitive) for every $v\in V$ and committees $S,S'\subseteq C$ with $|S| = |S'|$ and set $T\subseteq C\setminus S\cup S'$, if $S$ position-wise dominates $S'$ with respect to  $\succ_v$, we have $f_v(S) - f_v(S')\geq f_v(S\cup T) - f_v(S'\cup T)\geq 0$.\footnote{Voters award a higher score to the better committee $S$ than to the worse committee $S'$. Adding the same candidates to $S$ and $S'$ decreases the difference between their awarded scores (as the two become ``more similar'').}
    \end{itemize}
\end{definition}%
Among others, \Cref{def:score} captures the following score functions $\score(S) = \sum_{v\in V} f(S, \succ_v)$ relevant to our work:\footnote{As discussed in \Cref{sec:example}, \Cref{def:score} also covers the larger classes of committee-scoring rules and approval-based rules.} 
	\begin{itemize} %
		\item If $f(S,\succ)=\sum_{c\in S} \mathbbm{1}_{\pos_{\succ}(c)=1}$, $\score$ is the SNTV rule. %
		\item  If $f(S,\succ)= \sum_{c\in S} \inparen{m-\pos_{\succ}(c)}$,  $\score$ is the Borda rule.  
    \item  If $f(S,\succ)= \max_{c\in S} \inbrace{m-\pos_{\succ}(c)}$,  $\score$ is {the $\ell_1$-Chamberlin-Courant rule (or $\ell_1$-CC~rule) introduced by \citet{chamberlin_courant_1983}.}
	\end{itemize}
We conclude by defining the marginal contribution of a candidate:
Given a function $f: 2^C\times \prefs{C} \rightarrow\R_{\geq 0}$, preference list $\succ\ \in \prefs{C}$, committee $S\subsetneq C$ and candidate $c\in C\setminus S$, we define $f_S(c,\succ) \coloneqq f(S\cup \{c\}, \succ) - f(S, \succ)$ to be the marginal contribution of $c$ to $S$ with respect to $f$ and $\succ$.

\subsection{Multiwinner Voting in the Presence of Bias}\label{sec:bias}

In this section, we define the general problem of subset selection based on multiple input rankings in the presence of implicit bias. 
Our general approach is in line with previous works on subset selection \cite{kleinberg2018selection,celis2020interventions}, and assumes that the candidates are partitioned into an advantaged group $G_1$ and a disadvantaged group $G_2$, where disadvantaged candidates $c\in G_2$ face systematic biases.
Henceforth, we denote by $\succ_v$ the ``true'' or \textit{latent} preference of voter $v$, and by $\nsucc_v$ their \textit{biased} or \textit{observed} preference. 
We consider the following problem.

\begin{problem}[\bf{Multiwinner voting in the presence of bias}]
    \label{def:observed_list}
    Let $R= \left\{\succ_v\in \prefs{C}\colon v\in V\right\}$ be the voter's latent preference lists and $\widehat{R} = \left\{\nsucc_v\in \prefs{C}\colon v\in V\right\}$ be the voter's biased or observed preference lists.
    Further, let $\score\colon 2^C\rightarrow \R_{\geq 0}$ be a multiwinner score function, and $k\geq 1$ be the desired committee size.
    The goal of the \emph{multiwinner voting problem in the presence of bias} is to select a size-$k$ committee $S$ that (approximately) maximizes $\score(S)$, while only taking $\widehat{R}$ as input (and not knowing $R$).
\end{problem}

\subsubsection{Models for Latent and Biased Preferences} \label{sec:Models-latent-biased}
We consider a general setting of how the latent and biased preferences of voters are generated. 
Our main assumption is that there are no differences between voters, implying that all voters sample their latent and biased preferences from the same distribution. 
This is a common assumption both in the line of work on selection in the presence of bias \cite{kleinberg2018selection,celis2020interventions,emelianov2020on} and in many mechanisms to sample preferences in social choice theory \cite{DBLP:conf/atal/SzufaFSST20}.
We first define our generative model for latent preferences:
\begin{definition}[\bf{Generative model for latent preferences}]
\label{def:latent_preference}
    {A generative model $\mu$ for latent preferences is defined by a distribution $\pi$ over $\prefs{C}$: under $\mu$, each voter $v\in V$ draws its latent preference list $\succ_v$ from $\pi$ independently.}
\end{definition}
Similarly, in our generative model for biased preferences, we only require that all voters with the same latent preferences sample their biased preferences from the same distribution. Formally, our generative model for biased preferences takes a latent preference list $\succ$ as input and draws a biased preference list $\nsucc$ from a certain conditional distribution on $\succ$.
For instance, this would allow for swapping down a candidate from the disadvantaged group  uniformly at random in the given, latent preferences.

\begin{definition}[\bf{Generative model for biased preferences}]
\label{def:biased_preference}
    {A generative model $\wh{\mu}$ for biased preferences is defined by a distribution $\wh{\pi}$ over $\prefs{C}\times \prefs{C}$: under $\wh{\mu}$, each voter $v\in V$ with latent preference $\succ_v$ draws its biased preference list $\nsucc_v$ from $\widehat{\pi}\mid \succ_v$ independently.}\footnote{\cref{def:biased_preference} can be extended to condition on additional information beyond $\succ_v$. For instance, \Cref{def:utility_bias}, considers the conditional distribution with respect to a latent utility vector $w_v$ (which specifies $\succ_v$) instead of $\succ_v$.}
\end{definition}
    {Note that our generative model for biased preferences also allows that some voters are biased and others unbiased (or even that different voters have opposing biases). 
    For instance, $\wh{\mu}$ could be defined by a mixture of distributions, where each voter $v$ draws $\nsucc_v$ according to $\wh{\pi}_1\mid \succ_v$ with probability $\frac{1}{2}$ and according to $\wh{\pi}_2\mid \succ_v$ with probability $\frac{1}{2}$. 
    In this case, roughly half of the voters would generate their biased preferences according to $\wh{\pi}_1$, and the remaining voters would generate them using $\wh{\pi}_2$ and, hence, display no bias (or a potentially opposite) bias than the other voters.}

\subsubsection{Utility-Based Model}
\label{sec:utility}

In this section, we give a specific example of what a generative model for latent and biased preferences could look like. 
For this, we present a natural adaption of the bias model considered by \citet{kleinberg2018selection} to our setting. 
In the resulting utility-based model, we assume that each voter $v\in V$ has a non-zero latent utility $w_{v,c}$ (respectively biased utility $\wh{w}_{v,c}$) for each candidate $c\in C$,  and that in their latent (respectively biased) preferences voters rank candidates sorted in the non-increasing order of $w_{v,c}$ (respectively $\wh{w}_{v,c}$) with breaking ties arbitrarily. 
We start by explaining how the latent utilities are generated:
\begin{definition}[\bf{Utility-based latent generative model}]
	\label{ex:utility_generative} 
	Every candidate $c\in C$ has an intrinsic utility $\omega_c \geq 0$.
	For each voter $v\in V$ and candidate $c\in C$, the utility $w_{v,c}$ {is generated independently} by $w_{v,c} = \eta \cdot\omega_c$, where 
	$\eta$ {is a random variable drawn uniformly from} $[0,1]$.\footnote{More generally, one can consider other distributions and, we do so, in \cref{def:desired_properties_of_eta,def:desired_properties_of_eta2,def:eta_for_alpha}.}
\end{definition}
As in the work of \citet{kleinberg2018selection}, applying bias then boils down to scaling down the utilities of disadvantaged candidates.
\begin{definition}[\bf{Utility-based biased generative model}]
	\label{def:utility_bias}
	{Given a bias parameter $\theta\in [0,1]$,} for every $v\in V$, $\widehat{w}_{v,c} = w_{v,c}$ for all $c\in G_1$, and $\widehat{w}_{v,c} = \theta \cdot w_{v,c}$ for all $c\in G_2$. 
\end{definition}
The definition of the utility-based bias model implies that the intrinsic utilities of all disadvantaged candidates reduce by a factor of $\theta$. 
However, note that this seemingly uniform reduction can have different effects on different candidates:
As an extreme example suppose that $G_1 = \inbrace{c_3,c_4,\dots,c_m}$ and $G_2 = \inbrace{c_1,c_2}$ with intrinsic utility values $\omega_1 = 10$ and $\omega_2 = \omega_3 = \cdots = \omega_m = 1$. 
For any $\theta$, the probability that $c_1$'s position in $\nsucc$ is worse than their position in $\succ$ is $1 - \Theta(\theta)$. 
Whereas the probability that the other disadvantaged candidate $c_2$ is placed at a worse position in $\nsucc$ compared to $\succ$ is with $1 - \Theta(\theta^m)$ much higher.

\subsubsection{Order Preservation}\label{sub:general_model}
To analyze the capabilities of representational constraints to mitigate bias, different ways in which the used multiwinner score function interacts with the used generative distributions for latent and biased preferences are relevant. 
Capturing this interplay, in this section we introduce two order-preservation properties of multiwinner score functions. %
{Both of these properties hold for the utility-based model, but they may not hold for all generative models (\cref{def:latent_preference,def:biased_preference}).
In the remainder of this section and in \cref{sec:smooth,sec:main_algorithmic}, we present our results for general generative models. 
In \cref{sec:applications_of_result} we describe the implications of our general results for the utility-based model and in \cref{sec:case_study}, we describe simplified variants of our general arguments and properties tailored to the utility-based model.}

In order to be able to compute a committee with close-to-optimal utility in our algorithmic analysis, it will be crucial that if the number $n$ of voters is large enough, then the optimal committee consists of the $k$ candidates with the highest expected individual scores $\Exp_{\mu}\left[f(c,\succ) \right]$ for $c\in C$, as otherwise even if we would have access to the latent generative model computing an optimal committee might be computationally intractable.
To ensure this, it will turn out that it is sufficient to assume that each candidate has an intrinsic quality (i.e., $\Ex_\mu\insquare{f(c, \succ)}$) and that this quality determines the ordering of candidates by relative marginal contributions to committees.
More formally, we require that if $c$ has a higher intrinsic quality than $c'$, then $c$'s marginal contribution to a committee $S$ (for any $S$) is higher than that of $c'$ (Item 1 in \cref{def:latent}).
However, this still does not restrict how the difference between the intrinsic quality of candidates influences the difference between their marginal contributions to some committee. 
To address this, we also require that the difference in the marginal contributions of $c$ and $c'$ to committees is a non-increasing function with respect to adding candidates to the committee (Item 2 in \cref{def:latent}).

\begin{definition}[\bf{Order-preservation with respect to latent preferences}]
    \label{def:latent}
    Given a generative model $\mu$ for latent preference and a multiwinner score function $F = \sum_{v\in V} f(\cdot ,\succ_v)$, we say \emph{$F$ is order-preserving with respect to $\mu$} if for any two candidates $c \neq c'\in C$ belonging to the same group ($G_1$ or $G_2$) \textbf{if} $\Exp_{\mu}\left[f(c,\succ) \right] \leq \Exp_{\mu}\left[f(c',\succ)\right]$, \textbf{then} for any subsets $T\subseteq S\subseteq C\setminus \left\{c,c'\right\}$
        \begin{enumerate} %
            \item { $\Exp_{\mu}\left[f_S(c, \succ)\right] \leq \Exp_{\mu}\left[f_S(c', \succ)\right];$}
            \item $\Exp_\mu\insquare{f_S(c',\succ)} {-} \Ex_\mu\insquare{f_S(c,\succ)} \leq \Ex_\mu\insquare{f_T(c',\succ)}  {-} \Ex_\mu\insquare{f_T(c,\succ)}.$
        \end{enumerate} 
\end{definition}
We prove in \Cref{lem:utility_order_preserve} that any multiwinner score function $F$ is order-preserving with respect to the utility-based latent generative model (\Cref{ex:utility_generative}). 
The reason for this is that in the utility-based model the intrinsic quality of a candidate $c$ is an increasing function of $\omega_c$, as  $\omega_c\leq \omega_{c'}$ implies that $\pos(c',\succ) < \pos(c,\succ)$  with probability at least $\frac{1}{2}$.
The order preservation can then be shown using the domination sensitivity of $F$ (\cref{def:score}).

As an additional ingredient, we need a second (order-preservation) property controlling the relation between the latent $\mu$ and biased $\wh{\mu}$ generative models: 
Otherwise, if we would allow $\wh{\mu}$ to be arbitrarily different from $\mu$, then one cannot hope to identify a committee with high latent utility by just observing biased preferences (generated according to $\wh{\mu}$).
In particular, our second order-preservation property restricts by how much the ratio between the marginal contribution of two candidates $c$ and $c'$ to some set $S$ is allowed to change when bias is applied. 
\begin{definition}[\bf{Order-preservation between latent and biased preferences}]
    \label{def:bias}
    Given a generative model $\mu$ for latent preference lists and a generative model $\widehat{\mu}$ for biased preference lists, a multiwinner score function $F = \sum_{v\in V} f(\cdot ,\succ_v)$, and numbers $\beta, \gamma \in [0,1]$,
    we say $F$ is \emph{$(\beta,\gamma)$ order preserving between $\mu$ and $\wh{\mu}$} if for any two candidates $c \neq c'\in C$ belonging to the same group ($G_1$ or $G_2$) and any $S\subseteq C\setminus \{c,c'\}$, \textbf{if} $\beta \cdot \Exp_{\mu}\left[f_S(c',\succ)\right] \geq \Exp_{\mu}\left[f_S(c,\succ) \right] > 0$, \textbf{then} %
        $\gamma\cdot {\Exp_{\widehat{\mu}}\left[f_S(c', \nsucc)\right]}\geq   {\Exp_{\widehat{\mu}}\left[f_S(c, \nsucc)\right]}.$
\end{definition}
In this definition, $\beta$ specifies the range of candidates affected by the property and $1-\gamma$ is the allowed gap between their relative marginal contributions that can emerge by applying bias: 
When $\beta$ is close to 1 then we consider candidate pairs $c,c'$ and sets $S$ with a large range of $\frac{\Exp_{\mu}\left[f_S(c,\succ) \right]}{\Exp_{\mu}\left[f_S(c',\succ) \right]}$.
In contrast, as $\gamma$ goes to one,  expected marginal contributions  are allowed to become more and more similar even in case there is a substantial gap under $\mu$ (and, hence, candidates get harder to distinguish).

In case no bias is applied to the preferences, for each candidate $c$ and set $S$ the expected marginal contribution of $c$ to $S$ is the same in both $\mu$ and $\hat{\mu}$ (notably, this can also happen in other cases, for instance, if both the latent and biased preferences of a voter are independently drawn from the uniform distribution on $\prefs{C}$). 
If all expected marginal contributions remain unchanged, every multiwinner score function is $(\beta,\beta)$ order preserving between $\mu$ and $\wh{\mu}$ for any $\beta\in (0,1]$.
As shown in \Cref{lem:utility_order_preserve}, many multiwinner score functions (including the examples in \cref{sub:score-functions}) are $(\beta,\gamma)$ order preserving between the utility-based latent and biased generative model for any $0\leq \beta \leq 1-\lambda$ and $1-m^{-\Theta(1)}\cdot \lambda \leq \gamma\leq 1$ where $\lambda=m^{-\Theta(1)}$.
Intuitively, this is because we apply the same multiplicative bias to all members of the disadvantaged group.

{Recall that we have observed in \Cref{sec:Models-latent-biased} that  \cref{def:biased_preference} also allows the biased preference lists to be drawn from a mixture of two or more distributions. 
In the case of mixtures of utility-based models, certain order-preservation properties translate from the individual components to the mixture.  
Concretely, assume that $(\mu,\wh{\mu}_1)$ and $(\mu,\wh{\mu}_2)$ are two utility-based generative models with the same intrinsic values $\omega$ but heterogeneous bias parameters $\theta_1$ and $\theta_2$ respectively. Let us consider mixtures of $(\mu,\wh{\mu}_1)$ and $(\mu,\wh{\mu}_2)$. 
Using linearity of expectation, it can be shown that if $F$ is $(\beta_i,\gamma_i)$ order-preserving between $\mu$ and $\wh{\mu}_i$ for each $i{\in} \inbrace{1,2}$, then $F$ is $(\min\inbrace{\beta_1,\beta_2},\max\inbrace{\gamma_1,\gamma_2})$ order-preserving between $\mu$ and $\delta\wh{\mu}_1{+}(1{-}\delta)\wh{\mu}_2$ for any $\delta \in (0,1)$.}

\paragraph{Swapping-based bias model.}
In \Cref{sec:swapping-based}, we present a swapping-based bias model inspired by the popular Mallows model \cite{mallows1957non}, which is also captured by the general bias model. 
Given a latent preference list $\succ\ \in \prefs{C}$, for $t\geq 1$ iterations, this model  selects two candidates $c\in G_1$ and $c'\in G_2$ with $c'$ being ranked above $c$ and swaps their position, where the average difference in the position of swapped candidates can be controlled by a parameter $\phi\in [0,1]$.
We prove in \cref{lem:swap_order_preservation} that for $\lambda=\Theta(t\phi)$,
all multiwinner score functions are $(1-\lambda,1-\frac{\lambda}{2})$ order preserving between the utility-based latent generative model $\mu$ and swapping-based biased generative model $\wh{\mu}$.
Intuitively, $\beta$ is bounded away from $1$ because, unlike in the utility-based bias model, in the swapping-based model bias acts non-uniformly across members of the disadvantaged groups (depending on their positions relative to advantaged candidates in voters' preference lists).

\subsection{Representational Constraints}\label{sec:intervention}

Let $S^\star \coloneqq \arg\max_{S\subseteq C: |S|=k} \score(S)$ denote the optimal solution of the underlying multiwinner voting problem, and let $\OPT \coloneqq \score(S^\star)$.
Furthermore, let $M\subseteq C$ denote the collection of $k$ candidates $c$ with the highest value $\Exp_{\mu}\left[f(c,\succ)\right]$ (as discussed above if $F$ satisfies order-preservation with respect to the latent generative model $\mu$ and $n$ is high enough, $S^\star$ is likely to be equal to $M$).

A natural approach to solve the multiwinner voting problem in the presence of bias is to directly solve $\widehat{S}\coloneqq\arg \max_{S\subseteq C\colon |S|=C} \widehat{\score}(S)$.
However, due to biases, this may lead to committees with poor latent quality. %
In fact, we argue in \Cref{remark:rl} that one can lose significant value by not selecting candidates from $G_2$, i.e.,  $\score(\widehat{S}) \leq (1-\frac{|G_2\cap S^\star|}{2k})\cdot \OPT$.
To mitigate such adverse effects of bias, a popular intervention are \emph{representational constraints} that require the output subset to include at least a specified number of candidates from the disadvantaged group.\footnote{Note that enforcing representational constraints requires knowledge of the sets of advantaged and disadvantaged candidates. While they may not be available or costly to obtain in some contexts, e.g., web-search (see \cite{mehrotra2022noisy} and the references therein), they are available in relevant contexts such as elections when groups are based on (combinations of) socially-salient attributes (such as race, gender, age, and disability) \cite{evequoz2022equity}.}
\begin{definition}[\bf{Representational constraints}]
    \label{def:representational}
    {Given integer $0\leq \ell\leq k$,} the $\ell$-representational constraint requires $
    |S\cap G_2|\geq \ell$ for the selected committee $S\subseteq C$.
    Let $\calK(\ell)$ denote the collection of subsets that satisfy the $\ell$-representational constraint.
\end{definition} 
In the presence of the $\ell$-representational constraint, the straightforward output subset, say $\hS_\ell$, is a solution to the following optimization problem:$\max_{S\in \calK(\ell): |S| = k} \widehat{\score}(S).$
This paper centers around the following problem:
\begin{problem}[\bf{Effectiveness of representational constraints}]    \label{problem:representational_optimal}
    Is there an integer $0\leq \ell\leq k$ such that $\score(\widehat{S}_\ell)\approx \OPT$ (with high probability), where the inputs are as specified in \Cref{def:observed_list}?
    Moreover, is there a polynomial-time algorithm that outputs a set $S\in K(\ell)$ with $F(S)\approx \OPT$?
\end{problem}

\section{Algorithmic Results for \Cref{problem:representational_optimal}}
\label{sec:algorithmic}  

In this section, we show that representational constraints can help mitigate the adverse effects of bias in ranking-based subset selection with multiwinner score functions. %
Specifically, in \Cref{thm:main_algorithmic}, we provide a sufficient condition and an efficient algorithm for \Cref{problem:representational_optimal}.
Our algorithmic result is based on a new notion of smoothness (\Cref{def:smooth}), which distinguishes the capabilities of multiwinner score functions to be debiased by representational constraints.
Afterward, in \Cref{sec:applications_of_result}, we discuss specific implications of \Cref{thm:main_algorithmic}, among others, highlighting that under the utility-based model the sufficient number of voters to recover a close-to-optimal utility for SNTV or Chamberlin-Courant is exponential in $m$, whereas this dependence is only polynomial for Borda.

\subsection{Smoothness} \label{sec:smooth}

In a preliminary theoretical and empirical analysis, we observed that SNTV requires a significantly higher number of voters to recover the same solution quality as Borda when using representational constraints. 
We provide a theoretical demonstration of this contrast in \Cref{cor:algorithmic}. 
To quantitatively measure such differences between functions, we introduce the following notion of smoothness of a multiwinner score function $F$, which quantifies the ability of representational constraints to recover a latent utility that is close to optimal under $F$. %

To build some intuition, let us focus on a case where $F$ is order-preserving with respect to $\mu$ and $(1,0.01)$ order-preserving between $(\mu,\hat{\mu})$. 
Then, the set $M$ of candidates with the highest expected value  $\Exp_{\mu}\left[f(c,\succ)\right]$ is likely to be an optimal solution and those candidates are also likely to have the highest expected value $\Exp_{\widehat{\mu}}\left[f(c,\succ)\right]$. 
In this case, our algorithm only needs to identify all and, in particular, the weakest candidate $c$ from $M$ using samples in $\hR=\inbrace{\nsucc_v\in \mathcal{L}(C)\colon v\in V}$.
The marginal contribution of this candidate can be as low as $\Ex_{{\widehat{\mu}}}\insquare{f_{C\setminus \{c\}}(c,\nsucc)}$. 
Let $\tau_1(f)$ be the maximum possible expected score $\Exp_{\mu}\left[f(c,\succ)\right]$ of a candidate. 
To capture how many samples are needed to identify the weakest candidate $c$ from $M$, we introduce a new parameter $\alpha$ which needs to be at least $\alpha \geq \frac{1}{\tau_1(f)}\min_{c\in M} \Ex_{\wh{\mu}}\insquare{f_{C\setminus \{c\}}(c,\not\succ)}$ (note that $\tau_1(f)$ acts as a normalization). 
When $\alpha$ is close to 0, it is ``more difficult'' to distinguish candidates in $M$ from candidates in $C \setminus M$ since margin values of some candidates from $M$ are more likely to be close to 0, and hence, a larger number of voters is required to distinguish them.

This observation can be related to our initial finding regarding SNTV and Borda. 
SNTV requires more voters to identify the strength of a candidate, as only the top choice of a voter is taken into account. 
For instance, in the presence of a strong candidate $c^\star$, who is ranked first with a high probability, it may take many samples to observe the first vote where 
$c^\star$ is not ranked first. 
In contrast, for Borda, candidate strength can be distinguished much more easily, as every sampled vote provides new information on all candidates.

Note that in a different direction, we also need a sufficient number of voters to ensure that candidates from $M$ are likely to constitute an optimal solution.
For this, we again need that the sampled votes are enough so that candidates from $M$ have a higher contribution than candidates from $C\setminus M$ with respect to the latent votes; accordingly, similar to the above we also require that $\alpha\geq \frac{1}{\tau_1(f)}\min_{c\in M} \Ex_{{\mu}}\insquare{f_{C\setminus \{c\}}(c,\succ)}$.

Combining these observations with the order-preservation properties, we propose the following definition of smoothness. 
On an intuitive level, our smoothness definition boils down to how well the multiwinner score function $F$ can ``correctly distinguish'' the strength of candidates among the same group under the latent and biased preferences.

\begin{definition}[\bf{Smoothness}]
    \label{def:smooth}
        Let $F = \sum_{v\in V} f(\cdot ,\succ_v)$ be a multiwinner score function.
        Let $(\mu,\widehat{\mu})$ be a generative model defined in \Cref{def:bias,def:latent}.
        Given $\alpha,\beta, \gamma \in [0,1]$, we say $F$ is $(\alpha, \beta, \gamma)$-smooth with respect to $(\mu,\wh{\mu})$ if the following holds
        \begin{itemize} %
            \item $\alpha \geq \frac{1}{\tau_1(f)}\min_{c\in M} \Ex_{\wh{\mu}}\insquare{f_{C\setminus \{c\}}(c,\not\succ)}$ and $\alpha \geq \frac{1}{\tau_1(f)}\min_{c\in M} \Ex_{{\mu}}\insquare{f_{C\setminus \{c\}}(c,\succ)}$,
            \item $F$ is order-preserving with respect to $\mu$; and 
            \item $F$ is $(\beta, \gamma)$ order preserving between $\mu$ and $\widehat{\mu}$.
        \end{itemize}
\end{definition}
We have already argued before that $\alpha$ influences the number of voters required to identify a close-to-optimal committee. 
The parameter $\gamma$ also influences the number of voters needed but in this case, the larger $\gamma$ gets the more votes are needed:
If $\gamma$ is close to 1, then two candidates $c,c'$ with $\Exp_{\mu}\left[f(c, \succ)\right]\ll \Exp_{\mu}\left[f(c', \succ)\right]$, may have a similar quality in the presence of bias, $\Exp_{\wh{\mu}}\left[f(c, \nsucc)\right]\approx \Exp_{\wh{\mu}}\left[f(c', \nsucc)\right]$, making it harder to identify the stronger candidate needed for a close-to-optimal solution. 
In contrast to $\alpha$ and $\gamma$, factor $\beta$ bounds how close to the optimum one can get when observing biased voters: 
Intuitively, a value of  $\beta$ close to 1 implies that for two candidates the order of their marginal contributions remains unchanged when applying the bias. 
In contrast to this, for $\beta<1$, for two candidates $c,c'$ with $\Exp_{\mu}\left[f_S(c', \succ)\right]>\Exp_{\mu}\left[f_S(c, \succ)\right]$ and $\frac{\Exp_{\mu}\left[f_S(c, \succ)\right]}{\Exp_{\mu}\left[f_S(c', \succ)\right]}\geq \beta$ for some set $S$, we allow the ratio of their marginal contributions to $S$ to change arbitrarily in the presence of bias. 
Consequently, for such pairs of candidates $c$ and $c'$, even if the number of observed biased voters goes to infinity, we will  not be able to distinguish which of the two is the stronger candidate under the latent distribution, leading to a potential multiplicative loss of $\beta$ in the latent quality of the output solution due to wrongfully {including $c$ instead of $c'$ in the returned solution.}

\subsection{Main Theorem}
\label{sec:main_algorithmic}

As discussed above, for an $(\alpha, \beta, \gamma)$-smooth function, the values of $\alpha$ and $\gamma$ determine the number of samples required for representational constraints to return an approximately optimal solution, while the value of $\beta$ bounds the achievable multiplicative approximation factor.
Our main algorithmic result matches these intuitions, and we provide a sufficient condition on $n$ under which representational constraints recover a solution $S$ that is close to optimal. 
The proof of this result appears in \cref{sec:proofof:thm:main_algorithmic}.

\begin{theorem}[\bf{Algorithmic result for \Cref{problem:representational_optimal}}]
    \label{thm:main_algorithmic}
        Let $F\colon 2^C\to \R_{\geq 0}$ be a multiwinner score function.
        Let $\mu$ and $\widehat{\mu}$ be generative models of latent and biased preference lists respectively. 
        Suppose $F$ is $(\alpha, \beta, \gamma)$-smooth with respect to $(\mu,\wh{\mu})$ for some $\alpha, \beta, \gamma\in [0,1]$.
        For any $0<\eps,\delta<1$,
        if 
        \[  n \geq \frac{16k}{\inparen{\alpha\min\inbrace{\eps,1-\gamma}}^{2}} \cdot  \log{\frac{m}{\delta}},\]
        there is an algorithm that given groups $G_1, G_2$, numbers $k$, $\ell=\abs{M\cap G_2}$, and a value oracle $\calO$ to $\hF(\cdot)$ as input, outputs a subset $S\in \calK(\ell)$ of size $k$ such that 
        \[ \Pr_{\mu,\wh{\mu}}\insquare{F(S) \geq \inparen{\beta-\eps} \cdot  \OPT} \geq 1-\delta.\]
        The algorithm makes $O(mk)$ calls to oracle $\calO$ and performs $O\inparen{m\log{m}}$ arithmetic operations.
\end{theorem}
The algorithm underlying \cref{thm:main_algorithmic} is a standard greedy algorithm (\Cref{algo:main_algorithm} in \Cref{sec:proofof:thm:main_algorithmic}) that maximizes $\widehat{F}$ subjected to representational constraint $\ell=\abs{M\cap G_2}$.
The key idea used in the proof of \cref{thm:main_algorithmic} is that, due to the smoothness of $F$, when \Cref{algo:main_algorithm} adds the $i$th candidate to the committee, the incurred marginal contribution with respect to the latent preferences is at least  a $\beta$ fraction compared to when building $M$ and adding the $i$th highest scoring candidate to a set of the $(i-1)$ highest scoring candidates (\Cref{lem:lower_bound_on_marginal}). 
Note that due to the greedy nature of our algorithm, the output solution $S$ may not be identical to $\widehat{S}_{\ell}$ for some $\ell$ (recall $\widehat{S}_\ell=\arg \max_{S\in \cK(\ell): |S|=k} \widehat{\score}(S)$).
However, for modular score functions such as SNTV and Borda, the algorithm always outputs $S = \widehat{S}_{\abs{M\cap G_2}}$.
Hence, for some multiwinner score functions, \Cref{thm:main_algorithmic} also implies $F(\widehat{S}_{\ell}) \geq \inparen{\beta-\eps} \cdot  \OPT$ for $\ell=\abs{M\cap G_2}$, which partially addresses the first question of \Cref{problem:representational_optimal}.

Note that the value $\abs{M\cap G_2}$ is unknown in advance.
While in general, this value depends on $F$ and the generative models $(\mu,\wh{\mu})$, there are also natural special cases where it is independent. 
For instance, if we assume that preference lists drawn from $\mu$ are not systematically skewed toward candidates in either group, as may be the case in the real world \cite{evequoz2022equity}, then $\abs{M\cap G_2} \approx k\cdot\frac{\abs{G_2}}{m}$ with probability $1-o_k(1)$.
{
Moreover, in applications such as recommendation systems that use multiwinner scoring functions \cite{DBLP:conf/prima/StreviniotisC22}, $\ell$ may be tuned via A/B testing by trying different $\ell$, estimating latent quality from user engagement, and selecting the value of $\ell$ that maximizes the latent quality.}

    {It is worth noting that although} $\alpha$ has a similar form as the curvature of submodular functions, there are some  differences between them {that render} the curvature ineffective in measuring the effectiveness of representational constraints.
    For instance, the curvature is unable to distinguish modular functions.
    We discuss this and the relevance of \cref{def:smooth} to research on the effect of noise on multiwinner voting in \cref{sec:additional_remarks}. %

    \begin{table}[b!]
        \centering 
         %
         %\vspace{-0.1in}
        \begin{tabular}{lccc}
            \toprule 
            \textbf{Multiwinner scoring function} & $\alpha$ (\cref{lem:utility_alpha}) & $\beta$ (\cref{lem:utility_order_preserve2}) & $\gamma$ (\cref{lem:utility_order_preserve2})\\
            \midrule{}\hspace{-1mm}
            SNTV & $\Theta\inparen{\theta^{-2(m-1)}}$ & $1-\Theta(m^{-\sfrac12})$ & $1-\Theta(m^{-\sfrac32})$  \\
            $\ell_1$-CC & $\Theta\inparen{\theta^{-2(m-1)}}$ & $1-\Theta(m^{-\sfrac12})$ & $1-\Theta(m^{-\sfrac32})$ \\
            Borda & $\Theta\inparen{\theta^{-2}}$ & $1-\Theta(m^{-\sfrac12})$ & $1-\Theta(m^{-\sfrac52})$ \\
            \bottomrule{}
        \end{tabular} 
        %\vspace{-0.2in}
        \caption{
            Smoothness parameters for the utility-based model \cref{def:utility_bias}.
            The formal statements of the results appear as \cref{lem:utility_alpha,lem:utility_order_preserve2}.
            Note that these results hold for the utility-based model in \cref{def:utility_bias} as well as its variants where $\eta$ is not uniformly distributed on $[0,1]$ but is instead drawn from, say, an exponential distribution.
        }
        \label{tab:smoothness_results1}
    \end{table}

  \paragraph{Proof overview of \cref{thm:main_algorithmic}.}       
        {The smoothness condition is defined with respect to expectations over the generative models but in \cref{thm:main_algorithmic} we have only access to $n$ samples.
        The proof's first component is a concentration inequality showing that expectations over samples are close to the true expectations: 
            for any candidate $c\in C$ and committee $T\subseteq C$ ``of interest,'' $\Ex_\mu\sinsquare{f_T(c,\succ)} \approx \frac{1}{n} F_T(c)$ and $\Ex_{\wh{\mu}}\sinsquare{f_T(c,\nsucc)} \approx \frac{1}{n} \hF_T(c)$ (\cref{lem:concentration}), where $F_T(c)\coloneqq F(T\cup \inbrace{c})-F(T)$ and $\hF_T(x)\coloneqq \hF(T\cup \inbrace{c})-\hF(T)$.
        Let $S \in \cK(\ell)$ be the subset output by the algorithm.
        Recall that $M \subseteq C$ is the set of $k$ candidates $c$ with the highest value $\Exp_{\mu}\left[f(c,\succ)\right]$ and $\OPT {\coloneqq} \max_{S\subseteq C: |S|=k} \score(S)$.
        The proof strategy is to show that $F(M) \approx \OPT$ (\cref{lem:score_m}) and to compare the utility of ``prefixes'' of $S$ to ``prefixes'' of $M$ (\cref{lem:lower_bound_on_marginal}).
        For this, for large enough $n$, we show that our algorithm has the following property (\cref{eq:lb_ratio_of_biased_scores}):    
            there is an ordering $m_1,\dots,m_k$ of elements in $M$ such that
            \[
                \forall_{i \in [\ell]},\ \  
                \hF_{\inbrace{s_1,\dots,s_{i-1}}}(s_i)\geq \gamma \cdot \hF_{\inbrace{s_1,\dots,s_{i-1}}}(m_i),
                \yesnum\label{eq:invariant}
            \] 
            where $s_j$ is the $j$-th item added to $S$ for any $j$.
        If $F$ is modular, the remainder of the proof is straightforward:
            Due to order preservation between $\mu$ and $\wh{\mu}$, \cref{eq:invariant} implies that $F_{\inbrace{s_1,\dots,s_{i-1}}}(s_i) {\geq} \beta \cdot F_{\inbrace{s_1,\dots,s_{i-1}}}(m_i)$ for each $i$, and, since $F$ is modular, $F(S){=}\sum_i F(s_i) \geq \beta \cdot  \sum_i F(m_i) {=} \beta \cdot F(M)$.
        When $F$ is not modular, we need to show that \cref{eq:invariant} implies that $F_{\inbrace{s_1,\dots,s_{i-1}}}(s_i)\geq \beta \cdot F_{\inbrace{m_1,\dots,m_{i-1}}}(m_i)$ (note the change in the base).
        We do so in \cref{lem:lower_bound_on_marginal,eq:lb_marginal_1,eq:lb_marginal_2} using order preservation with respect to $\mu$.}

\subsection{Applications of \cref{thm:main_algorithmic}}
\label{sec:applications_of_result}
    In this section, we focus on the utility-based generative models $(\mu,\widehat{\mu})$ (\cref{def:utility_bias}) and derive bounds for some specific multiwinner score functions; \blue{see \cref{tab:smoothness_results1} for a summary of all results on the utility-based model and \cref{tab:smoothness_results2} for a summary of all results on the swapping-based model.}
\Cref{lem:utility_order_preserve,lem:utility_order_preserve2} give bounds on $\beta$ and $\gamma$ for this case. 
    We provide the missing values of $\alpha$ for some multiwinner score functions and the resulting sufficient numbers of voters using \cref{thm:main_algorithmic} in the following result, whose proof appears in  \cref{sec:proofof:lem:utility_alpha}.

    \begin{table}[b!]
        \centering
         %
         %\vspace{-0.1in}
        \begin{tabular}{lccc}
            \toprule 
            \textbf{Multiwinner scoring function} & $\alpha$ (\cref{lem:swap_alpha}) & $\beta$ (\cref{lem:swap_order_preservation}) & $\gamma$ (\cref{lem:swap_order_preservation})\\
            \midrule{}\hspace{-1mm}
            SNTV & $1-\Theta\inparen{\phi t}$ & $1-\Theta\inparen{\phi t}$ & $1-\Theta\inparen{\phi t}$  \\
            $\ell_1$-CC & $1-\Theta\inparen{\phi t}$ & $1-\Theta\inparen{\phi t}$ & $1-\Theta\inparen{\phi t}$ \\
            Borda & $1-\Theta\inparen{\phi t}$ & $1-\Theta\inparen{\phi t}$ & $1-\Theta\inparen{\phi t}$ \\
            \bottomrule{}
        \end{tabular} 
        %\vspace{-0.2in}
        \caption{
            Smoothness parameters for the swapping-based model (\cref{def:swapping}). 
            The formal statements of the results appear as \cref{lem:swap_alpha,lem:swap_order_preservation}.
        }
        \label{tab:smoothness_results2}
    \end{table}

    \begin{theorem}[\textbf{Algorithmic result for the utility-based generative model; Informal}]
    \label{cor:algorithmic}
        Let $(\mu,\wh{\mu})$ be the utility-based generative models from \cref{ex:utility_generative,def:utility_bias} with bias parameter $\theta\in (0,1]$.
        For $\ell_1$-CC and SNTV it holds that $\alpha\geq \Theta\inparen{\theta^{-2(m-1)}}$, and for Borda that $\alpha\geq \Theta\inparen{\theta^{-2}}$.
        
        Using this and \Cref{lem:utility_alpha}, \Cref{thm:main_algorithmic} applies for 
        \begin{enumerate} %
            \item SNTV and $\ell_1$-CC with $n\geq {\theta^{-2(m-1)}\cdot m^{\Theta(1)}\eps^{-2}}$ and $\beta=1-m^{-\Theta(1)}$; and
            \item Borda with $n\geq {\theta^{-2}\cdot m^{\Theta(1)}\eps^{-2}}$ and $\beta=1-m^{-\Theta(1)}$.
        \end{enumerate}
    \end{theorem}  
    Note that it is quite intuitive that the above computed $\alpha$ values depend on $\theta$ as in the utility-based model $\theta$ controls the multiplicative gap between the scores awarded to candidates from $G_2$ compared to $G_1$ and, thus, the value of $\frac{1}{\tau_1(f)}\min_{c\in M} \Ex_{\wh{\mu}}\insquare{f_{C\setminus \{c\}}(c,\not\succ)}$.

    The sufficient number of voters in the above result varies significantly depending on the multiwinner score function: on the one hand, for  $\ell_1$-CC rule and SNTV the dependence is $\theta^{-O(m)}$, on the other hand, the dependence is only $\theta^{-2}$ for the Borda rule.
    We can also prove that these dependencies are not only sufficient but also necessary by providing an impossibility result (\Cref{thm:main_impossibility}) that shows that representational constraints cannot recover an (approximately) optimal solution if $n$ is  ``substantially''  smaller than these bounds; see \Cref{sec:impossible} for more discussions.
    Combined with our impossibility result, \Cref{cor:algorithmic} shows a stark contrast between different score functions, e.g., SNTV and Borda, under the utility-based model, implying that the latter function is advantageous in the presence of implicit bias.

    {The above results extend to certain mixtures of generative models. For instance, if $F$ is $(\alpha_1,\beta_1,\gamma_1)$-smooth with respect to $(\mu,\wh{\mu}_1)$ and $(\alpha_2,\beta_2,\gamma_2)$-smooth with respect to $(\mu,\wh{\mu}_2)$, then for any $\delta\in (0,1)$, it can be shown that $F$ is $(\delta\alpha_1+(1-\delta)\alpha_1,\min\inbrace{\beta_1,\beta_2},\max\inbrace{\gamma_1,\gamma_2})$-smooth with respect to the mixture $(\mu,\delta\wh{\mu}_1+(1-\delta)\wh{\mu}_2)$; this follows from \cref{def:smooth} and linearity of expectation.}

    \paragraph{A tool for analyzing multiwinner score functions
    (\cref{sec:simulations_synthetic}).}
    In addition to the above computations for \textit{existing} bias models and multiwinner score functions, we also provide code to estimate the smoothness of \textit{new} multiwinner score functions with respect to \textit{new} generative models (\cref{sec:simulations_synthetic}).
    The code takes as input oracles that (1) evaluate the multiwinner score function $F$ and (2) sample from generative models $(\mu,\wh{\mu})$. 
    First, for specified $m$ and $k$, it outputs estimates $(\wt{\alpha},\wt{\beta},\wt{\gamma})$ of the smoothness of $F$ with respect to  $(\mu,\wh{\mu})$, along with corresponding confidence intervals (implied by a concentration inequality; \cref{lem:concentration}).
    This allows for theoretical estimates of the capabilities of representational constraints using our main result (\cref{thm:main_algorithmic}).
    
    Second, given values of $n$, $m$, and $k$, it estimates the fraction of the optimal score recovered by representational constraints for $F$ with respect to the given $(\mu,\wh{\mu})$.
    In \cref{sec:simulations_synthetic}, we illustrate the code using a set of latent generative models provided by \citet{DBLP:conf/atal/SzufaFSST20} in combination with the swapping-based bias model (\cref{def:swapping}).
    In line with \cref{cor:algorithmic}, our observations in these simulations show a stark contrast between SNTV and Borda: for all values of $n$ and generative models we consider, representational constraints recover a significantly larger fraction of the maximum achievable latent score for Borda than for SNTV.

    \subsection{More Discussion on \cref{thm:main_algorithmic}}\label{sec:additional_remarks}

    Our definition of smoothness (\cref{def:smooth}) is also relevant to works that study the robustness of scoring functions to noise in the ``ground truth'' preference list. These works assess the capabilities of scoring functions to uncover some ground truth from noisy signals, which is a popular question in the field of epistemic social choice \cite{DBLP:conf/uai/ProcacciaRS12,caragiannis2016noisy,Caragiannis2017Learning,DBLP:journals/aamas/CaragiannisKKK22}.
    Our results and smoothness definitions extend to this setting as follows. 
    Suppose there is a ground truth preference list $\succ_{\rm truth}$ of all candidates. The generative model $\mu$ of latent preferences returns $\succ_{\rm truth}$ with probability one, and all candidates are part of the disadvantaged group (i.e., $G_1$ is empty). Now, the noise model one is interested in simply becomes our generative model $\hat{\mu}$ for biased preference lists.
        
    The smoothness of a function then gives insights into its robustness to such noise generalizing some concepts from the literature: For instance, in a closely related setting, \citet{DBLP:journals/aamas/CaragiannisKKK22}  study the robustness of approval-based multiwinner scoring functions and introduce the notion of ``accurate in the limit,'' which roughly speaking means that if the number of voters is high enough, then the underlying best committee maximizes the score function with respect to the noisy votes with high probability. This can be captured by our smoothness definition. If for some noise model, the scoring function $F$ is $(\alpha,\beta,\gamma)$-smooth with $\beta=1$, then our results imply that $F$ is accurate in the limit (\cref{thm:main_algorithmic}). 
    Moreover, unlike existing work, \cref{thm:main_algorithmic} also gives a bound on how many noisy votes are required to achieve a $(1-\varepsilon)$-``approximately'' optimal committee \cite{DBLP:journals/aamas/CaragiannisKKK22}.
        
    \blue{Further, our work is related to a line of works on predicting the outcome of an election by sampling some of the voters \cite{DBLP:journals/ai/BhattacharyyaD21,DBLP:journals/corr/abs-2203-00083}.} 
    \begin{remark}[\bf{Comparing $\alpha$ with other notations}]
        \blue{Several works \cite{DBLP:journals/ai/BhattacharyyaD21,DBLP:journals/corr/abs-2203-00083} bound the number of voters required to accurately predict the outcome of the election.}
        Such sampling  bounds are often parameterized by the margin of victory \cite{DBLP:conf/uss/MagrinoRS11,DBLP:conf/sigecom/Xia12}, i.e., the lead of the election winner in the full election. 
        Conceptually, the margin of victory is related to $\alpha$ in the definition of smoothness (\cref{def:smooth}), as $\alpha$ captures the quality of the weakest candidate of the winning committee and thereby in some sense its ``lead'' against the remaining ones.
    \end{remark}

    \noindent Also, note that $\alpha$ has a similar form as the curvature of submodular functions. 
    However, there are some significant differences between them implying that the curvature does not measure the effectiveness of representational constraints, as, e.g., the curvature is unable to distinguish modular functions--as explained next.

\begin{remark}[\textbf{Comparison of $\alpha$ and the curvature of submodular functions}]\label{remark:curvature}
    The curvature of a submodular function often shows up in approximation-ratios of (constrained) monotone submodular maximization algorithms \cite{conforti1984submodular,vondrak2010submodularity}. 
    The curvature $\lambda(f)\in [0,1]$ of $\Ex_\mu\insquare{f(\cdot,\succ)}$ is defined as $1-\min_{S\subseteq C, c\not\in S} \frac{\Ex_\mu[f_S(c,\succ)]}{\Ex_\mu[f(c,\succ)]}$.\footnote{Note that the curvature of $\Ex_\mu\insquare{f(\cdot,\succ)}$ is well-defined as it is a submodular function.}
    Hence, $\frac{1}{1-\lambda(f)}=\min_{S\subseteq C, c\not\in S} \frac{\Ex_\mu[f_S(c,\succ)]}{\Ex_\mu[f(c,\succ)]}$.
    If $\mu=\wh{\mu}$, then this has a similar form as $\alpha$ in \cref{def:smooth}, but there are a few important differences:
    the main difference is the denominator in $\frac{1}{1-\lambda(f)}$ depends both on $c$ and $f$, whereas the denominator in $\alpha$ depends $f$ but not $c$.
    Because of this difference, while $\lambda(f)$ measures the ``closeness'' of $\Ex_\mu\insquare{f(\cdot)}$ to being modular, $\alpha$ is not related to the ``closeness'' to being modular. 
    Many common multiwinner voting functions (including the SNTV and Borda rules) are modular.
    Hence, the curvature $\lambda(f)$ of all of these functions (including, both the SNTV and the Borda rule) are the same (equal to 0). %
    However, as our results show, in some cases, representational constraints have significantly higher effectiveness for the Borda rule compared to the SNTV rule (\cref{cor:algorithmic}).
    Thus, the curvature is not the right parameter to measure the effectiveness of lower-bound constraints.
    In contrast to the curvature, $\alpha$ varies across modular functions (e.g., \cref{cor:algorithmic}).
\end{remark}

%%%%%%%%%%%%%%%%%%%%%%%%%%%%%%%%%%%%%%%%%%%%%%%%%%%%%%%%%%%%%%%    
\section{Proofs of Algorithmic Results in Section~\ref{sec:algorithmic}}
\label{sec:proof}
    In this section, we prove our algorithmic results.
    In Section~\ref{sec:proofof:thm:main_algorithmic}, we propose our algorithm (Algorithm~\ref{algo:main_algorithm}) and prove Theorem~\ref{thm:main_algorithmic}.
    In Section~\ref{sec:proofof:lem:utility_alpha}, we prove Theorem~\ref{cor:algorithmic} by first showing the order preservation of the utility-based models (Lemmas~\ref{lem:utility_order_preserve} and~\ref{lem:utility_order_preserve2}) and then bounding term $\alpha$ in the smoothness notion (Lemma~\ref{lem:utility_alpha}).
    In Section~\ref{sec:swapping-based}, we investigate another common bias generative model, called the swapping-based model (Definition~\ref{def:swapping}), and similarly analyze its order-preservation (Lemma~\ref{lem:swap_order_preservation}) and smoothness (Lemma~\ref{lem:swap_alpha}).
    Table~\ref{tab:notation} summarizes the used symbols.
    
    \begin{table}[ht!]
        \centering
        \footnotesize
        \subfigure[Basic notation]{
            \footnotesize
            \hspace{-10mm}\begin{tabular}{p{1cm}p{6cm}}
            \toprule 
            \textbf{Symbol} & \textbf{Meaning}\\
            \midrule{}\hspace{-1mm}
            $n$ & Number of voters\\
            $m$ & Number of candidates\\
            $k$ & Number of selected candidates\\
            $V$ & Set of voters\\
            $C$ & Set of candidates \\
            $\succ$ & A ``latent'' ranking of all candidates in $C$\\
            $\nsucc$ & A ``biased'' ranking of all candidates in $C$\\
            $\mu$ & Generative model of latent preferences (\cref{def:latent_preference})\\
            $\wh{\mu}$ & Generative model of biased preferences (\cref{def:biased_preference})\\
             $\mathcal{L}(C)$ & Set of all strict and complete orders over $C$ \\ 
             $G_1,G_2$ & Advantaged and disadvantaged groups, respectively. Disjoint subsets of $C$.\\ 
             $F$ & Multiwinner score function $\score(S) = \sum_{v\in V} f(S, \succ_v)$ (\cref{def:score})\\ 
            \bottomrule{}
        \end{tabular}
        } 
        \footnotesize
        \subfigure[Notation specific to multiwinner scoring functions]{
            \footnotesize
            \begin{tabular}{lp{6cm}}
            \toprule 
            \textbf{Symbol} & \textbf{Meaning}\\
            \midrule{}\hspace{-1mm}
             $f$ & A function $f\colon 2^C\times \prefs{C} \to \R_{\geq 0}$ such that $\score(S) = \sum_{v\in V} f(S, \succ_v)$; see \cref{def:score}\\[5mm]
             $f_S(c,\succ)$ & The marginal contribution of $c$ to $S$ with respect to $f$ and $\succ$: $f_S(c,\succ) \coloneqq f(S\cup \{c\}, \succ) - f(S, \succ)$.\\[5mm]
             $\pos_\succ(c)$ & Position of $c$ in the preference list $\succ$\\[5mm]
             $\tau_1(f)$ & The maximum possible expected score $\Exp_{\mu}\left[f(c,\succ)\right]$ of a candidate (in the case of multiwinner score functions this is the score which function $f(\cdot ,\succ)$ awards to the set consisting of the candidate ranked in the first position of $\succ$)\\
            \bottomrule{}
        \end{tabular}
        }\par%\vspace{-0.1in}
        \subfigure[Notation specific to the utility-based model]{
            \hspace{-10mm}
            \begin{tabular}{lp{6cm}}
            \toprule 
            \textbf{Symbol} & \textbf{Meaning}\\
            \midrule{}\hspace{-1mm}
             $\theta$ & Bias parameter (\cref{def:utility_bias})\\[2.15mm]
             $\omega_c$ & Intrinsic utility of $c\in C$ (\cref{ex:utility_generative})\\[2.15mm]
             $w_{v,c}$ & Latent utility of $c\in C$ observed by voter $v$ (\cref{ex:utility_generative})\\[2.15mm]
             $\hat{w}_{v,c}$ & Biased utility of $c\in C$ observed by voter $v$  (\cref{def:utility_bias})\\[2.15mm]
            \bottomrule{}
        \end{tabular}
        }
        \subfigure[Notation specific to the swapping-based model]{
            \footnotesize
            \begin{tabular}{lp{6cm}}
            \toprule 
            \textbf{Symbol} & \textbf{Meaning}\\
            \midrule{}\hspace{-1mm}
             $\phi$ & Bias parameter (\cref{def:swapping}) \\
             $t$ & Number of swaps (\cref{def:swapping})\\
             $A(\succ)$ & The collection of all pairs $(i,j)$ such that there exist $c \in G_1$ and $c' \in G_2$ with $\pos_{\succ}(c) = i > j = \pos_{\succ}(c')$ (\cref{def:swapping})\\ 
             $Z(\succ)$ & Normalization factor: $\sum_{(i',j')\in A(\succ)} \phi^{i'-j'}$ (\cref{def:swapping})\\
            \bottomrule{}
        \end{tabular}
        }\par%\vspace{-0.1in}
        \subfigure[Voting rules and smoothness parameters]{ 
            \hspace{-10mm}\begin{tabular}{lp{6cm}}
            \toprule 
            \textbf{Symbol} & \textbf{Meaning}\\
            \midrule{}\hspace{-1mm}
            SNTV & Single non-transferable vote; defined by  $f(S,\succ)=\sum_{c\in S} \mathbbm{1}_{\pos_{\succ}(c)=1}$\\[0.9mm]
            Borda & A multiwinner score functions defined by $f(S,\succ)= \sum_{c\in S} \inparen{m-\pos_{\succ}(c)}$\\[0.9mm]
            $\ell_1$-CC & $\ell_1$ Chamberlin-Courant rule; defined by $f(S,\succ)= \max_{c\in S} \inbrace{m-\pos_{\succ}(c)}$\\[0.9mm]
            $\alpha$ & \cref{def:smooth}\\[0.9mm]
            $\beta$ & \cref{def:bias}\\[0.9mm]
            $\gamma$ & \cref{def:bias}\\ 
            \bottomrule{}
        \end{tabular}
        }
        \subfigure[Special subsets]{ 
            \begin{tabular}{lp{6cm}}
            \toprule  
            \textbf{Symbol} & \textbf{Meaning}\\
            \midrule{}\hspace{-1mm}
            $R$ & Voters' latent preference profile $\left\{\succ_v \in \prefs{C}: v\in V\right\}$\\
            $\wh{R}$ & Voters' biased preference profile $ \left\{\nsucc_v \in \prefs{C}: v\in V\right\}$\\
             $S^\star$ & $\arg\max_{S\subseteq C: |S|=k} \score(S)$\\
             $\wh{S}$ & $\arg\max_{S\in \calK(\ell): |S| = k} \widehat{\score}(S)$\\
             $M$ & The collection of $k$ candidates $c$ with the highest value $\Exp_{\mu}\left[f(c,\succ)\right]$ \\
             $\cK(\ell)$ & The collection of subsets of $C$ that have at least $\ell$ candidates from $G_2$\\
            \bottomrule{}
        \end{tabular}
        }
        \caption{ 
            Table of notations.
        }
        \label{tab:notation}
    \end{table}

\subsection{Proof of \cref{thm:main_algorithmic}: Main Algorithmic Result}\label{sec:proofof:thm:main_algorithmic}
    We first provide the algorithm for \cref{thm:main_algorithmic}, say \cref{algo:main_algorithm}, which is a simple greedy algorithm that first selects $\ell$ candidates from $G_2$ in Line 1, then selects $k-\ell$ candidates from $G_1$ in Line 2, and finally outputs their union $S$.

    \begin{algorithm}[ht!] %
            \caption{A greedy algorithm with an intervention constraint}\label{algo:main_algorithm}
            \begin{algorithmic}[1]
                \STATE {\bfseries Input:} Numbers $k,\ell\in \N$, sets $G_1,G_2\subseteq C$, and a value oracle $\calO$ for $\hF(\cdot)$
                \STATE {\bfseries Output:} A subset $S\subseteq \calK(\ell)$ of size $k$
                \vspace{2mm}
                \STATE Select $S_2\coloneqq \hyperref[alg:greedy]{\textsc{Greedy}}(\ell, G_2, \calO, B=\emptyset)$
                \STATE Select $S_1\coloneqq \hyperref[alg:greedy]{\textsc{Greedy}}(k, G_1, \calO, B=S_2)$%
                \STATE \textbf{return} $S \coloneqq S_1\cup S_2$
            \end{algorithmic}
    \end{algorithm}

    \begin{algorithm}[ht!] %
            \caption{\textsc{Greedy}(Oracle for $F$, $C$, $k$) (\cite{nemhauser1978analysis}) }\label{alg:greedy}
            \begin{algorithmic}[1]
                \STATE {\bfseries Input:} A number $\ell\in \N$, two sets $B$ and $G$, and a value oracle $\calO$ for $\hF(\cdot)$
                \STATE {\bfseries Output:} A subset $S\subseteq G\cup B$ with $\abs{S} = k$
                \STATE Initialize $S = B$
                \WHILE{$\abs{S} < k$}
                    \STATE Set $S=S\cup \arg\max_{i\in G}F_S(i)$
                \ENDWHILE
                \STATE \textbf{return} $S$
            \end{algorithmic}
        \end{algorithm}
    To prove \cref{thm:main_algorithmic}, we need to show that for any multiwinner score voting $\score: 2^C\rightarrow \R_{\geq 0}$ that is $(\alpha,\beta,\gamma)$-smooth with respect to generative models $(\mu,\wh{\mu})$ the following holds.
    If the number of voters is at least 
    $n \geq \Omega\inparen{k(\alpha\eps_0)^{-2}\cdot \log{\frac{2}{\delta_0}}}$ 
    (for any $0<\eps_0,\delta_0<1$),
    then there exists an integer $0\leq \ell\leq m$ specifying the lower bound constraint such that 
    \[
        \Pr_{\mu,\wh{\mu}}\insquare{F(S) \geq \inparen{\beta-\eps_0} \cdot  \OPT} \geq 1-\delta_0.
    \]
    Where $S$ is the subset output by \cref{algo:main_algorithm}, given the number $\ell$, an oracle $\calO$ for $\hF(\cdot)$, and other parameters (namely, $k$, $G_1$, and $G_2$) as input.
    Moreover, \cref{algo:main_algorithm}  makes $O(mk)$ calls to $\calO$ and performs $O(m\log{m})$ arithmetic operations.

    Fix any $\alpha,\beta,\gamma \in (0,1]$.
    Let $\score(\cdot) = \sum_{v\in V} f(\cdot ,\succ_v)$ be any multiwinner score function  that is $(\alpha,\beta,\gamma)$-smooth with respect to generative models $(\mu,\wh{\mu})$.
    Recall that $M\subseteq C$ is the set of $k$ candidates $c$ with the highest value $\Exp_{\mu}\left[f(c,\succ)\right]$.
    Define $\ell$ as the following value
    \[
        \ell \coloneqq \abs{M\cap G_2}.
    \] 
    We use $\tau$ to represent $\tau_1(f)$ for simplicity.
    We claim that this $\ell$ satisfies the claim in \cref{thm:main_algorithmic}.
    To simplify the notation, we define the following parameters 
    \[
        \eps \coloneqq \frac{\min\inbrace{\eps_0,1-\gamma}}{{4}\beta},\quad
        \delta \coloneqq \frac{\delta_0}{{3}},\quad\text{and}\quad
        n_0(\eps_0,\delta_0) \coloneqq \frac{51 k}{\inparen{\alpha\min\inbrace{\eps_0,1-\gamma}}^2}\cdot \log{\frac{2}{\delta_0}}.
    \]

    \noindent We divide the proof of \cref{thm:main_algorithmic} into the following two lemmas.
    \begin{lemma}[\bf{$S$ approximates $M$}]
    \label{lem:score_sell}
        For any $0<\eps_0,\delta_0<1$, if
        $n \geq n_0(\eps_0,\delta_0)$,
        then it holds that 
        \[
            \Pr_{\mu,\wh{\mu}}\insquare{
                F(S) \geq \beta \cdot (1-\eps)\cdot F(M)
            }
            \geq 1 - \delta.
            \yesnum\label{eq:lem_score_sell}
        \]
    \end{lemma}
    
    \begin{lemma}[\bf{$M$ is near optimal}]
    \label{lem:score_m}
        For any $0<\eps_0,\delta_0<1$, if
        $n \geq n_0(\eps_0,\delta_0)$,
        then it holds that 
        \[
            \Pr_{\mu,\wh{\mu}}\insquare{
                F(M) \geq (1-\eps)\cdot F(S^\star)
            }
            \geq 1 - 2\delta.
            \yesnum\label{eq:lem_score_m}
        \]
    \end{lemma}
    \noindent \cref{thm:main_algorithmic} follows from the above lemmas as follows.
    \begin{proof}[Proof of \cref{thm:main_algorithmic} assuming \cref{lem:score_m,lem:score_sell}]
        Due to the lower bound on $n$ in \cref{thm:main_algorithmic}, it holds that $n \geq n_0(\eps_0,\delta_0)$.
        Hence, from  \cref{lem:score_m,lem:score_sell} the \cref{eq:lem_score_sell,eq:lem_score_m} hold.
        Since $3\delta \leq \delta_0$, taking a union bound over \cref{eq:lem_score_sell,eq:lem_score_m}, it follows that 
        \[
            \Pr_{\mu,\wh{\mu}}\insquare{
                F(S) \geq \beta \cdot (1-\eps)^2 \cdot F(S^\star) 
            }
            \geq 1 - \delta.
        \]
        Since $(1-\eps)^2\geq 1-2\eps$ (for any $\eps\in\R$) and $2\eps\cdot \beta\leq \eps_0$, it follows, as required, that
        \[
            \Pr_{\mu,\wh{\mu}}\insquare{
                F(S) \geq \inparen{\beta-\eps_0} \cdot F(S^\star) 
            }
            \geq 1 - \delta.
        \]
        It remains to bound the number of calls and the number of arithmetic operations in \cref{algo:main_algorithm}.
        Note that \cref{alg:greedy} is called as a subroutine from \cref{algo:main_algorithm} in Steps 1 and 2.
        Each run of \cref{alg:greedy} makes exactly $k\abs{G}$ calls to $\calO$ and does $O(\abs{G}\log{\abs{G}})$ arithmetic operations (to sort the marginal scores).
        \cref{alg:greedy} is called twice in \cref{algo:main_algorithm}, once with parameters $(k=\ell,\abs{B}=0,G=G_1)$, and once with $(k=k,\abs{B}=\ell,G=G_2)$.
        Since $\abs{G_1},\abs{G_2}\leq m$, the total number oracle calls is $O(mk)$ and the total number of arithmetic operations are $O(m\log{m})$.
    \end{proof}
    In the remainder of this section, we prove \cref{lem:score_sell,lem:score_m}.
    The proof of both \cref{lem:score_sell,lem:score_m} uses the following concentration result.
    \begin{lemma}[\textbf{Concentration of marginal utilities}]\label{lem:concentration}
        For any generative models of preference lists $\mu,\wh{\mu}$ and any score function $\score(\cdot) = \sum_{v\in V} f(\cdot ,\succ_v)$ from \cref{def:score} the following holds.
        For any $\delta>0$
        \begin{align*}
            \Pr_{\mu}\insquare{
                    \exists_{T=C \text{ or } (T\subseteq C\colon \abs{T}\leq k)},\ \ \exists_{c\in C}, \quad 
                    \abs{F_T(c) - \Ex_\mu\insquare{F_T(c)}} \geq \err{}
            }
            &\leq 
            \delta,\\
            \Pr_{\wh{\mu}}\insquare{
                    \exists_{T=C \text{ or } (T\subseteq C\colon \abs{T}\leq k)},\ \ \exists_{c\in C}, \quad 
                    \abs{\hF_T(c) - \Ex_{\wh{\mu}}\insquare{\hF_T(c)}} \geq \err{}
            }
            &\leq 
            \delta.  
        \end{align*} 
    \end{lemma}
    Here, we slightly abuse the notation and denote singleton sets $\inbrace{c}$ by $c$.
    We also note in passing that concentration inequality holds for any generative models of preference lists $(\mu,\wh{\mu})$ and not just the generative models that satisfy \cref{def:bias}.
    \begin{proof}
        We first prove the first inequality.
        Since $F$ is a separable function (\cref{def:score}) and for each voter $v\in V$, their preference list $\succ_v$ is drawn iid from $\mu$, for any $c\in C$, $\inbrace{f(c,\succ_v)}_{v\in V}$ is a set of iid and bounded random variables.
        The concentration inequality follows from Hoeffding's inequality  and the union bound \cite{motwani1995randomized} as shown next.

        Fix any $T\subseteq C$ and any $c\in C$.
        For each $v\in V$, define the random variable $Z_v\coloneqq f_T(c,\succ_v)$.
        As discussed, $Z_u$ and $Z_v$ are independent for any $u\neq v$.
        Moreover, for all $v\in V$, $0\leq Z_v\leq \tau$ with probability 1 (by the non-negativity of $f$ and the definition of $\tau$).
        Hence, Hoeffding's inequality is applicable on $F(c)=\sum_{v\in V}f(c,\succ_v)$ \cite{motwani1995randomized}.
        From the Hoeffding's inequality \cite{motwani1995randomized}, it holds that
        \begin{align*}
            \Pr_{\mu}\insquare{
                    \abs{F_T(c) - \Ex_\mu\insquare{F_T(c)}} \geq \err{}
            }
            &\leq 
            \exp\inparen{-\frac{2}{n\tau^2}\cdot \tau^2 \cdot n\cdot k\cdot \log\inparen{\frac{m}{\delta}}}\\
            &\leq 
            \frac{\delta^{2k}}{m^{2k}}\\
            &\leq 
            \frac{\delta}{m^{2k}}. \tag{Using that $0\leq \delta \leq 1$}
        \end{align*} 
        The first concentration inequality in \cref{lem:concentration} follows by taking the union bound over all choices of $T\subseteq C$ of either (1) size at most $k$ or (2) $T=C$, and any $c\in C$,  as there are at most $2^k \cdot m+1\leq m^{2k}$ choices of $(T,c)$.
        The proof of the second inequality follows by replacing $\mu$ and $F$ by $\wh{\mu}$ and $\hF$ in the above proof.
    \end{proof}

    \subsubsection{Proof of \cref{lem:score_sell}: Relation between $S$ and $M$}

            The proof is divided into multiple steps. We begin by defining the additional notation used in this proof.

                \paragraph{Notation.}
                Recall that $\ell=\abs{M\cap G_2}$.
                Define the following sets 
                \begin{align*}
                    S\cap G_1&\coloneqq \inbrace{a_1,a_2,\dots,a_{k-\ell}}
                    \hspace{0.5mm} \quad\text{and}\quad 
                    S\cap G_2 \coloneqq  \inbrace{b_1,b_2,\dots,b_{\ell}},\\
                    M\cap G_1 &\coloneqq  \inbrace{x_1,x_2,\dots,x_{k-\ell}}
                    \quad\text{and}\quad 
                    M\cap G_2  \coloneqq  \inbrace{y_1,y_2,\dots,y_{\ell}}.
                \end{align*}
                Where the elements in $S\cap G_1$ and $S\cap G_2$ are in the order they are selected in \cref{algo:main_algorithm}.
                The elements in $M\cap G_1$ and $M\cap G_2$ are ordered in non-increasing order by $\Ex_\mu\insquare{F(\cdot)}$:
                for all $i\in [k-\ell]$ and $j\in [\ell]$, 
                \begin{align*}
                    \Ex_\mu\insquare{F(x_i)}&\geq \Ex_\mu\insquare{F(x_{i+1})}
                    \quad\text{and}\quad 
                    \Ex_\mu\insquare{F(y_j)}\geq \Ex_\mu\insquare{F(y_{j+1})}.
                \end{align*}
                To simplify the notation, for each $i\in [k-\ell]$ and $j\in [\ell]$, define the following prefixes:
                \begin{align*}
                    A(i) &\coloneqq \inbrace{a_1,a_2,\dots,a_{i}}
                    \hspace{0.5mm} \quad\text{and}\quad
                    B(j) \coloneqq  \inbrace{b_1,b_2,\dots,b_{j}},\\
                    X(i) &\coloneqq  \inbrace{x_1,x_2,\dots,x_{i}}
                    \quad\text{and}\quad 
                    Y(j) \coloneqq  \inbrace{y_1,y_2,\dots,y_{j}}.
                \end{align*}
                Define $A(0)$, $B(0)$, $X(0)$, and $Y(0)$ as empty sets.
                Since $B(j-1)$ has $j-1$ elements by the Pigeonhole principle (for any $j\in [\ell]$), there exists a $y\in Y(j)$ such that $y\not\in B(j-1)$.
                Label this $y$ as $y_{(j)}.$
                Similarly, there exists an $x_{(i)}\in X(i)$ (for any $i\in [k-\ell]$) such that $x_{(i)}\not\in A(i-1)$. 
                Let $\evE$ be the following event 
                \begin{align*}
                    \forall_{T=C \text{ or } (T\subseteq C\colon \abs{T}\leq k)},\ \ \forall_{c\in C},\qquad 
                    &\abs{F_T(c) - \Ex_\mu\insquare{F_T(c)}}  \ \hspace{0.5mm}  \leq \  \err{},\\
                    \forall_{T=C \text{ or } (T\subseteq C\colon \abs{T}\leq k)},\ \ \forall_{c\in C},\qquad 
                    &\abs{\hF_T(c) - \Ex_{\wh{\mu}}\insquare{\hF_T(c)}} \leq \err{}.
                \end{align*}
                    
                \cref{lem:concentration} shows that $\Pr[\evE]\geq 1-2\delta$.

                \begin{lemma}\label{lem:lower_bound_on_marginal}
                    Fix any $i\in [k-\ell]$ and $j\in [\ell]$.
                    Conditioned on the event $\evE$, the following inequalities hold
                    \begin{align*}
                        {F_{B(\ell)\cup A(i-1)}(a_i)} \geq (1-\eps)\cdot \beta\cdot {F_{B(\ell)\cup  A(i-1)}(x_{(i)})}
                        \quad\text{and}\quad 
                        {F_{B(j-1)}(b_j)} \geq (1-\eps)\cdot \beta\cdot {F_{B(j-1)}(y_{(j)})}.
                    \end{align*}
                \end{lemma}
                \begin{proof}
                    Fix any $i\in [k-\ell]$ and $j\in [\ell]$.
                    Suppose the event $\evE$ holds.

                \paragraph{Step 1 (Lower bound on $\hF_{B(j-1)(b_j)}$ and $\hF_{B(j-1)(y_{(j)})}$):}
                    Consider the step where the set of items selected so far is $B(j-1)$.
                    At this step, \cref{algo:main_algorithm} selects the item with the largest value with respect to  $\hF_{B(j-1)}(\cdot)$.
                    Since \cref{algo:main_algorithm} selects $b_j$ instead of $y_{(j)}$.
                    it must hold that 
                    \[
                        \hF_{B(j-1)}(b_j) \geq \hF_{B(j-1)}(y_{(j)}).
                        \yesnum\label{eq:lem:score_sell:lb_1}
                    \]
                    Since $\evE$ holds, it follows that
                    \[
                        \hF_{B(j-1)}(y_{(j)})
                        \geq 
                        \Ex_{\wh{\mu}}\insquare{\hF_{B(j-1)}(y_{(j)})} - \err{}.
                        \yesnum\label{eq:lem:score_sell:lb_2}
                    \]
                    Since $y_{(j)}\in M$ and $y_{(j)}\not\in B(j-1)$, one can use the definition of $\alpha$ (\cref{def:smooth}) to get the following lower bound
                    \begin{align}
                        \label{eq:lb_on_bj}
                        \begin{aligned}
                        & \quad \Ex_{\wh{\mu}}\insquare{\hF_{B(j-1)}(y_{(j)})} & \\
                        \geq & \quad                        \Ex_{\wh{\mu}}\insquare{\hF_{C\setminus\inbrace{y_{(j)}}}(y_{(j)})} &
                        (\text{$F_R(c)\geq F_{R\cup T}(c)$ holds for any sets $R$ and $T$ and element $c$})\\
                        \geq & \quad
                        \alpha\tau n. & (\text{Using the definition of $\alpha$ and that $y_{(j)}\in M$})
                        \end{aligned}
                    \end{align}
                    Using Inequalities~\eqref{eq:lem:score_sell:lb_1}, \eqref{eq:lem:score_sell:lb_2}, and \eqref{eq:lb_on_bj} and the fact that $n\geq n(\eps_0,\delta_0)\geq 4\alpha^{-2}k\log{\frac{m}{\delta}}$, it follows that 
                    \begin{align*}
                        \hF_{B(j-1)}(b_j),\ \ \hF_{B(j-1)}(y_{(j)}) \geq \frac{\alpha\tau n}{2}.
                        \yesnum\label{eq:lb_on_y}
                    \end{align*}

                    \paragraph{Step 2 (Lower bound on $\frac{\Ex_{{\mu}}\insquare{F_{B(j-1)}(b_j)}}{\Ex_{{\mu}}\insquare{F_{B(j-1)}(y_{(j)})}}$):}
                    From Inequalities~\eqref{eq:lem:score_sell:lb_2} and \eqref{eq:lb_on_bj}, and the fact that  $n\geq n(\eps_0,\delta_0)\geq 4\eps^{-2}\alpha^{-2}k\log{\frac{m}{\delta}}$, it follows that 
                    \begin{align*}
                        \hF_{B(j-1)}(y_{(j)})
                        \geq 
                        \inparen{1-\eps}\cdot \Ex_{\wh{\mu}}\insquare{\hF_{B(j-1)}(y_{(j)})}.
                        \yesnum\label{eq:lb_3}
                    \end{align*}
                    Since the event $\evE$ holds, it also follows that 
                    \begin{align*}
                        \Ex_{\wh{\mu}}\insquare{\hF_{B(j-1)}(b_j)}
                        \ \ &\Stackrel{}{\geq}\ \  \hF_{B(j-1)}(b_j)-\err{}\\ 
                        \ \ &\Stackrel{\eqref{eq:lem:score_sell:lb_1}}{\geq}\ \ \hF_{B(j-1)}(y_{(j)})-\err{}\\ 
                        \ \ &\Stackrel{}{\geq}\ \ (1-\eps)\cdot \hF_{B(j-1)}(y_{(j)}). \tagnum{Using \cref{eq:lb_on_y} and the fact that $n\geq n(\eps_0,\delta_0)\geq 4\eps^{-2}\alpha^{-2}k\log{\frac{m}{\delta}}$}
                        \customlabel{eq:lb_on_exp_bj}{\theequation}
                    \end{align*}
                    Chaining Inequalities~\eqref{eq:lb_3} and  \eqref{eq:lb_on_exp_bj}, it follows that
                    \begin{align*}
                        \Ex_{\wh{\mu}}\insquare{\hF_{B(j-1)}(b_j)}
                        \ \ \Stackrel{\eqref{eq:lb_on_exp_bj}}{\geq}\ \  (1-\eps)\cdot \hF_{B(j-1)}(y_{(j)})
                        \ \ \Stackrel{\eqref{eq:lb_3}}{\geq}\ \   
                        \inparen{1-\eps}^2 \cdot \Ex_{\wh{\mu}}\insquare{\hF_{B(j-1)}(y_{(j)})}.
                        \yesnum\label{eq:lb_ratio_of_biased_scores}
                    \end{align*}
                    If $\frac{\Ex_{{\mu}}\insquare{F_{B(j-1)}(b_j)}}{\Ex_{{\mu}}\insquare{F_{B(j-1)}(y_{(j)})}}\leq \beta$, then 
                    using $(\beta,\gamma)$ order-preservation between $\mu$ and $\wh{\mu}$, it follows that 
                    \begin{align*}
                        \frac{\Ex_{\wh{\mu}}\insquare{\hF_{B(j-1)}(b_j)}}{\Ex_{\wh{\mu}}\insquare{\hF_{B(j-1)}(y_{(j)})}} \ \ \leq \ \ \gamma. 
                        \yesnum\label{eq:using_beta_gamma_order_preservation}
                    \end{align*}
                    Since $\evE$ holds, the above inequality implies that 
                    \begin{align*}
                        \hF_{B(j-1)}(y_{(j)})
                        &\geq \gamma^{-1} \cdot\hF_{B(j-1)}(b_j)  - 2\err{} \tag{Using that $\gamma\leq 1$}\\
                        &\geq \gamma^{-1}\cdot (1-\eps) \cdot \hF_{B(j-1)}(b_j)
                        \tag{Using \cref{eq:lb_on_y} and the fact that $n\geq n(\eps_0,\delta_0)\geq 4\gamma^2(\eps\alpha)^{-2}k\log{\frac{m}{\delta}}$}\\ 
                        &\geq \gamma^{-1}\cdot (1-\eps) \cdot \hF_{B(j-1)}(b_j).
                    \end{align*}
                    Since $\eps < 1-\gamma$, the above equation is a contradiction to \cref{eq:lem:score_sell:lb_1}.
                    Hence, 
                    \begin{align*}
                        \frac{\Ex_{{\mu}}\insquare{F_{B(j-1)}(b_j)}}{\Ex_{{\mu}}\insquare{F_{B(j-1)}(y_{(j)})}}
                        \geq \beta.
                        \yesnum\label{eq:ratio_lb}
                    \end{align*}

                \paragraph{Step 3 (Completing the proof of the claim):}
                    Since $y_{(j)}\in M$ and $y_{(j)}\not\in B(j-1)$, the definition of $\alpha$ (\cref{def:smooth}) implies that 
                    \[
                        \Ex_{{\mu}}\insquare{F_{B(j-1)}(y_{(j)})} \geq \alpha\tau n.
                        \yesnum\label{eq:lb:prefix_for_y}
                    \]
                    Substituting this in \cref{eq:ratio_lb} gives the following inequality
                    \begin{align*}
                        \Ex_{{\mu}}\insquare{F_{B(j-1)}(b_j)} 
                        &\geq  \beta\cdot \Ex_{{\mu}}\insquare{F_{B(j-1)}(y_{(j)})}.
                        \yesnum\label{eq:lb:prefix_for_y2}
                    \end{align*}
                    The above, as $\evE$ holds and $\beta,(1-\eps)^2\leq 1$, implies the following inequality
                    \begin{align*}
                        F_{B(j-1)}(b_j) 
                        &\ \ \geq\ \   \beta \cdot F_{B(j-1)}(y_{(j)}) - 2\err{}.
                        \yesnum\label{eq:lb:margin_with_b}
                    \end{align*}
                    \cref{eq:lb:prefix_for_y} and the fact that $\evE$ holds, also gives us that
                    \begin{align*}
                        {F_{B(j-1)}(y_{(j)})} %
                        &\geq \alpha\tau n - \err{}\\
                        &\geq \alpha\tau n \cdot \inparen{1 - \frac{1}{\alpha\tau n}\cdot\err{}}\\
                        &\geq \frac{\alpha n\tau}{2}.
                        \tagnum{Using that $n\geq n_0(\eps_0,\delta_0)\geq 4\cdot \alpha^{-2}\cdot k\log{\frac{m}{\delta}}$}
                        \customlabel{eq:lem:score_sell:lb4}{\theequation}
                    \end{align*}
                    Substituting this lower bound in Equation~\eqref{eq:lb:margin_with_b} gives us the following multiplicative lower bound on $F_{B(j-1)}(b_j)$
                    \begin{align*}
                    F_{B(j-1)}(b_j) 
                        &\ \ \Stackrel{\eqref{eq:lem:score_sell:lb4}}{\geq}\ \  \beta\cdot F_{B(j-1)}(y_{(j)})\cdot\inparen{1 - \frac{4}{\beta\alpha n\tau}\cdot \err{}}\\ 
                        &\ \ \Stackrel{}{\geq}\ \  \beta\cdot (1-\eps)\cdot F_{B(j-1)}(y_{(j)}).
                        \tag{Using that $n\geq n(\eps_0,\delta_0)\geq 16\cdot (\alpha\eps\cdot\beta)^{-2}\cdot k\log{\frac{m}{\delta}}$}
                    \end{align*}
                    Replacing $B(j-1), b_j$ and $y_{(j)}$ by $B(\ell)\cup A(i-1)$, $a_i$, and $x_{(i)}$  in Steps 1, 2, and 3 shows that
                    \[
                        F_{B(\ell)\cup A(i-1)}(a_i)\geq \beta\cdot (1-\eps)\cdot F_{B(\ell)\cup A(i-1)}(x_{(i)}).
                    \]
                
                \end{proof}

                \paragraph{Completing the proof of \cref{lem:score_sell}.}
                    Now we are ready to complete the proof of \cref{lem:score_sell}.
                \begin{proof}[Proof of \cref{lem:score_sell}]
                    Suppose the event $\evE$ holds.
                    $F(S)$ can be lower bounded as follows 
                    \begin{align*}
                        F(S) 
                        &=  \sum_{j=1}^{\ell} F_{B(j-1)}(b_j) + \sum_{i=1}^{k-\ell} F_{B(\ell)\cup A(i-1)}(a_i)\\
                        &\geq  \beta\cdot(1-\eps) \cdot\inparen{\sum_{j=1}^{\ell} F_{B(j-1)}(y_{(j)}) + \sum_{i=1}^{k-\ell} F_{B(\ell)\cup A(i-1)}(x_{(i)})}. 
                        \tagnum{Using \cref{lem:lower_bound_on_marginal} and the fact that $\evE$ holds}
                        \customlabel{eq:itnermediate}{\theequation}
                    \end{align*}
                    Next, we lower bound the expected value of each term in the parenthesis in the RHS.
                    Fix any $j\in [\ell]$.
                    Let $c_1,c_2,\dots,c_{j-1}\subseteq B(j-1)$ be a rearrangement of the elements of $B(j-1)$ such that 
                    $\Ex_\mu\insquare{F(c_r)}\geq \Ex_\mu\insquare{F(c_{r+1})}$ for each $r\in [j-2]$.
                    Since $Y(j-1)$ is the set $j-1$ candidates in $G_2$ corresponding to the $j-1$ largest values in $\inbrace{\Ex_\mu\insquare{F(c)}\mid c\in G_2}$, it follows that 
                    \[
                        \forall_{r\in [j-1]},\quad 
                        \Ex_\mu\insquare{F(y_r)}\geq \Ex_\mu\insquare{F(c_r)}.
                        \yesnum\label{eq:score_dominance_y}
                    \]
                    We consider two cases depending on the value of $y_{(j)}\in Y(j)$ to compute a lower bound.

                    \noindent \textit{Case A ($y_{(j)}\not\in Y(j-1)$):}
                        Since $y_{(j)}\in Y(j)$ and $y_{(j)}\not\in Y(j-1)$, in this case $y_{(j)}=y_{j}$.
                        The following lower bound holds
                        
                    \begin{align*}
                        \Ex_\mu\insquare{F_{\inbrace{c_1,c_2,\dots,c_{j-1}}}(y_{(j)})}
                        &= \Ex_\mu\insquare{F_{\inbrace{c_1,c_2,\dots,c_{j-1}}}(y_{j})}
                        \tagnum{Using that, in this case, $y_{(j)}=y_j$}.
                        \customlabel{eq:lowerbound_y_j}{\theequation}
                    \end{align*}
                    Let $R\subseteq T\subseteq C$ be the following sets 
                    \[
                        R\coloneqq \inbrace{c_2,\dots,c_{j-1}}
                        \quad\text{and}\quad 
                        T\coloneqq \inbrace{c_2,\dots,c_{j-1},y_{j}}.
                    \]
                    Substituting these in the above equation implies the following equation 
                    \begin{align*}
                        \Ex_\mu\insquare{F_{\inbrace{c_1,c_2,\dots,c_{j-1}}}(y_{j})}
                        &=
                        \Ex_\mu\insquare{
                            F\inparen{c_1\cup T}
                        }- 
                        \Ex_\mu\insquare{
                            F\inparen{c_1\cup R}
                        }\\ 
                        &=
                        \Ex_\mu\insquare{ F_T\inparen{c_1} }
                        - 
                        \Ex_\mu\insquare{ F_R\inparen{c_1} } 
                        + 
                        \Ex_\mu\insquare{F(T)} - \Ex_\mu\insquare{F(R)}\\ 
                        &\geq 
                        \Ex_\mu\insquare{ F_T\inparen{y_1} }
                        - 
                        \Ex_\mu\insquare{ F_R\inparen{y_1} } 
                        + 
                        \Ex_\mu\insquare{F(T)} - \Ex_\mu\insquare{F(R)}\\
                        \tag{Using $\Ex_\mu\insquare{ F\inparen{y_1} }\geq \Ex_\mu\insquare{ F\inparen{c_1} }$ and the fact that $F$ is order-preserving with respect to $\mu$}\\ 
                        &= 
                        \Ex_\mu\insquare{ F\inparen{y_1\cup T} }
                        - 
                        \Ex_\mu\insquare{ F\inparen{y_1\cup R} }\\
                        &= 
                        \Ex_\mu\insquare{F_{\inbrace{y_1,c_2,\dots,c_{j-1}}}(y_{j})}.
                    \end{align*}
                Similarly, using that $\Ex_\mu\insquare{ F\inparen{y_2} }\geq \Ex_\mu\insquare{ F\inparen{c_2} }$, it follows that 
                    \[
                        \Ex_\mu\insquare{F_{\inbrace{y_1,c_2,\dots,c_{j-1}}}(y_{j})}
                        \geq 
                        \Ex_\mu\insquare{F_{\inbrace{y_1,y_2,\dots,c_{j-1}}}(y_{j})}.
                    \]
                More generally, it holds that for each $r\in [j-2]$
                \[
                        \Ex_\mu\insquare{F_{\inbrace{y_1,\dots,y_{r-1},c_r,c_{r+1},\dots,c_{j-1}}}(y_{j})}
                        \geq 
                        \Ex_\mu\insquare{F_{\inbrace{y_1,\dots,y_{r-1},y_r,c_{r+1},\dots,c_{j-1}}}(y_{j})}.
                \]
                Chaining these $j-2$ inequalities and Equality~\eqref{eq:lowerbound_y_j}, it follows that 
                \[
                    \Ex_\mu\insquare{F_{\inbrace{c_1,c_2,\dots,c_{j-1}}}(y_{(j)})}
                    \ \ \Stackrel{\eqref{eq:lowerbound_y_j}}{\geq}\ \ 
                    \Ex_\mu\insquare{F_{\inbrace{c_1,c_2,\dots,c_{j-1}}}(y_{j})}
                    \geq 
                    \Ex_\mu\insquare{F_{\inbrace{y_1,y_2,\dots,y_{j-1}}}(y_{j})}.
                    \yesnum\label{eq:lower_bound_marginal}
                \]

                    \noindent \textit{Case B ($y_{(j)}\in Y(j-1)$):}
                    Suppose $y_{(j)} = y_s\in Y(j-1)$.
                    In this case, the following lower bound holds 
                    \begin{align*}
                        \Ex_\mu\insquare{F_{\inbrace{c_1,c_2,\dots,c_{j-1}}}(y_{(j)})}
                        &= \Ex_\mu\insquare{F_{\inbrace{c_1,c_2,\dots,c_{j-1}}}(y_s)}\tag{Using that, in this case, $y_{(j)}=y_s$}\\
                        &\geq 
                        \Ex_\mu\insquare{F_{\inbrace{c_1,c_2,\dots,c_{j-1}}}(y_j)} \tag{Using $\Ex_\mu\insquare{ F\inparen{y_s} }\geq \Ex_\mu\insquare{ F\inparen{y_j} }$ and the fact that $F$ is order-preserving with respect to $\mu$}\\ 
                        &
                        \ \ \Stackrel{\geq}{}\ \ 
                        \Ex_\mu\insquare{
                            F_{\inbrace{y_1,y_2,y_3,\dots,y_{j-1}}}\inparen{y_{j}}
                        }. \tag{Using \cref{eq:lower_bound_marginal}}
                    \end{align*}
                    Hence, in either case, the following holds
                    \begin{align*}
                        \Ex_\mu\insquare{F_{\inbrace{c_1,c_2,\dots,c_{j-1}}}(y_{(j)})}
                        \qquad\qquad 
                        \Stackrel{\inbrace{c_1,c_2,\dots,c_{j-1}} = B(j-1)}{=}
                        \qquad\qquad 
                        \Ex_\mu\insquare{F_{B(j-1)}(y_{(j)})}
                        &\ \ \Stackrel{\geq}{}\ \ 
                        \Ex_\mu\insquare{
                            F_{{Y(j-1)}}\inparen{y_{j}}
                        }.
                        \yesnum
                        \label{eq:lb_marginal_1}
                    \end{align*}
                    Replacing $B(j-1)$, $Y(j-1)$, $y_{(j)}$, $G_2$, and $j$ by $B(\ell)\cup A(i-1)$, $X(i-1)$, $x_{(i)}$, $G_1$, and $i$ gives the following lower bound
                    \begin{align*}
                        \Ex_\mu\insquare{F_{B(\ell)\cup A(i-1)}(x_{(i)})}
                        &\ \ \Stackrel{\geq}{}\ \ 
                        \Ex_\mu\insquare{
                            F_{Y(\ell)\cup X(i-1)}\inparen{x_{i}}
                        }.\yesnum
                        \label{eq:lb_marginal_2}
                    \end{align*}  
                    Substituting \cref{eq:lb_marginal_1,eq:lb_marginal_2} in Equation~\eqref{eq:itnermediate}, we get that 
                    \begin{align*}
                        F(S) 
                        &\geq  \beta\cdot (1-\eps) \cdot\inparen{\sum_{j=1}^{\ell} F_{Y(j-1)}(y_{j}) + \sum_{i=1}^{k-\ell} F_{Y(\ell)\cup X(i-1)}(x_{i})}\\ 
                        &=     \beta\cdot (1-\eps) \cdot\inparen{F\inparen{Y(\ell)\cup X(k-\ell)}}\\ 
                        &=  \beta\cdot (1-\eps) \cdot{F\inparen{M}}
                        \tag{Using that $Y(\ell)\cup X(k-\ell)=M$}.
                    \end{align*}
                \end{proof}

    \subsubsection{Proof of \cref{lem:score_m}: Performance Analysis of $M$}
        \begin{proof}
            Let $\evE$ be the following event 
                \begin{align*}
                    \forall_{T=C \text{ or } (T\subseteq C\colon \abs{T}\leq k)},\ \ \forall_{c\in C},\qquad 
                    &\abs{F_T(c) - \Ex_\mu\insquare{F_T(c)}}  \ \hspace{0.5mm}  \leq \  \err{},\\
                    \forall_{T=C \text{ or } (T\subseteq C\colon \abs{T}\leq k)},\ \ \forall_{c\in C},\qquad 
                    &\abs{\hF_T(c) - \Ex_{\wh{\mu}}\insquare{\hF_T(c)}} \leq \err{}.
                \end{align*}
                From \cref{lem:concentration}, it follows that $\Pr[\evE]\geq 1-2\delta$.
                Suppose the event $\evE$ holds.
                Since $M$ is the set of $k$ candidates with the largest values of $\Ex_\mu\insquare{F(\cdot)}$ and $F$ is order preserving with respect to $\mu$, it holds that 
                \[
                    \forall\ {d\in M},\ \ \forall\ {c\not\in M},\ 
                    \quad \Ex_\mu\insquare{F_S(d)}\geq \Ex_\mu\insquare{F_S(c)}.
                \]
                Further, as $\evE$ holds, the following inequality also holds
                \begin{align*}
                    \forall\ {d\in M},\ \ \forall\ {c\not\in M},\ 
                    \quad F_S(d)\geq F_S(c) - 2\err{}.
                    \yesnum\label{eq:lem:score_m}
                \end{align*}
                Suppose $\abs{M\cap S^\star}=a$.
                Since both $M$ and $S^\star$ have size $k$, it holds that $\abs{M\backslash S^\star}=\abs{S^\star\backslash M}=k-a.$
                Let 
                \[
                    M\backslash S^\star\coloneqq \inbrace{d_1,d_2,\dots, d_{k-a}}
                    \quad\text{and}\quad
                    S^\star\backslash M \coloneqq \inbrace{c_1,c_2,\dots, c_{k-a}}.
                \]
                For each $i\in [k-a]$, define 
                \[
                    D(i)\coloneqq\inbrace{d_1,d_2,\dots,d_i}
                    \quad\text{and}\quad
                    C(i)\coloneqq\inbrace{c_1,c_2,\dots,c_i}.
                \]
                For $i=0$, define $D(0)$ and $C(0)$ to be the emptyset.
                Conditioned on $\evE$, it holds that 
                \begin{align*}
                    F(M)  
                    &= F(M\cap S^\star) + F_{M\cap S^\star}(M\backslash S^\star)
                    = F(M\cap S^\star) + F_{M\cap S^\star}\inparen{\inbrace{d_1,d_2,\dots,d_{k-a}}}.
                    \yesnum\label{eq:expression_for_fm}
                \end{align*}
                Next, we lower bound $F_{M\cap S^\star}\inbrace{d_1,d_2,\dots,d_{k-a}})$.
                The following inequality holds
                \begin{align*}
                    & \quad F_{M\cap S^\star}\inparen{\inbrace{d_1,d_2,\dots,d_{k-a}}}\quad \\
                    =&\quad F_{M\cap S^\star}\inparen{\inbrace{d_2,\dots,d_{k-a}}} + F_{(M\cap S^\star)\cup \inbrace{d_2,\dots,d_{k-a}}}\inparen{d_1}\\
                    \Stackrel{\eqref{eq:lem:score_m}}{\geq}&\quad F_{M\cap S^\star}\inparen{\inbrace{d_2,\dots,d_{k-a}}} + F_{(M\cap S^\star)\cup \inbrace{d_2,\dots,d_{k-a}}}\inparen{c_1} -2\err{}\\
                    =&\quad F_{M\cap S^\star}\inparen{\inbrace{c_1,d_2,\dots,d_{k-a}}} -2\err{}.
                \end{align*}
                Similarly, for any $r\in [k-a]$, it holds that
                \begin{align*}
                    & \quad F_{M\cap S^\star}\inparen{\inbrace{c_1,\dots,c_{r-1},d_r,d_{r+1}\dots,d_{k-a}}} \\
                    \geq & \quad
                    F_{M\cap S^\star}\inparen{\inbrace{c_1,\dots,c_{r-1},c_r,d_{r+1}\dots,d_{k-a}}} - 2\err{}.
                \end{align*}
                Chaining these $k-a$ inequalities, we get that 
                \[
                    F_{M\cap S^\star}\inparen{\inbrace{d_1,d_2,\dots,d_{k-a}}} 
                    \geq 
                    F_{M\cap S^\star}\inparen{\inbrace{c_1,c_2,\dots,c_{k-a}}}
                    - 2(k-a)\cdot \err{}.
                \]
                Substituting this in \cref{eq:expression_for_fm} and upper bounding the coefficient of the second term, $2(k-a)$, by $2k$ implies that 
                \begin{align*}
                    F(M) 
                    &\geq F\inparen{M\cap S^\star} + F_{M\cap S^\star}\inparen{\inbrace{c_1,c_2,\dots,c_{k-a}}}
                    - 2k\cdot \err{}\\ 
                    &\geq F\inparen{S^\star}
                    - 2k\cdot \err{}.
                    \yesnum\label{eq:lem:score_m:final}
                \end{align*}
                To convert this to a multiplicative guarantee, we need to lower bound $F(S^\star)$.
                Since $S^\star$ maximizes $F$ among all sets of size at most $k$ and $M$ has size $k$, a lower bound on $F(S^\star)$ is as follows 
                \begin{align*}
                    F(S^\star)
                    &\geq F(M)\\
                    &\geq \sum_{c\in M} F_{C\setminus \inbrace{c}}(c)\\
                    &\geq k\cdot \min_{c\in M} F_{C\setminus \inbrace{c}}(c)\\
                    &\geq k\cdot \min_{c\in M} \Ex_\mu\insquare{F_{C\setminus \inbrace{c}}(c)} - k\err{} \tag{Using that the event $\evE$ holds}\\
                    &\geq kn\cdot \min_{c\in M} \Ex_\mu\insquare{f_{C\setminus \inbrace{c}}(c,\succ)} - k\err{}  \tag{Using the separability of $F$; see \cref{def:score}}\\
                    &\geq \alpha kn\tau - k\err{} \tag{Using the definition of $\alpha$; see \cref{def:smooth}}\\ 
                    &\geq \frac{\alpha kn\tau}{2}  \tagnum{Using that $n\geq n_0(\eps_0,\delta_0)\geq \alpha^{-2}k\log{\frac{m}{\delta}}$}.                    \customlabel{eq:lb_on_fstar}{\theequation}
                \end{align*}
                Using this inequality we can get a multiplicative lower bound on $F(M)$ as follows 
                \begin{align*}
                    F(M) 
                    \ \ 
                    &\Stackrel{\eqref{eq:lem:score_m:final}}{\geq} \ \ 
                    F\inparen{S^\star}\inparen{1 - \frac{2k}{F\inparen{S^\star}}\cdot \err{}}\\
                    &\Stackrel{\eqref{eq:lb_on_fstar}}{\geq} \ \ 
                    F\inparen{S^\star}\inparen{1 - \frac{4k}{\alpha k n\tau}\cdot \err{}}\\
                    &\Stackrel{}{\geq} \ \ 
                    \inparen{1 - \eps}\cdot F\inparen{S^\star}.
                    \tagnum{Using that $n\geq n_0(\eps_0,\delta_0)\geq 16\alpha^{-2}k\log{\frac{m}{\delta}}$}
                \end{align*}
                it follows that $F(S^\star)\geq \min_{c\in M}F(c)$.
                From the definition of $\alpha$, we have that $\min_{c\in M}\Ex_{\mu}\insquare{F(c)}\geq n\alpha\tau$.
                Conditioned on $\evE$, it holds that $\min_{c\in M} {F(c)}\geq \min_{c\in M}\Ex_{\mu}\insquare{F(c)}-\err{}$.
                Chaining the last three inequalities shows that conditioned on $\evE$
                \[
                    F(S^\star) %
                    \geq n\alpha\tau - \err{}
                    \geq \frac{n\alpha\tau}{2}
                    \tagnum{Since $n\geq n(\eps_0,\delta_0)\geq \alpha^{-2}\cdot k\cdot \log{\frac{m}{\delta}}$}.
                    \customlabel{eq:lem:score_m:lb_on_opt}{\theequation}
                \]
                Finally, as $\Pr[\evE]\geq 1-2\delta$, the above inequality implies the desired result:
                \begin{align*}
                    \Pr_{\mu,\wh{\mu}}\insquare{F(M) \geq F(S^\star) \cdot \inparen{1-\eps}} \geq 1-2\delta.
                \end{align*}
        \end{proof}

\subsection{Proof of \cref{cor:algorithmic}: Algorithmic results for the Utility-Based Model}         
\label{sec:proofof:lem:utility_alpha}

We then prove \cref{cor:algorithmic}.
    
\subsubsection{Order-Preservation of a Utility-Based Model}
\label{sec:utility_order_preserve}
    We first show that the utility-based generative models (\cref{ex:utility_generative,def:utility_bias}) are order-preserving with respect to $\mu$ and between $\mu$ and $\wh{\mu}$.

        \paragraph{Order Preservation With Respect to $\mu$.}
            We prove the following lemma.
            \begin{lemma}[\textbf{Order-preserving properties of the latent utility-based model}]\label{lem:utility_order_preserve}
        		Let $F = \sum_{v\in V} f(\cdot ,\succ_v)$ be a multiwinner score function. 
        		Let $\mu$ be a utility-based generative model defined in \cref{ex:utility_generative}. 
        		$F$ is order-preserving with respect to $\mu$. 
        \end{lemma}
        Recall that in the utility-based model (\cref{ex:utility_generative}), the variable $\eta$ is drawn from the uniform distribution on $[0,1]$.
        We will, in fact, prove a more general version of the above lemma that holds for any distribution of $\eta$ that satisfies certain properties (\cref{def:desired_properties_of_eta}).
        With some abuse of notation, we use $\eta$ to denote both the distribution and a value drawn from the distribution $\eta$ (independent of all other randomness).
        \begin{definition}[\textbf{A family of distributions $\eta$}]\label{def:desired_properties_of_eta}
            Let $\eta$ be the distribution on $\R_{\geq 0}$ from \cref{def:utility_bias} that parameterizes the generative model $\mu$.
            Let $\cdf_{\eta}\colon \R\to [0,1]$ be the cumulative distribution function of $\eta$.
            We define the following properties of $\eta$.
            \begin{itemize} %
                \item (Order preserving A) $\eta$ is order-preserving if
                for all $\eps \hspace{-0.5mm}\in \hspace{-0.5mm} [0,1]$ and $a_2\geq a_1\geq 0$, 
                \[
                    \Pr_{X,Y\sim \eta}\insquare{X\hspace{-0.5mm}>\hspace{-0.5mm} Y(1-\eps)\mid  X,Y\hspace{-0.5mm}\in\hspace{-0.5mm} [a_1,a_2]} \geq \frac{1}{2}.
                \]
                \item (Order preserving B) We say $\eta$ is order preserving if %
                for all $0\leq a_1 < a_2 \leq a_3 < a_4$ and all $\eps\in [0,1]$, the following holds:
                \begin{itemize}
                    \sloppy
                    \item Let $\evE$ be the event that either $X\in \insquare{a_3,a_4} \text{ and } Y(1-\eps)\in \insquare{a_1,a_2}$ or $X\in \insquare{a_1,a_2} \text{ and } Y(1-\eps)\in \insquare{a_3,a_4}$
                    \item 
                    $\Pr_{X,Y\sim \eta}\insquare{X\in \insquare{a_3,a_4} \text{ and } Y\in \insquare{a_1,a_2}  \mid \evE} \geq $ \\ $  \Pr_{X,Y\sim \eta}\insquare{X\in \insquare{a_1, a_2} \text{ and } Y(1-\eps)\in \insquare{a_3, a_4}  \mid \evE}.$
                \end{itemize}
            \end{itemize}
        \end{definition} 
            
            \noindent Several distributions on $\R_{\geq 0}$ including the uniform distribution on $[0,1]$ and the exponential distributions satisfy the properties in \cref{def:desired_properties_of_eta}.
            Order preservation properties A and B are used to ensure that if $\omega_c > \omega_{c'}$ (for any $c,c'\in C$ in the same group), then conditioned on the event, say $\evF$, that $\inbrace{c,c'}$ appear in positions $\inbrace{\ell_1,\ell_2}$ (for any $1\leq \ell_1 < \ell_2\leq m$), then $c'$ is more likely to appear in position $\ell_1$ than in $\ell_2$.
            (Note that conditioned on $\evF$, $c$ and $c'$ may appear in positions $\ell_1$ and $\ell_2$ respectively or in positions $\ell_2$ and $\ell_1$ respectively)
        
            Fix the following parameters.
            \begin{enumerate} %
                \item Any multiwinner score function $F$;
                \item Any distribution $\eta$ satisfying the properties in \cref{def:desired_properties_of_eta};
                \item Any pair of candidates $c,c'\in C$ in the same group ($G_1$ or $G_2$); and 
                \item Any subset $S\subseteq C\setminus\inbrace{c,c'}$.
            \end{enumerate}
            To prove that $F$ is order-preserving with respect to $\mu$, we need to show that 
            if $\Exp_{\mu}\left[f(c,\succ) \right] < \Exp_{\mu}\left[f(c',\succ)\right]$, then the following two conditions hold 
            \begin{enumerate} %
                \item (Property A) for any subsets $R\subseteq S\subseteq C\setminus \left\{c,c'\right\}$, \ \ $\Exp_{\mu}\left[f_S(c, \succ)\right] \leq \Exp_{\mu}\left[f_S(c', \succ)\right];$
                \sloppy
                \item (Property B) for any subsets $R\subseteq S\subseteq C\setminus \left\{c,c'\right\}$, \ \ $\Exp_\mu\insquare{f_S(c',\succ)} - \Ex_\mu\insquare{f_S(c,\succ)} \leq \Ex_\mu\insquare{f_R(c',\succ)} - \Ex_\mu\insquare{f_R(c,\succ)}$.
            \end{enumerate}
            \paragraph{Proof of Property A.}
            The following is the main lemma used to prove Property A.
            \begin{lemma}\label{claim:order_preserv_utility}
                For any $d,d'\in C$ in the same group ($G_1$ or $G_2$) and any set $T\subseteq C\setminus\inbrace{c,c'}$, if $\omega_{d'}\geq \omega_{d}$, then the following holds 
                \begin{align*}
                    \Exp_{\mu}\left[f_T(c',\succ)\right]
                        \geq  \Exp_{\mu}\left[f_T(c,\succ)\right]
                    \quad\text{and}\quad 
                    \Exp_{\wh{\mu}}\left[f_T(c',\nsucc)\right]
                        \geq 
                        \Exp_{\wh{\mu}}\left[f_T(c,\nsucc)\right].
                \end{align*}
            \end{lemma}
            Property A straightforwardly follows from the above lemma.
            \begin{proof}[Proof of Property A]
                We claim that $\Exp_{\mu}\left[f(c,\succ) \right] < \Exp_{\mu}\left[f(c',\succ)\right]$, implies that $\omega_c \leq \omega_{c'}$.
                To see this, suppose $\omega_c > \omega_{c'}$, then from \cref{claim:order_preserv_utility} (invoked with $T=\emptyset$) it follows that $\Exp_{\mu}\left[f(c,\succ) \right]\geq \Exp_{\mu}\left[f(c',\succ)\right]$, which is a contradiction.
                Hence, $\omega_c \leq \omega_{c'}$.
                Now, using \cref{claim:order_preserv_utility} with the set $T=S$ and the fact that $\omega_c \leq \omega_{c'}$, Property A follows.
            \end{proof}

            \paragraph{Proof of Property B.}
            We first derive an equivalent version of Property B that is easier to prove.

            \noindent \textit{Step 1 (An alternate version of Property B):}
            By rearranging the terms in the inequality in Property B and using the definition of the marginal score, we get the following equivalent version of Property B
            \[
                \forall_{R\subseteq S\subseteq C\setminus \left\{c,c'\right\}}, \quad 
                \Exp_\mu\insquare{f(S\cup c',\succ)} - \Ex_\mu\insquare{f(R\cup c',\succ)} 
                    \leq 
                    \Ex_\mu\insquare{f(S\cup c,\succ)} - \Ex_\mu\insquare{f(R\cup c,\succ)}.
            \]
            Using the definition of the marginal score again, we get the following equivalent version of Property B
            \[
                \forall_{R\subseteq S\subseteq C\setminus \left\{c,c'\right\}}, \quad 
                \Exp_\mu\insquare{f_{R\cup c'}(S\setminus R,\succ)}
                    \leq 
                    \Ex_\mu\insquare{f_{R\cup c}(S\setminus R,\succ)}.
            \]
            Observe that the set $S\setminus R$ is disjoint from $R\cup c'$ and $R\cup c$.
            Defining $A\coloneqq R$ and $B\coloneqq S\setminus R$, we get the following equivalent version
            \[
                \forall_{A,B\subseteq C\setminus \left\{c,c'\right\}\colon A\cap B = \emptyset}, \quad 
                \Exp_\mu\insquare{f_{A\cup c'}(B,\succ)}
                    \leq 
                    \Ex_\mu\insquare{f_{A\cup c}(B,\succ)}.
                \yesnum\label{eq:equiv_version}
            \]
            We can prove the above inequality, however, the fact that $B$ can have multiple candidates makes the proof unnecessarily technical.
            Instead, we will simplify the above version of Property B using the following lemma.
            \begin{lemma}\label{lem:simplify}
                For any $d,d'\in C$, the following holds 
                \begin{align*}
                    \forall_{X,Y\subseteq C\setminus \left\{c,c'\right\}\colon X\cap Y = \emptyset}, \quad 
                    &\Exp_\mu\insquare{f_{X\cup c'}(Y,\succ)}
                        \leq 
                        \Ex_\mu\insquare{f_{X\cup c}(Y,\succ)},\\ 
                    \iff\quad 
                    \forall_{Y\subseteq C\setminus \left\{c,c'\right\}}\ \ 
                    \forall_{y\in C\setminus X\setminus\inbrace{c,c'}}, \quad 
                    &\Exp_\mu\insquare{f_{X\cup c'}(y,\succ)}
                        \leq 
                        \Ex_\mu\insquare{f_{X\cup c}(y,\succ)}.
                \end{align*}
            \end{lemma}
            \begin{proof}
                We are required to prove $(1)$ and equivalent $(2)$, where (1) and (2) are the following conditions respectively
                \begin{align*}
                    \forall_{X,Y\subseteq C\setminus \left\{c,c'\right\}\colon X\cap Y = \emptyset}, \quad 
                    &\Exp_\mu\insquare{f_{X\cup c'}(Y,\succ)}
                        \leq 
                        \Ex_\mu\insquare{f_{X\cup c}(Y,\succ)},\yesnum\label{eq:version1}\\ 
                    \forall_{Y\subseteq C\setminus \left\{c,c'\right\}}\ \ 
                    \forall_{y\in C\setminus X\setminus\inbrace{c,c'}}, \quad 
                    &\Exp_\mu\insquare{f_{X\cup c'}(y,\succ)}
                        \leq 
                        \Ex_\mu\insquare{f_{X\cup c}(y,\succ)}.\yesnum\label{eq:version2}
                \end{align*}

                \paragraph{Step A ($2\implies 1$):}
                    Fix any disjoint $X,Y\subseteq C\setminus \left\{c,c'\right\}$.
                    Let $Y\coloneqq \inbrace{y_1,y_2,\dots,y_t}$.
                    It holds that 
                    \begin{align*}
                        \Exp_\mu\insquare{f_{X\cup c'}(Y,\succ)}
                        &= \sum_{i\in [t]}\Exp_\mu\insquare{f_{X\cup c'\cup y_1\cup y_2\cup\dots\cup y_{i-1}}(y_i,\succ)}\\ 
                        &\leq \sum_{i\in [t]}\Exp_\mu\insquare{f_{X\cup c\cup y_1\cup y_2\cup\dots\cup y_{i-1}}(y_i,\succ)} \tag{Using \cref{eq:version1} to switch $c'$ with $c$}\\
                        &\leq \Exp_\mu\insquare{f_{X\cup c}(Y,\succ)}. \tag{Using that $F_S(T\cup R)=F_S(T)+F_{S\cup T}(R)$ for all sets $S,R,T$ and any submodular function $F$}
                    \end{align*}
                
                \paragraph{Step B ($1\implies 2$):}
                    \cref{eq:version2} follows by setting $Y=\inbrace{y}$ in \cref{eq:version1}.
            \end{proof}
            \noindent \textit{Step 2 (Proving alternate version of Property B):} Thus, from \cref{eq:equiv_version} and \cref{lem:simplify}, it follows that the following condition is equivalent to Property B.
            \[
                \forall_{A\subseteq C\setminus \left\{c,c'\right\}}\ \ 
                    \forall_{b\in C\setminus X\setminus\inbrace{c,c'}}, \quad 
                    \Exp_\mu\insquare{f_{A\cup c'}(b,\succ)}
                        \leq 
                        \Ex_\mu\insquare{f_{A\cup c}(b,\succ)}.
                    \yesnum\label{eq:final_equival_version}
            \]
            \paragraph{Notation.}
            We first define some notation that is used to express $\Exp_\mu\insquare{f_{A\cup c'}(b,\succ)}$ and $\Ex_\mu\insquare{f_{A\cup c}(b,\succ)}$.
            Fix any voter $v$ and consider $\succ_v\coloneqq \succ$ where $\succ\ \sim \mu$.
            Define $w_{(i)}$ as the $i$-th largest order statistic among the utilities of candidates in $C\setminus\inbrace{c,c'}$ for each $i\in [m-1]$.
            In other words, $w_{(i)}$ is the $i$-th largest value in the set 
            \[
                \cW\coloneqq \inbrace{ w_{v,d} \mid d\in C\setminus\inbrace{c,c'}}.
            \]
            By this definition, we have the inequalities
            \[
                w_{(1)} \geq w_{(2)} \geq \dots \geq w_{(m)}.
            \]
            Define the interval $I_\ell\coloneqq \insquare{w_{(\ell)}, w_{(\ell+1)}}$ for each $\ell\in [m-2]$.
            Define $\evE_{\ell,k}$, for each $\ell,k\in [m-2]$ as the event that 
            \[
                w_{v,c'} \in I_\ell \quad \text{and} \quad w_{v,c} \in I_k. \tag{Event $\evE_{\ell,k}$}
            \]
            Define $\evF$ as the event that 
            \[
                w_{v,c'}>w_{v,c}. \tag{Event $\evF$}
            \] 
            Note that $\evF$ is equivalent to the event $\pos_{\succ}(c') < \pos_{\succ}(c)$.
            Finally, for each $\ell\in [m]$, define the random variable 
            \[
                \tau_{\pos_\succ(b)}(\pos_\succ(A),\ell) \coloneqq f_{A\cup i(\ell)}(b,\succ).
            \]
            Where $i(\ell)$ is the candidate at the $\ell$-th position in $\succ$.
            The randomness in $\tau_S(\ell)$ due to the randomness in $\pos_\succ(A)$ and in $\pos_\succ(b)$.
            Conditioned on any value of $\pos_\succ(A)$ and $\pos_\succ(b)$, due to domination sensitivity (\cref{def:score}), it holds that 
            \[
                \forall 1\leq \ell < k \leq m,\quad \tau_{A\cup i(\ell)}(b) \geq \tau_{A\cup i(k)}(b).
                \yesnum\label{eq:inequality2}
            \]
            We can express $\Exp_\mu\insquare{f_{A\cup c'}(b,\succ)}$ and $\Ex_\mu\insquare{f_{A\cup c}(b,\succ)}$ as follows
            \begin{align*}
                \Exp_\mu\insquare{f_{A\cup c'}(b,\succ)}
                    &= \Ex_{\cW}\insquare{ \sum_{\ell,m\in [m-2]:\ell < k}\inparen{ \Pr\insquare{\evE_{\ell, k}} \cdot \tau_{A\cup i(\ell)}(b)  + \Pr\insquare{\evE_{k,\ell}}\cdot \tau_{A\cup i(k)}(b)  } }\\ 
                    &\quad + \Ex_{\cW}\insquare{ \sum_{\ell\in [m-2]}\Pr\insquare{\evE_{\ell, \ell}} \inparen{ \Pr\insquare{\evF\mid \evE_{\ell,\ell} }\cdot \tau_{A\cup i(\ell)}(b) + \Pr\insquare{\lnot\evF\mid \evE_{\ell,\ell} }\cdot \tau_{A\cup i(\ell+1)}(b) } },\\
                \Ex_\mu\insquare{f_{A\cup c}(b,\succ)}
                    &= \Ex_{\cW}\insquare{ \sum_{\ell,m\in [m-2]:\ell < k}\inparen{ \Pr\insquare{\evE_{\ell, k}} \cdot \tau_{A\cup i(k)}(b)  + \Pr\insquare{\evE_{k,\ell}} \cdot \tau_{A\cup i(\ell)}(b) } }\\ 
                    &\quad + \Ex_{\cW}\insquare{ \sum_{\ell\in [m-2]}\Pr\insquare{\evE_{\ell, \ell}} \inparen{ \Pr\insquare{\evF\mid \evE_{\ell,\ell} }\cdot \tau_{A\cup i(\ell+1)}(b)  + \Pr\insquare{\lnot\evF\mid \evE_{\ell,\ell} } \cdot \tau_{A\cup i(\ell)}(b) } }.
            \end{align*}
            Hence, it follows that 
            \begin{align*}
                &\Exp_\mu\insquare{f_{A\cup c'}(b,\succ)}  - \Ex_\mu\insquare{f_{A\cup c}(b,\succ)}\\
                &\qquad\geq 
                \Ex_{\cW}\insquare{ 
                        \sum_{\ell,m\in [m-2]:\ell < k}\inparen{ 
                            \Pr\insquare{\evE_{\ell, k}} - \Pr\insquare{\evE_{k, \ell}}
                        }
                        \cdot \inparen{
                            \tau_{A\cup i(\ell)}(b) - \tau_{A\cup i(k)}(b)
                        }
                }\\ 
                &\qquad + \Ex_{\cW}\insquare{ 
                        \sum_{\ell\in [m-2]}
                        \Pr\insquare{\evE_{\ell, \ell}}\cdot \inparen{ 
                            \Pr\insquare{\evF\mid \evE_{\ell,\ell} } - \Pr\insquare{\lnot\evF\mid \evE_{\ell,\ell} }}
                            \cdot \inparen{
                                \tau_{A\cup i(\ell)}(b) - \tau_{A\cup i(\ell+1)}(b)
                            }
                }.
                \yesnum\label{eq:diff2}
            \end{align*}
            Note that $\pos(A)$ and $\pos(b)$ are deterministic conditioned on any specific value of $\cW$ and either (1) the event $\evE_{\ell,k}$ for $\ell\neq k$ or (2)  the event $\evE_{\ell,\ell}$ and $\evF$.
            Hence, \cref{eq:inequality2} holds.
            Further, we have the following claims using the two properties in \cref{def:desired_properties_of_eta}.
            \begin{align*}
                \forall_{1\leq \ell < k \leq m},\qquad \qquad 
                  \white{.}\quad  \Pr\insquare{\evE_{\ell, k}} - \Pr\insquare{\evE_{k, \ell}} &\geq 0,\yesnum\label{eq:claim3}\\ 
                \forall_{1\leq \ell \leq m},\quad 
                    \Pr\insquare{\evF\mid \evE_{\ell,\ell} } - \Pr\insquare{\lnot\evF\mid \evE_{\ell,\ell} } &\geq 0.\yesnum\label{eq:claim4}
            \end{align*}
            Substituting \cref{eq:inequality2,eq:claim3,eq:claim4} in \cref{eq:diff2}, it follows that 
            \begin{align*}
                \Exp_\mu\insquare{f_{A\cup c'}(b,\succ)}  - \Ex_\mu\insquare{f_{A\cup c}(b,\succ)} \leq 0.
            \end{align*}
            The above is equivalent to \cref{eq:final_equival_version}.
            Since \cref{eq:final_equival_version}, itself, is equivalent to Property B, Property B also follows.

        \paragraph{Order preservation between $\mu$ and $\wh{\mu}$.}
            We define a parameter $\sigma(f)$ used in \Cref{lem:utility_order_preserve2}.
            
            \begin{definition}
                \label{def:sigma_f}
                Let $F = \sum_{v\in V} f(\cdot ,\succ_v)$ be a multiwinner score function.
                For each $j\in [m]$, let $i_{\succ,j}$ be the $j$-th candidate in $\succ$.
                For each $j\in [m]$ and $S\subseteq S$, define $\tau_{j,S}(f)$ to be the marginal score of $i_{\succ,j}$ with respect to $S$, i.e., 
                $$\tau_{j, S}(f)\coloneqq f_S(\inbrace{i_{\succ,j}},\succ)$$ 
                which is independent of $\succ\ \in \prefs{C}$.
                Let $\tau_{m+1,S}=0$ for any $S\subseteq C$ and $\tau_{\min}\coloneqq \min_{S\subseteq C}\tau_{m-1,S}.$
                $\sigma(f)$ is defined as follows.
                \begin{align*}
                    \sigma(f)\coloneqq \min_{S\subseteq C}\min_{\ell\in [m-1]} \frac{\tau_{\ell, S}(f)+\tau_{\ell+2,S}(f)-2\tau_{\ell+1,S}(f)}{\tau_{\ell, S}(f) - \tau_{\ell+2,S}(f)},
                    \yesnum
                \end{align*}
                where we use the convention $\frac{0}{0}=1$.
            \end{definition}
            Intuitively, $\sigma(f)$ measures a certain notion of convexity measure of $f$.
            Concretely, $\sigma(f)\geq 0$ if and only if there is a convex function $g$ that takes value $s_\ell$ at $\ell$ (for each $\ell\in [m+1]$). 
            Moreover, strict inequality holds (i.e., $\sigma(f)>0$) if and only if $g$ is strictly convex.
            The normalization in $\sigma(f)$ ensures that it is invariant to a multiplicative scaling of $F$.
            As for examples: $\sigma(f) = 1$ for SNTV and $\sigma(f)=0$ for Borda rule or $\ell_1$-CC.

            \begin{lemma}[\textbf{Order-preservation between latent and biased utility-based models}]\label{lem:utility_order_preserve2}
        		Let $F = \sum_{v\in V} f(\cdot ,\succ_v)$ be a multiwinner score function. 
        		Let $(\mu,\wh{\mu})$ be a utility-based generative model defined in \cref{ex:utility_generative,def:utility_bias}. 
        		The following holds for any $0\leq \lambda \leq m^{-1/2}$.
        		\begin{enumerate} %
        		    \item If $\sigma(f)>0$, then $F$ is $\inparen{1-\lambda, 1 - \sigma(f) \cdot m^{-1} \cdot\Omega\inparen{\lambda}}$ order preserving between $\mu$ and $\wh{\mu}$; and 
        		    \item If $\sigma(f)=0$ and $\tau_{\min}>0$ and $\tau_{m,\emptyset}=0$, then $F$ is $\inparen{1-\lambda,1- m^{-2}\cdot \frac{\tau_{\min}}{\tau}\cdot \Omega\inparen{\lambda}}$ order preserving between $\mu$ and $\wh{\mu}$.
        		\end{enumerate}
        	\end{lemma}
            In particular, this result implies the following: %
            \begin{enumerate}
                \item The SNTV rule is $\inparen{1-O(m^{-1/2}), 1-\Omega\inparen{m^{-3/2}}}$ order preserving between $\mu$ and $\wh{\mu}$;
                \item The $\ell_1$-CC rule is $\inparen{1-O(m^{-1/2}), 1-\Omega\inparen{m^{-3/2}}}$ order preserving between $\mu$ and $\wh{\mu}$; and
                \item The Borda rule is $\inparen{1-O(m^{-1/2}), 1-\Omega\inparen{m^{-3.5}}}$ order preserving between $\mu$ and $\wh{\mu}$.
            \end{enumerate}
            Recall that in the utility-based model (\cref{ex:utility_generative}), the variable $\eta$ is drawn from the uniform distribution on $[0,1]$.
            We will, in fact, prove a more general version of the above lemma that holds for any distribution of $\eta$ that satisfies certain properties (\cref{def:desired_properties_of_eta2}).
            With some abuse of notation, we use $\eta$ to denote both the distribution and a value drawn from the distribution $\eta$ (independent of all other randomness).

            We will bound the parameters $(\beta,\gamma)$ for any $\eta$ from \cref{def:utility_bias} that satisfies the following properties.
            \begin{definition}[\textbf{Properties of the utility distribution $\eta$ from \cref{def:utility_bias}}]\label{def:desired_properties_of_eta2}
                Let $\eta$ be the distribution on $\R_{\geq 0}$ from \cref{def:utility_bias} that parameterizes the generative model $\mu$.
                Let $\cdf_{\eta}\colon \R\to [0,1]$ be the cumulative distribution function of $\eta$.
                We define the following properties of $\eta$.
                \begin{itemize}
                    \item (Log-Lipshictzness) We say that $\cdf_{\eta}$ is Log-Lipschitz there exists a constant $\pi>0$ such that, 
                    for all $0 < x < y$, 
                        $\frac{\cdf_{\eta}(x)}{\cdf_{\eta}(y)}\geq 1-\frac{\pi x}{y}$; and
                    \item (Order preservation) We say $\eta$ is order-preserving if there exists a constant $\pi>0$ such that,  
                    for all $\eps\in [0,1]$ and $t\geq 0$, if $\Pr_{X,Y\sim \eta}\insquare{X>Y(1-\eps)\mid X,Y\geq t} \geq \frac{1}{2}\cdot \inparen{1+\frac{\eps}{\pi}}$.
                \end{itemize}
            \end{definition}
            \noindent Several distributions on $\R_{\geq 0}$ including the uniform distribution on $[0,1]$ and the exponential distributions satisfy the properties in \cref{def:desired_properties_of_eta2}.
            Roughly, log-lipshitzness requires that the $\log\inparen{\cdf_\eta(\cdot)}$ is Lipshictz.
            It guarantees that if $x$ and $y$ are multiplicatively close to each other, then $\cdf_\eta(x)$ and $\cdf_\eta(y)$ are multiplicatively close to each other.
            Order preservation guarantees that if $\omega_c > (1+\eps)\cdot\omega_{c'}$ (for any $c,c'\in C$ and $\eps\in [0,1]$), then $\Pr\left[w_{v,c} > w_{v,c'}\right] > 0.5\cdot (1+\Omega(\eps))$.
            Both of these guarantees are required to establish the multiplicative guarantees in $(\beta,\gamma)$ order preservation.

            To prove that $F$ is $(\beta,\gamma)$ order preserving between $\mu$ and $\wh{\mu}$, we need to show that the following implication holds
                \[
                    \beta\cdot \Exp_{\mu}\left[f_S(c',\succ)\right] >  \Exp_{\mu}\left[f_S(c,\succ) \right] > 0.
                    \implies 
                    \gamma\cdot \Exp_{\widehat{\mu}}\left[f_S(c', \nsucc)\right] \geq \Exp_{\widehat{\mu}}\left[f_S(c, \nsucc)\right].
                    \yesnum\label{eq:precondition_order}
                \]
            We divide the proof of \cref{lem:utility_order_preserve2} into two parts, corresponding to the two conditions in \cref{lem:utility_order_preserve2}.
            Both parts rely on the following lemma, whose proof appears later.
            \begin{lemma}\label{lem:ratio_of_utils}
                Let $F = \sum_{v\in V} f(\cdot ,\succ_v)$ be a multiwinner score function. %
        		Let $(\mu,\wh{\mu})$ be a utility-based generative model defined in \cref{ex:utility_generative,def:utility_bias}. 
                For any $\beta\in [1-m^{-1/2},1]$, any candidates $c,c'\in C$ in the same group ($G_1$ or $G_2$), and any set $S\subseteq C\setminus\inbrace{c,c'}$, the following implication holds
                \[
                    \beta\cdot \Ex_\mu\insquare{f_S(c',\succ)} \geq  \Ex_\mu\insquare{f_S(c,\succ)} 
                    \implies 
                    \omega_{c'} \inparen{1-m^{-1}\cdot \Theta\inparen{1-\beta}} > \omega_c.
                \]
            \end{lemma}

            \paragraph{Part 1 ($(\beta,\gamma)$ order preservation with $\sigma(f)>0$):}
                Suppose $\sigma(f)>0$.
                The following is the main lemma in this part.
                \begin{lemma}\label{lem:part_1_order}
                    Let $F = \sum_{v\in V} f(\cdot ,\succ_v)$ be a multiwinner score function. %
            		Let $(\mu,\wh{\mu})$ be a utility-based generative model defined in \cref{ex:utility_generative,def:utility_bias}. 
                    Suppose candidates $c,c'\in C$ are in the same group ($G_1$ or $G_2$) and set $S\subseteq C\setminus\inbrace{c,c'}$.
                    If there exists a $\rho>0$ such that $\omega_{c'}(1-\rho)\geq \omega_c$ and $\Ex_\mu\insquare{f_S(c',\succ)},\Ex_\mu\insquare{f_S(c,\succ)}>0$, then there is a constant $\eps>0$ such that
                    \[
                        \frac{\Ex_{\wh{\mu}}\insquare{f_S(c',\nsucc)}}{\Ex_{\wh{\mu}}\insquare{f_S(c,\nsucc)}} \geq \frac{1+\eps\rho\sigma(f)}{1-\eps\rho\sigma(f)}.
                    \]
                \end{lemma}
                $(\beta,\gamma)$ order preservation follows from \cref{lem:ratio_of_utils,lem:part_1_order}.
                Concretely, the proof is as follows.
        
                \smallskip
        
                \noindent\textit{Proof of $(\beta,\gamma)$ order preservation assuming \cref{lem:ratio_of_utils,lem:part_1_order}.} 
                    From the LHS in \cref{eq:precondition_order} and \cref{lem:ratio_of_utils}, it follows that $\omega_{c'} \inparen{1-m^{-1}\cdot \Theta\inparen{1-\beta}} > \omega_c$.
                    Hence, \cref{lem:part_1_order} applicable with $\rho=m^{-1}\cdot \Theta\inparen{1-\beta}$.
                    From \cref{lem:part_1_order}, it follows that the RHS of \cref{eq:precondition_order} holds with $\gamma = 1-\Theta(m^{-1}\cdot \Theta\inparen{1-\beta}\cdot \sigma(f))$

                In the remainder of this section, we prove \cref{lem:part_1_order}.
                
                \begin{proof}[Proof of \cref{lem:part_1_order}]
                    Fix any voter $v$ and consider $\succ_v\coloneqq \succ$ where $\succ\ \sim \mu$.
                    Define $w_{(i)}$ as the $i$-th largest order statistic among the utilities of candidates in $C\setminus\inbrace{c,c'}$ for each $i\in [m-1]$.
                    In other words, $w_{(i)}$ is the $i$-th largest value in the set 
                    \[
                        \cW\coloneqq \inbrace{ w_{v,d} \mid d\in C\setminus\inbrace{c,c'}}.
                    \]
                    By this definition, we have the inequalities
                    \[
                        w_{(1)} \geq w_{(2)} \geq \dots \geq w_{(m)}.
                    \]
                    Define the interval $I_\ell\coloneqq \left[w_{(\ell)}, \infty\right)$ for each $\ell\in [m-2]$.
                    Define $\evE_{\ell}$, for each $\ell\in [m-2]$, as the event that 
                    \[
                        w_{v,c'} \in I_\ell. %
                        \tag{Event $\evE_{\ell}$}
                    \]
                    Define $\evF$ as the event that 
                    \[
                        w_{v,c'}>w_{v,c}. \tag{Event $\evF$}
                    \] 
                    Note that $\evF$ is equivalent to the event $\pos_{\succ}(c') < \pos_{\succ}(c)$.
                    For each $j\in [m]$, let $i_{\succ,j}$ be the $j$-th candidate in $\succ$.
                    Finally, for each $\ell\in [m]$, define the random variable 
                    \[
                        \tau_S(\ell) \coloneqq f_S(i_{\succ,j},\succ).
                    \]
                    The randomness in $\tau_S(\ell)$ due to the randomness in $\pos(S)$.
                    Conditioned on any value of $\pos(S)$, due to domination sensitivity (\cref{def:score}), it holds that 
                    \[
                        \forall 1\leq \ell < k \leq m,\quad \tau_S(\ell) \geq \tau_S(k).
                        \yesnum\label{eq:inequality3}
                    \]
                    Using this notation, we can express $\Exp_{\mu}\left[f_S(c,\succ) \right]$ and $\Exp_{\mu}\left[f_S(c',\succ) \right]$ as follows 
                    \begin{align*}
                        &\frac{\Exp_{\mu}\left[f_S(c',\succ) \right]}{\Exp_{\mu}\left[f_S(c,\succ) \right]}\\
                        = &  \frac{
                                \Ex_{\cW}\insquare{\sum_{\ell\in [m-2]} \Pr\insquare{\evE_\ell}
                                \cdot \inparen{  \Pr\insquare{ \evF \mid \evE_\ell}
                                            \cdot \inparen{\tau_{S}(\ell)- \tau_{S}(\ell+1)}
                                            + \Pr\insquare{\lnot \evF \mid \evE_\ell}
                                            \cdot \inparen{\tau_{S}(\ell+1)- \tau_{S}(\ell+2)}
                                    }   
                                }  
                            }{
                                \Ex_{\cW}\insquare{\sum_{\ell\in [m-2]} \Pr\insquare{\evE_\ell}
                                \cdot \inparen{  \Pr\insquare{ \evF \mid \evE_\ell}
                                            \cdot \inparen{\tau_{S}(\ell+1)- \tau_{S}(\ell+2)}
                                            + \Pr\insquare{\lnot \evF \mid \evE_\ell}
                                            \cdot \inparen{\tau_{S}(\ell)- \tau_{S}(\ell+1)}
                                    }   
                                } 
                            }\\  
                        = & \frac{
                                \Ex_{\cW}\insquare{\sum_{\ell\in [m-2]} \Pr\insquare{\evE_\ell} \cdot 
                                \inparen{\tau_{S}(\ell)- \tau_{S}(\ell+2)}
                                \cdot \inparen{  
                                            1+\eps\rho \cdot {  \frac{\tau_{\ell, S}(f)+\tau_{\ell+2,S}(f)-2\tau_{\ell+1,S}(f)}{\tau_{\ell, S}(f) - \tau_{\ell+2,S}(f)}}
                                    }   
                                }  
                            }{
                                \Ex_{\cW}\insquare{\sum_{\ell\in [m-2]} \Pr\insquare{\evE_\ell} \cdot 
                                \inparen{\tau_{S}(\ell)- \tau_{S}(\ell+2)}
                                \cdot \inparen{  
                                            1-\eps\rho \cdot {  \frac{\tau_{\ell, S}(f)+\tau_{\ell+2,S}(f)-2\tau_{\ell+1,S}(f)}{\tau_{\ell, S}(f) - \tau_{\ell+2,S}(f)}}
                                    }   
                                }  
                            }\yesnum\label{eq:expression_complicated}\\ 
                        = &   \min_{T\subseteq C\colon \abs{T}=\abs{S}}\min_{\ell\in [m]}
                            \frac{
                                1+\eps\rho \cdot {  \frac{\tau_{\ell, S}(f)+\tau_{\ell+2,S}(f)-2\tau_{\ell+1,S}(f)}{\tau_{\ell, S}(f) - \tau_{\ell+2,S}(f)}}
                            }{
                                1-\eps\rho \cdot {  \frac{\tau_{\ell, S}(f)+\tau_{\ell+2,S}(f)-2\tau_{\ell+1,S}(f)}{\tau_{\ell, S}(f) - \tau_{\ell+2,S}(f)}}
                            }\tag{Using that $\tau_{S}(\ell)- \tau_{S}(\ell+2) \geq 0$ for all $\ell$}\\ 
                        \geq &   \frac{1+\eps\rho\sigma(f)}{1-\eps\rho\sigma(f)}. \tag{Using \cref{def:sigma_f} and that $\eps,\rho,\sigma(f)\geq 0$}
                    \end{align*}
                \end{proof}

            \paragraph{Part 2 ($(\beta,\gamma)$ order preservation with $\sigma(f)=0$):}
                Suppose $\sigma(f)=0$ and $\min_{S\subseteq C}\tau_{m-1,S}>0$.
                The following is the main lemma in this part.
                \begin{lemma}\label{lem:part_2_order}
                    Let $F = \sum_{v\in V} f(\cdot ,\succ_v)$ be a multiwinner score function. %
            		Let $(\mu,\wh{\mu})$ be a utility-based generative model defined in \cref{ex:utility_generative,def:utility_bias}. 
                    Suppose candidates $c,c'\in C$ are in the same group ($G_1$ or $G_2$) and set $S\subseteq C\setminus\inbrace{c,c'}$.
                    If there exists a $\rho>0$ such that $\omega_{c'}(1-\rho)\geq \omega_c$ and $\Ex_\mu\insquare{f_S(c',\succ)},\Ex_\mu\insquare{f_S(c,\succ)}>0$, then %
                    \[
                        \frac{\Ex_{\wh{\mu}}\insquare{f_S(c',\nsucc)}}{\Ex_{\wh{\mu}}\insquare{f_S(c,\nsucc)}} \geq 1+\rho\eps \cdot m^{-1}\cdot \frac{\tau_{\min}}{\tau}\cdot \Omega(1).
                    \]
                \end{lemma}
                $(\beta,\gamma)$ order preservation follows from \cref{lem:ratio_of_utils,lem:part_2_order}.
                Concretely, the proof is as follows.

                \smallskip 
                
                \noindent\textit{Proof of $(\beta,\gamma)$ order preservation assuming \cref{lem:ratio_of_utils,lem:part_1_order}.} 
                    From the LHS in \cref{eq:precondition_order} and \cref{lem:ratio_of_utils}, it follows that $\omega_{c'} \inparen{1-m^{-1}\cdot \Theta\inparen{1-\beta}} > \omega_c$.
                    Hence, \cref{lem:part_2_order} applicable with $\rho=m^{-1}\cdot \Theta\inparen{1-\beta}$.
                    From \cref{lem:part_2_order}, it follows that the RHS of \cref{eq:precondition_order} holds with $\gamma = 1+ m^{-2}\cdot \frac{\tau_{\min}}{\tau}\cdot \Theta(1-\beta)$.
                
                In the remainder of this section, we prove \cref{lem:part_2_order}.
        
                \begin{proof}[Proof of \cref{lem:part_2_order}]
                    We borrow notation from \cref{lem:part_1_order}.
                    In addition, for each $\ell\in [m-1]$, define
                    \[
                        \sigma_\ell(f) \coloneqq \frac{\tau_{\ell, S}(f)+\tau_{\ell+2,S}(f)-2\tau_{\ell+1,S}(f)}{\tau_{\ell, S}(f) - \tau_{\ell+2,S}(f)}.
                        \yesnum\label{def:sigma_ell}
                    \]
                    Note that $\sigma(f)=\min_{\ell\in [m-1]} \sigma_\ell(f)$.
                    Since $\sigma(f)=0$, it follows that for each $\ell\in [m-1]$,
                    \[
                        \sigma_\ell(f)\geq 0.
                        \yesnum\label{eq:nonnegative_sigma}
                    \]
                    Moreover, since $\tau_{m-1,S}\geq \tau_{\min}>0$ (for all $S\subseteq C\setminus i_{m-1}$) and $\tau_{m,S}=\tau_{m+1,S}=0$  (for all $S\subseteq C$) it follows that 
                    \[
                        \sigma_{m-1}(f)=1.\yesnum\label{eq:one_sigma}
                    \]
                    Now we are ready to lower bound $\frac{\Exp_{\mu}\left[f_S(c',\succ) \right]}{\Exp_{\mu}\left[f_S(c,\succ) \right]}$:
                    \cref{eq:expression_complicated} shows that 
                    \begin{align*}
                        &\frac{\Exp_{\mu}\left[f_S(c',\succ) \right]}{\Exp_{\mu}\left[f_S(c,\succ) \right]}\\
                        &\quad \geq \ \ \frac{
                                \Ex_{\cW}\insquare{\sum_{\ell\in [m-1]} \Pr\insquare{\evE_\ell} \cdot 
                                \inparen{\tau_{S}(\ell)- \tau_{S}(\ell+2)}
                                \cdot \inparen{  
                                            1+\eps\rho \cdot {  \frac{\tau_{\ell, S}(f)+\tau_{\ell+2,S}(f)-2\tau_{\ell+1,S}(f)}{\tau_{\ell, S}(f) - \tau_{\ell+2,S}(f)}}
                                    }   
                                }  
                            }{
                                \Ex_{\cW}\insquare{\sum_{\ell\in [m-1]} \Pr\insquare{\evE_\ell} \cdot 
                                \inparen{\tau_{S}(\ell)- \tau_{S}(\ell+2)}
                                \cdot \inparen{  
                                            1-\eps\rho \cdot {  \frac{\tau_{\ell, S}(f)+\tau_{\ell+2,S}(f)-2\tau_{\ell+1,S}(f)}{\tau_{\ell, S}(f) - \tau_{\ell+2,S}(f)}}
                                    }   
                                }  
                            }\\
                            &\quad \geq \ \ \frac{
                                \Ex_{\cW}\insquare{\sum_{\ell\in [m-1]} \Pr\insquare{\evE_\ell} \cdot 
                                \inparen{\tau_{S}(\ell)- \tau_{S}(\ell+2)}
                                \cdot \inparen{  
                                            1+\eps\rho \cdot \sigma_\ell(f)
                                    }   
                                }  
                            }{
                                \Ex_{\cW}\insquare{\sum_{\ell\in [m-1]} \Pr\insquare{\evE_\ell} \cdot 
                                \inparen{\tau_{S}(\ell)- \tau_{S}(\ell+2)}
                                \cdot \inparen{  
                                            1-\eps\rho \cdot \sigma_\ell(f)
                                    }   
                                }  
                            }\tag{Using \cref{def:sigma_ell}}\\
                            &\quad \geq \ \ \frac{
                                \Ex_{\cW}\insquare{\sum_{\ell\in [m-2]} \Pr\insquare{\evE_\ell} \cdot 
                                \inparen{\tau_{S}(\ell)- \tau_{S}(\ell+2)} 
                                + \Pr\insquare{\evE_{m-1}}(1+\eps\rho)\cdot \tau_{\min}
                                }  
                            }{
                                \Ex_{\cW}\insquare{\sum_{\ell\in [m-2]} \Pr\insquare{\evE_\ell} \cdot 
                                \inparen{\tau_{S}(\ell)- \tau_{S}(\ell+2)} 
                                    + \Pr\insquare{\evE_{m-1}}(1-\eps\rho)\cdot \tau_{\min}
                                }  
                            }\tag{Using \cref{eq:nonnegative_sigma,eq:one_sigma}}\\
                            &\quad \geq \ \ \frac{
                                \Ex_{\cW}\insquare{(m-2)\tau
                                + \Pr\insquare{\evE_{m-1}}(1+\eps\rho)\cdot \tau_{\min}
                                }  
                            }{
                                \Ex_{\cW}\insquare{(m-2)\tau
                                    + \Pr\insquare{\evE_{m-1}}(1-\eps\rho)\cdot \tau_{\min}
                                }  
                            }\tag{Using that $\Pr[\evE_\ell],{\tau_{S}(\ell)- \tau_{S}(\ell+2)} \geq 0$ for all $\ell\in [m-1]$}\\
                         \end{align*}
                         \begin{align*}
                            &= \ \ \frac{
                                {(m-2)\tau
                                + (1+\eps\rho)\cdot \tau_{\min}
                                }  
                            }{
                                {(m-2)\tau
                                    + (1-\eps\rho)\cdot \tau_{\min}
                                }  
                            }\tag{Using that $\Pr\insquare{\evE_{m-1}}=1$}\\
                            &= \geq  \ \ 1+\Omega\inparen{\eps\rho m^{-1}\cdot \frac{\tau_{\min}}{\tau}}.
                    \end{align*}
                    
                \end{proof}

            \noindent
            Now we complete the proof of \cref{lem:ratio_of_utils}.
            \begin{proof}[Proof of \cref{lem:ratio_of_utils}]
                Let $\lambda\coloneqq 1-\beta$.
                As $\beta \geq 1-m^{-1/2}$, $\lambda\leq m^{1/2}$.
                We will prove (a strengthening of) the contra-positive of the above statement: 
                for any candidates $d,d'\in C$ in the same group ($G_1$ or $G_2$), the following
                \begin{align*}
                     \frac{\omega_{d'}}{\omega_d} \in 1\pm m^{-1}\cdot \Theta(\lambda)
                     \implies \frac{ \Ex_\mu\insquare{f_S(d',\succ)}}{ \Ex_\mu\insquare{f_S(d,\succ)}}\in 1\pm \Theta(\lambda).
                     \yesnum\label{eq:to_prove_lem}
                \end{align*}
                From the Log-Lipschitzness in \cref{def:desired_properties_of_eta2}, we have the following result:
                for any $x,y>0$
                \begin{align*}
                    \frac{x}{y}\in 1\pm {\lambda} \implies \frac{\cdf_\eta(x)}{\cdf_\eta(y)}\in 1\pm O_\pi(\lambda),\yesnum\label{eq:using_Log_lipschitz}
                \end{align*}
                where $\pi$ is a constant in \cref{def:desired_properties_of_eta2}.

                Towards proving \cref{eq:to_prove_lem}, suppose (without loss of generality) that $\omega_{d'}>\omega_d$.
                We can express $\Ex_\mu\insquare{f(d', \succ)}$ and $\Ex_\mu\insquare{f(d, \succ)}$ as follows
                \begin{align*}
                    \Ex_\mu\insquare{f(d', \succ)} = \sum_{\ell=1}^m \Pr_{\mu}\insquare{\pos_\succ(d')=\ell} \tau_{\ell,S}(f)
                    \quad\text{and}\quad
                    \Ex_\mu\insquare{f(d, \succ)} &= \sum_{\ell=1}^m \Pr_{\mu}\insquare{\pos_\succ(d)=\ell} \tau_{\ell,S}(f).
                    \yesnum\label{eq:subs_into}
                \end{align*}
                To prove \cref{eq:to_prove_lem}, it suffices to show that $\Pr_{\mu}\insquare{\pos_\succ(d')=\ell}$ and $\Pr_{\mu}\insquare{\pos_\succ(d)=\ell}$ are multiplicatively close for all $\ell\in [m]$.
                From \cref{eq:using_Log_lipschitz}, for any $i\neq d,d'$, we have the following 
                \[
                    \forall x\geq 0,\quad \frac{\cdf_\eta\inparen{\frac{\omega_{d'}\cdot x}{\omega_i}}}{\cdf_\eta\inparen{\frac{\omega_{d}\cdot x}{\omega_i}}} 
                    \in 1 \pm m^{-1}\cdot \Theta_\pi(\lambda). 
                    \tag{Using that $\frac{\omega_{d'}}{\omega_i}\cdot \frac{\omega_i}{\omega_{d}} \in 1\pm m^{-1}\Theta(\lambda)$}
                    \customlabel{eq:lb_diff}{\theequation}
                \]
                Moreover, we also have 
                \begin{align*}
                    \forall x>0,\quad 
                    \cdf_\eta\inparen{\frac{\omega_{d'}\cdot x}{\omega_{d} \cdot x}} \in 1\pm m^{-1}\cdot \Theta_\pi(\lambda).
                    \yesnum\label{eq:lb_same}
                \end{align*}
                We can show this as follows:
                \begin{align*}
                    \frac{\Pr_{\mu}\insquare{\pos_\succ(d)=\ell}}{\Pr_{\mu}\insquare{\pos_\succ(d')=\ell}}
                    &= \frac{
                        \Ex_{\eta_{d'}}\insquare{
                            \sum_{S\subseteq C\setminus\inbrace{d'}\colon \abs{S}=\ell-1} \prod_{i\in S}\inparen{1-\cdf_\eta\inparen{\frac{\omega_{d'}\eta_{d'}}{\omega_i}}}
                            \prod_{i\not\in S}\cdf_\eta\inparen{\frac{\omega_{d'}\eta_{d'}}{\omega_i}}
                            }
                        }{
                        \Ex_{\eta_{d}}\insquare{
                            \sum_{S\subseteq C\setminus\inbrace{d}\colon \abs{S}=\ell-1} \prod_{i\in S}\inparen{1-\cdf_\eta\inparen{\frac{\omega_{d}\eta_{d}}{\omega_i}}}
                            \prod_{i\not\in S}\cdf_\eta\inparen{\frac{\omega_{d}\eta_{d}}{\omega_i}}
                            }
                        }.
                        \yesnum\label{eq:ratio_bnd_2}
                \end{align*}
                Toward bounding the RHS of the above equality,
                fix any value $\eta_{d'}=\eta_{d}=\eta>0$.
                We have the following 
                \begin{align*}
                    &\frac{
                            {
                            \sum_{S\subseteq C\setminus\inbrace{d'}\colon \abs{S}=\ell-1} \prod_{i\in S}\inparen{1-\cdf_\eta\inparen{\frac{\omega_{d'}}{\omega_i}}}
                            \prod_{i\not\in S}\cdf_\eta\inparen{\frac{\omega_{d'}}{\omega_i}}
                            }
                        }{
                            {
                            \sum_{S\subseteq C\setminus\inbrace{d}\colon \abs{S}=\ell-1} \prod_{i\in S}\inparen{1-\cdf_\eta\inparen{\frac{\omega_{d}}{\omega_i}}}
                            \prod_{i\not\in S}\cdf_\eta\inparen{\frac{\omega_{d}}{\omega_i}}
                            }
                        }\\ 
                    &\qquad \geq\ \  \frac{\cdf_\eta\inparen{\frac{\omega_{d'}}{\omega_d}}}{\cdf_\eta\inparen{\frac{\omega_d}{\omega_{d'}}}}
                    \cdot 
                    \frac{
                            {
                            \sum_{S\subseteq C\setminus\inbrace{d',d}\colon \abs{S}=\ell-1} \prod_{i\in S}\inparen{1-\cdf_\eta\inparen{\frac{\omega_{d'}}{\omega_i}}}
                            \prod_{i\not\in S}\cdf_\eta\inparen{\frac{\omega_{d'}}{\omega_i}}
                            }
                        }{
                            {
                            \sum_{S\subseteq C\setminus\inbrace{d',d}\colon \abs{S}=\ell-1} \prod_{i\in S}\inparen{1-\cdf_\eta\inparen{\frac{\omega_{d}}{\omega_i}}}
                            \prod_{i\not\in S}\cdf_\eta\inparen{\frac{\omega_{d}}{\omega_i}}
                            }
                        }\\ 
                    &\qquad\Stackrel{\eqref{eq:lb_same}}{\geq}\ \  \inparen{1-m^{-1}\Theta(\lambda)}
                    \cdot 
                    \frac{
                            {
                            \sum_{S\subseteq C\setminus\inbrace{d',d}\colon \abs{S}=\ell-1} \prod_{i\in S}\inparen{1-\cdf_\eta\inparen{\frac{\omega_{d'}}{\omega_i}}}
                            \prod_{i\not\in S}\cdf_\eta\inparen{\frac{\omega_{d'}}{\omega_i}}
                            }
                        }{
                            {
                            \sum_{S\subseteq C\setminus\inbrace{d',d}\colon \abs{S}=\ell-1} \prod_{i\in S}\inparen{1-\cdf_\eta\inparen{\frac{\omega_{d}}{\omega_i}}}
                            \prod_{i\not\in S}\cdf_\eta\inparen{\frac{\omega_{d}}{\omega_i}}
                            }
                        }\\ 
                    &\qquad\Stackrel{\eqref{eq:lb_diff}}{\geq}\ \  \inparen{1-m^{-1}\Theta(\lambda)}
                        \cdot (1-m^{-1}\Theta(\lambda))^{m-\ell+1}\\
                    &\qquad \geq \ \  1-\Theta(\lambda). \tagnum{Using that $\ell\geq 0$ and $0\leq \lambda \leq 1$}\customlabel{eq:bound_on_ratio}{\theequation}
                \end{align*}
                Substituting Equation~\eqref{eq:bound_on_ratio} in \cref{eq:ratio_bnd_2} and using the fact that if $\frac{x_1}{y_1} ,\frac{x_2}{y_2} \geq r$, then $\frac{x_1+x_2}{y_1+y_2}\geq r$ (for any $x_1,x_2,y_1,y_2,r\geq 0$), it follows that 
                \begin{align*}
                    \frac{\Pr_{\mu}\insquare{\pos_\succ(d)=\ell}}{\Pr_{\mu}\insquare{\pos_\succ(d')=\ell}}
                        \geq 
                        1-\Theta(\lambda).
                        \yesnum\label{eq:lblkhdlkd}
                \end{align*}
                An analogous argument also shows that 
                \begin{align*}
                    \frac{\Pr_{\mu}\insquare{\pos_\succ(d)=\ell}}{\Pr_{\mu}\insquare{\pos_\succ(d')=\ell}}
                        \leq  
                        1+\Theta(\lambda).
                        \yesnum\label{eq:ublkhdlkd}
                \end{align*}
                The result follows by substituting \cref{eq:lblkhdlkd,eq:ublkhdlkd} into \cref{eq:subs_into}.
    \end{proof}

\subsubsection{Bounding $\alpha$ for the Utility-Based Model}

        In this section, we prove \cref{cor:algorithmic}--which lower bounds $\alpha$ for the utility-based latent and biased generative models.

        Recall that in the utility-based model (\cref{ex:utility_generative}), the variable $\eta$ is drawn from the uniform distribution on $[0,1]$.
        We will, in fact, prove a more general version of the above lemma that holds for any distribution of $\eta$ that satisfies certain properties (\cref{def:eta_for_alpha}).
        With some abuse of notation, we use $\eta$ to denote both the distribution and a value drawn from the distribution $\eta$ (independent of all other randomness).
        \begin{definition}[\textbf{Properties of the utility distribution $\eta$ from \cref{def:utility_bias}}]\label{def:eta_for_alpha}
            Let $\eta$ be the distribution on $\R_{\geq 0}$ from \cref{def:utility_bias} that parameterizes the generative model $\mu$.
            Let $\cdf_{\eta}\colon \R\to [0,1]$ be the cumulative distribution function of $\eta$.
            We define the following properties of $\eta$.
            \begin{itemize} %
                \item (Linear scaling) We say that the cumulative distribution function of $\eta$ scales linearly if there exist constants $\lambda_0,\pi\in (0,1)$ such that for all $\lambda \in (0,\lambda_0)$ and $\rho > 0$, $\cdf_\eta(\lambda \rho) \in \insquare{\pi\lambda \cdot \cdf_\eta(\rho), \pi^{-1}\lambda\cdot  \cdf_\eta(\rho)}$;
                \item (Exponential tail) We say that $\eta$ has an exponential tail if there exists a constant $\pi>0$ such that, for all $t\geq \pi^{-1}$, $\Pr_{X\sim \eta}\insquare{X\geq t}\leq \exp\inparen{-\pi t}$.
            \end{itemize}
        \end{definition}
        \noindent Several distributions satisfy the above properties including the uniform distribution on $[0,1]$, the exponential distribution, and the normal distribution truncated to $\R_{\geq 0}$.
        Roughly, the linear scaling property ensures that $\cdf_\eta(\lambda \rho)$ behaves as $\lambda\cdot \cdf_\eta(\rho)$ for $\lambda$ and $\rho$ close to 0.
        The exponential tail property ensures that with high probability the utilities of all candidates are bounded above.
        These are relevant as for two candidates $c\in G_1$ and $d\in G_2$ and any user $v\in V$, the probability that $\Pr\left[\wh{w}_{v,d} > \wh{w}_{v,c}\mid w_{v,c} \right] = \cdf_\eta\inparen{\theta\cdot \frac{w_{v,d}}{\omega_{c}}}$.
        The exponential tail ensures that $\frac{w_{v,d}}{\omega_{c}}$ is bounded and, hence, we can invoke the linear scaling property.
        The linear scaling property, itself, ensures that the probability that a disadvantaged candidate $d\in G_2$ is placed before an advantaged candidate $c\in G_1$ scales with $\theta$.

        For each $j\in [m]$, let $i_{\succ,j}$ be the $j$-th candidate in $\succ$.
        For each $j\in [m]$, define $\sigma_j(f)$ to be the marginal score of $i_{\succ,j}$ with respect to $C\setminus\inbrace{i_{\succ,j}}$, i.e., $\sigma_j(f)\coloneqq f_{C\setminus\inbrace{i_{\succ,j}}}(\inbrace{i_{\succ,j}},\succ)$ for all $\succ\ \in \prefs{C}$.

        For any $\eta$ which has the properties in \cref{def:eta_for_alpha}, we show the following bounds on $\alpha.$
        
        \begin{lemma}[\textbf{Bound on $\alpha$ for the utility-based generative model}]\label{lem:utility_alpha}
            Let $F = \sum_{v\in V} f(\cdot ,\succ_v)$ be a latent multiwinner score function such that $\sigma_{j}>0$ and  $\sigma_{j+1}=\sigma_{j+2}=\dots = 0$. 
            Let $(\mu,\wh{\mu})$ be generative models in \cref{ex:utility_generative,def:utility_bias} with any parameters $\omega\in \R_{\geq 0}^C$ and $\theta\in (0,1]$.
            Let $\mu$ be such that $M\cap G_2\neq \emptyset$.
            There exists a $\theta_0 > 0$ (which is a function of $\eta$ and $\omega$) such that for all $\theta \in (0,\theta_0)$ the following bounds holds for $\alpha$.
            \[
                \Theta\inparen{\theta^{\abs{G_1}-j}\times \alpha_0\times \frac{\sigma_j}{\tau} }  \leq \alpha \leq \Theta\inparen{\theta^{\abs{G_1}-j}}.
            \]
            Where $\alpha_0\coloneqq (\tau(f))^{-1}\cdot \min_{c\in M} \Ex_{\mu}\insquare{f_{C\setminus \inbrace{c}}(c,\succ)}$.
            If $\ell=m-1$, then the upper bound improves to $1$.
            If $M\cap G_2=\emptyset$, the lower bound improves to $\alpha_0$ (the best possible).
        \end{lemma}
    \noindent Recall the closer $\alpha$ is to 1 the ``easier'' it is to distinguish candidates in $M$ and in $C\setminus M$.
    The above result shows that for any multiwinner score function $F$, $\alpha$ is characterized by the value $b$, which is the smallest position such that $f_{C\setminus c}(c,\succ)=0$ for any $c$ with $\pos_{\succ}(c) > b$.
    Substituting the values of $b$ for functions $F$ from \cref{sec:example} gives us the following bounds on $\alpha$ for any $\theta\in (0,\theta_0)$.
    \begin{itemize} %
        \item If $F$ is the $\ell_1$ CC rule, then $\Theta\inparen{\alpha_0\cdot \theta^{\abs{G_1}-1}} \leq \alpha \leq \Theta\inparen{\theta^{\abs{G_1}-1}}$; 
        \item If $F$ is the $b$-Bloc rule, then $\Theta\inparen{\alpha_0\cdot \theta^{\abs{G_1}-b}} \leq \alpha \leq \Theta\inparen{\theta^{\abs{G_1}-b}}$; and
        \item If $F$ is the Borda rule, then $\Theta\inparen{\alpha_0\cdot \theta} \leq \alpha \leq \Theta(1)$.
    \end{itemize}
    Substituting the lower bounds on $\alpha$ into \cref{thm:main_algorithmic}, we recover \cref{cor:algorithmic}.
    It remains to prove \cref{lem:utility_alpha}.

    Order entries of $\omega$, to get the following values $\omega_{(1)} \geq \omega_{(2)} \geq \dots \geq \omega_{(m)}.$
    Define ${\rho_0 \coloneqq \frac{\omega_{(1)}}{\omega_{(m)}}.}$
    Let $\pi,\lambda_0$ be the constants corresponding to $\rho_0$ in the linear scaling property (\cref{def:eta_for_alpha}).
    Define
    \[
        {\theta_0 \coloneqq \min\inbrace{\rho_0, \lambda_0, \pi}.}
    \]
    Suppose $\theta\in (0,\theta_0)$.
    Fix any candidate $c\in G_2\cap M$ and any voter $v\in V$.
    To prove \cref{lem:utility_alpha}, we need to prove that 
    \[
        \Theta\inparen{\theta^{m-j}\times \alpha_0\times \frac{\sigma_j}{\tau}} 
        \quad \leq \quad
        \frac{1}{\tau}\cdot \Ex_{\wh{\mu}}\insquare{f_{C\setminus \inbrace{c}}(c,\succ)}
        \quad \leq \quad 
        \Theta\inparen{\theta^{\abs{G_1}-j}}.
        \yesnum\label{eq:to_prove}
    \]
    We can express $\Ex_{\wh{\mu}}\insquare{f_{C\setminus \inbrace{c}}(c,\succ)}$ as follows 
    \begin{align*}
        \Ex_{\wh{\mu}}\insquare{f_{C\setminus \inbrace{c}}(c,\nsucc)}
            &= \sum_{\ell=1}^m \Pr_\mu\insquare{ \pos_{\nsucc}(c) = \ell }\sigma_\ell\\
            &= \sum_{\ell=1}^m \Pr_{\wh{\mu}}\insquare{ \pos_{\nsucc}(c) \leq \ell }\inparen{\sigma_\ell-\sigma_{\ell+1}}.
        \yesnum\label{eq:prob_first_ell}
    \end{align*}
    We will bound $\Pr_\mu\insquare{ \pos_{\nsucc}(i) = \ell }$.
    Let $\cdf_\eta(\cdot)$ be the cumulative distribution function of distribution $\eta$.
    For each $d\in C$, let $\eta_d$ be a draw from distribution $\eta$ such that  $w_{v,d}=\omega_d\cdot \eta_d$.
    Using the fact for all $i,j\in C$, $i$ appears before $j$ in $\succ$ if and only if $\omega_i \cdot\theta^{\mathds{I}[i\in G_2]}\cdot \eta_i > \omega_j\cdot \theta^{\mathds{I}[j\in G_2]}\cdot \eta_j$, we can upper bound $\Pr_\mu\insquare{ \pos_{\nsucc}(i) = \ell }$ as follows 
    \begin{align}
            \label{eq:product_to_ub}
        \Pr_{\wh{\mu}}\insquare{ \pos_{\nsucc}(c) \leq \ell }
            &= \Ex_{\eta_c}\insquare{
                    \sum_{S\subseteq C\setminus \inbrace{c}\colon \abs{S}\leq \ell} \ \ 
                        \prod_{d\in S}\inparen{1-\cdf_\eta\inparen{\frac{\omega_c\cdot \theta \cdot \eta_c}{w_d\cdot \theta^{\mathds{I}[d\in G_2]}}}}
                        \prod_{d\not\in S}{\cdf_\eta\inparen{\frac{\omega_c\cdot \theta \cdot \eta_c}{w_d\cdot \theta^{\mathds{I}[d\in G_2]}}}}
                }.
    \end{align}
    We can separate the second product as follows 
    \begin{align*}
        \prod_{d\not\in S}{\cdf_\eta\inparen{\frac{\omega_c\cdot \theta \cdot \eta_c}{w_d\cdot \theta^{\mathds{I}[d\in G_2]}}}}
        = 
            \prod_{d\in G_1\setminus S}{\cdf_\eta\inparen{\frac{\omega_c\cdot \theta \cdot \eta_c}{w_d}}}\cdot 
            \prod_{d\in G_2\setminus S}{\cdf_\eta\inparen{\frac{\omega_c\cdot \eta_c}{w_d}}}.
        \yesnum\label{eq:product_to_ub2}
    \end{align*}
    From the linear scaling property of $\eta$ (\cref{def:eta_for_alpha}) and the facts that $\theta\leq \theta_0\leq \lambda_0$ and $\frac{\omega_c \cdot \eta_c}{w_d} > 0$, we have the following bounds for each $d\in G_1$
    \begin{align*}
        \pi \theta \cdot \cdf_\eta\inparen{\frac{\omega_c \cdot \eta_c}{w_d}}
        \leq \cdf_\eta\inparen{\frac{\omega_c\cdot \theta \cdot \eta_c}{w_d}}
        \leq 
        \frac{\theta}{\pi}\cdot \cdf_\eta\inparen{\frac{\omega_c \cdot \eta_c}{w_d}}
        \leq 
        \frac{\theta}{\pi}.\tag{Using that $\cdf_\eta(x)\leq 1$ for all $x\in \R$}
    \end{align*}
    Substituting these bounds in Equation~\eqref{eq:product_to_ub2}, we get 
    \begin{align*}
        \prod_{d\not\in S}{\cdf_\eta\inparen{\frac{\omega_c\cdot \theta \cdot \eta_c}{w_d\cdot \theta^{\mathds{I}[d\in G_2]}}}}
        &\leq \inparen{\frac{\theta}{\pi}}^{\abs{G_1\setminus S}} \cdot \prod_{d\in G_2\setminus S}{\cdf_\eta\inparen{\frac{\omega_c\cdot \eta_c}{w_d}}}\\
        &\leq \inparen{\frac{\theta}{\pi}}^{\abs{G_1\setminus S}},\tag{Using that $\cdf_\eta(x)\leq 1$ for all $x\in \R$}\\
        \prod_{d\not\in S}{\cdf_\eta\inparen{\frac{\omega_c\cdot \theta \cdot \eta_c}{w_d\cdot \theta^{\mathds{I}[d\in G_2]}}}}
        &\geq \inparen{\pi \theta}^{\abs{G_1\setminus S}} \cdot \prod_{d\not\in S}{\cdf_\eta\inparen{\frac{\omega_c\cdot \theta \cdot \eta_c}{w_d\cdot \theta^{\mathds{I}[d\in G_2]}}}}.
    \end{align*}
    Substituting these in Equation~\eqref{eq:product_to_ub}, implies the following bounds
    \begin{align*}
        & \quad \Pr_{\wh{\mu}}\insquare{ \pos_{\nsucc}(c) \leq \ell } \\
            \leq & \quad \Ex_{\eta_c}\insquare{
                    \sum_{S\subseteq C\setminus \inbrace{c}\colon \abs{S}\leq \ell} \ \ 
                        \prod_{d\in S}\inparen{1-\cdf_\eta\inparen{\frac{\omega_c\cdot \theta \cdot \eta_c}{w_d\cdot \theta^{\mathds{I}[d\in G_2]}}}}
                        \cdot \inparen{\frac{\theta}{\pi}}^{\abs{G_1\setminus S}}
                }\\ 
            \leq & \quad \Ex_{\eta_c}\insquare{
                        \inparen{\frac{\theta}{\pi}}^{\abs{G_1}-\ell}
                }\tag{Using that $\abs{S}\leq \ell$, $\theta\leq \theta_0\leq \pi$, and $\cdf_\eta(x)\leq 1$ for all $x\in \R$}\\
            \leq & \quad \inparen{\frac{\theta}{\pi}}^{\abs{G_1}-\ell},
    \end{align*}
        and
    \begin{align*}
        & \quad \Pr_{\wh{\mu}}\insquare{ \pos_{\nsucc}(c) \leq \ell } \\
            \geq & \quad \small \Ex_{\eta_c}\insquare{
                    \sum_{S\subseteq C\setminus \inbrace{c}\colon \abs{S}\leq \ell} \ \ 
                        \inparen{\pi \theta}^{\abs{G_1\setminus S}}\cdot \ \
                        \prod_{d\in S}\inparen{1-\cdf_\eta\inparen{\frac{\omega_c\cdot \theta \cdot \eta_c}{w_d\cdot \theta^{\mathds{I}[d\in G_2]}}}}
                        \cdot \prod_{d\not\in S}{\cdf_\eta\inparen{\frac{\omega_c\cdot \theta \cdot \eta_c}{w_d\cdot \theta^{\mathds{I}[d\in G_2]}}}}
                }\\ 
                 & \quad \tag{Using that $\abs{S}\leq \ell$, $\theta\pi\leq \theta_0\leq 1$, and $\cdf_\eta(x)\leq 1$ for all $x\in \R$}\\
            \geq & \quad \inparen{\pi \theta}^{\abs{G_1}-\ell} \cdot \Pr_\mu\insquare{ \pos_{\succ}(c) \leq \ell }.
        \end{align*}
    Substituting this in \cref{eq:prob_first_ell}, we get 
    \begin{align*}
        \Ex_{\wh{\mu}}\insquare{f_{C\setminus \inbrace{c}}(c,\nsucc)}
            &= \sum_{\ell=1}^m \Pr_{\wh{\mu}}\insquare{ \pos_{\nsucc}(c)\leq \ell }\inparen{\sigma_\ell-\sigma_{\ell+1}}\\
            &\leq \sum_{\ell=1}^m \inparen{\frac{\theta}{\pi}}^{\abs{G_1}-\ell} \inparen{\sigma_\ell-\sigma_{\ell+1}}\\
            &= \sum_{\ell=1}^j \inparen{\frac{\theta}{\pi}}^{\abs{G_1}-\ell} \inparen{\sigma_\ell-\sigma_{\ell+1}}\tag{Using that $\sigma_{j+1}=\sigma_{j+2}=\dots=\sigma_{m}=0$} %
        \end{align*}
        \begin{align*}
            &\leq \tau\cdot \sum_{\ell=1}^j \inparen{\frac{\theta}{\pi}}^{\abs{G_1}-\ell} \tag{Using that $\sigma_{\ell}\leq \tau$ for all $\ell\in [m]$}\\
            &\leq \tau\cdot \inparen{\frac{\theta}{\pi}}^{\abs{G_1}-j}  \tag{Using that $0\leq \frac{\theta}{\pi}\leq 1$}.
    \end{align*}
    We also have the following lower bound 
    \begin{align*}
        \Ex_{\wh{\mu}}\insquare{f_{C\setminus \inbrace{c}}(c,\nsucc)}
            &= \sum_{\ell=1}^m \Pr_{\wh{\mu}}\insquare{ \pos_{\nsucc}(c) \leq \ell }\inparen{\sigma_\ell-\sigma_{\ell+1}}\\
            &\geq \sum_{\ell=1}^m \inparen{\pi \theta}^{\abs{G_1}-\ell} \cdot  \Pr_\mu\insquare{ \pos_{\succ}(c) \leq \ell }\cdot  \inparen{\sigma_\ell-\sigma_{\ell+1}}\\
            &= \sum_{\ell=1}^j \inparen{\pi \theta}^{\abs{G_1}-\ell} \cdot  \Pr_\mu\insquare{ \pos_{\succ}(c) \leq \ell }\cdot  \inparen{\sigma_\ell-\sigma_{\ell+1}}\tag{Using that $\sigma_{j+1}=\sigma_{j+2}=\dots=\sigma_{m}=0$}\\
            &\geq \inparen{\pi \theta}^{\abs{G_1}-j} \cdot  \Pr_\mu\insquare{ \pos_{\succ}(c) \leq j } \cdot \inparen{\sigma_j-\sigma_{j+1}}\tag{Using that $\sigma_{\ell}\geq \sigma_{\ell+1}$ and $\Pr_\mu\insquare{ \pos_{\succ}(c) \leq \ell }\geq 0$ for all $\ell\in [m]$}\\
            &\geq \inparen{\pi \theta}^{\abs{G_1}-j} \cdot \Pr_\mu\insquare{ \pos_{\succ}(c) \leq j } \cdot \sigma_j  \tag{Using that $\sigma_{j+1}=0$.}\\ 
            &\geq \inparen{\pi \theta}^{\abs{G_1}-j} \cdot \frac{\alpha_0}{\tau}.
            \tag{Using that $\alpha_0 = \sum_{\ell=1}^m \Pr_\mu\insquare{ \pos_{\succ}(c) \leq \ell }\cdot  \inparen{\sigma_\ell-\sigma_{\ell+1}} \leq \Pr_\mu\insquare{ \pos_{\succ}(c) \leq j }\cdot \tau$}
    \end{align*}
    This completes the proof of \cref{eq:to_prove} and, hence, \cref{lem:utility_alpha}.

\subsection{Extension to Swapping-Based Biased Generative Model}   \label{sec:proofof:lem:swap_order_preservation}\label{sec:swapping-based}
In this section, we introduced a ``swapping-based'' biased model.
Let $\mu$ be a generative model of latent preference lists $\succ$, e.g., \Cref{ex:utility_generative}.
We propose the following generative model for biased preferences which can be seen as a biased variant of the popular Mallows model \cite{mallows1957non}.
\begin{definition}[\bf{Swapping-based generative model of biased preference lists}]
	\label{def:swapping}
	Let $\phi \in [0,1]$ and $t\geq 1$ be a bias parameter and number of swaps respectively.
	For any $\succ$, let $A(\succ)\subseteq [m]\times [m]$ be the collection of all pairs $(i,j)$ such that there exist $c\in G_1$ and $c'\in G_2$ with $\pos_{\succ}(c) = i > j = \pos_{\succ}(c')$.
	Let $\succ_1$ be a preference list drawn from $\mu$.
	Given $\succ_i$ (for any $i\in [t]$), $\succ_{i+1}$ is generated as follows.
	\begin{enumerate} %
		\item Sample a pair $(i,j)\in A(\succ_i)$ with probability $\frac{\phi^{i-j}}{Z(\succ_i)}$, where, for any $\succ$, $Z(\succ)\geq 0$ is a normalization factor defined as $Z(\succ) = \sum_{(i', j')\in A(\succ)} \phi^{i'-j'}$; and
		\item Swap the candidates at positions $i$ and $j$ in $\succ_{i}$, and obtain $\succ_{i+1}$.
	\end{enumerate}
	Define $\nsucc\coloneqq \succ_{t+1}$.
	Let $\widehat{\mu}$ denote the generative model of $\nsucc$ that depends on $\succ$ and $\phi$.
\end{definition}

\noindent
Intuitively, we randomly improve the ranking of a candidate from the advantaged group and lower the ranking of a candidate from the disadvantaged group, where the probability is proportional to their ranking difference in $\succ$.
As $\phi$ comes closer to 1, the average distance between the positions of swapped candidates increases.

\subsubsection{Order-Preservation of the Swapping-Based Model}
\label{sec:swapping_order_preserve:app}

We first show that all multiwinner score functions satisfying \cref{def:score} are order preserving between $\mu$ and $\wh{\mu}$, where $\wh{\mu}$ arises from the swapping-based generative model of biased preference lists and $\mu$ satisfies the following condition for some parameter $\rho>0$
\[
    \min_{c\in C}\Ex_\mu\insquare{f_{C\setminus \inbrace{c}}(c,\succ{})}\geq \rho.
    \yesnum\label{eq:lbddd}
\]

\begin{lemma}[\textbf{Order-preserving properties of a swapping-based model}]\label{lem:swap_order_preservation}
    \sloppy
    Let $F = \sum_{v\in V} f(\cdot, \succ_v)$ be a latent multiwinner score function satisfying \Cref{def:score}. %
    Let $\mu$ be any generative model such that $F$ is order-preserving with respect to $\mu$ (\cref{def:bias}) and satisfies \cref{eq:lbddd}.
    For any numbers $t\geq 1$ and $\phi\in (0,t^{-1})$ and the generative model $\wh{\mu}$ in \cref{def:swapping} with parameters $\mu$, $\phi$, and $t=1$,
    $F$ is $\inparen{1-\lambda,1-\frac{\lambda}{2}}$ order preserving between $\mu$ and $\wh{\mu}$ where 
    \[
        \lambda\coloneqq \Ex_{\mu}\insquare{\frac{1}{Z(\succ)}}\cdot \Theta\inparen{\frac{t\phi}{1-\phi}\cdot \frac{\tau}{\rho}},
    \]
    and $Z(\succ)$ is the normalizing constant corresponding to preference list $\succ$, as defined in \cref{def:swapping}.
\end{lemma}
Note that \cref{lem:swap_order_preservation} does not fix a specific generative model of latent preference lists $\mu$.
The bound on the parameter $\gamma$ depends on $\mu$ via the term $\Ex_{\mu}\insquare{\frac{1}{Z(\succ)}}$.
In general, we expect this term to be of the order $\Omega\inparen{m^{-1}\phi^{-1}}$.
To see why, note that for any preference list $\succ$ where there are at least $r$ candidates $c'\in G_2$ (for any $r\geq 1$) who are placed before $r^{-1}m$ candidates $c\in G_1$, then $Z(\succ)\geq \phi r$.
When $\Ex_{\mu}\insquare{\frac{1}{Z(\succ)}}=\Omega\inparen{m^{-1}\phi^{-1}}$, then \cref{lem:swap_order_preservation} implies that $F$ is $\inparen{1, 1-O\inparen{m^{-1}t}}$ order preserving between $\mu$ and $\wh{\mu}$ for any $\phi\in (0,t^{-1})$.

Additionally, we have the following example that shows that this order preservation between $\mu$ and $\wh{\mu}$ does not hold for all $\beta\in [0,1]$ when $\wh{\mu}$ the swapping-based bias generative model.
\begin{example}[\textbf{Order preservation does not hold for all $\beta\in [0,1]$}]\label{ex:counter_example_opII}
    Suppose $C = \inbrace{d_1,d_2,a_1}$, $G_1=\inbrace{a_1}$, and $G_2=\inbrace{d_1,d_2}$.
    Let $F\colon 2^C\to \R_{\geq 0}$ be the $2$-Bloc rule.
    Let $\succ_1,\succ_2,$ and $\succ_3$  be the following preference lists
    \begin{align*}
        \succ_1 \coloneqq (d_2\succ d_1\succ a_1),\quad 
        \succ_2 \coloneqq (d_1\succ a_1 \succ d_2),\quad 
        \succ_3 \coloneqq (d_2\succ a_1 \succ d_1).
    \end{align*}
    For some small $\delta>0$, let $\mu$ be a distribution such that 
    \begin{align*}
        \Pr_{\succ\sim\mu}[\succ=\succ_1] = 1-3\delta,\quad
        \Pr_{\succ\sim\mu}[\succ=\succ_2] = 2\delta,\quad
        \Pr_{\succ\sim\mu}[\succ=\succ_3] = \delta.
    \end{align*}
    In this case, on the one hand, the following holds
    \begin{align*}
        \frac{\Ex_\mu\insquare{f(d_1,\succ)}}{\Ex_\mu\insquare{f(d_2,\succ)}}= \frac{1-\delta}{1-2\delta} = \frac{1}{1-\eps}>1\quad \text{for } \eps\approx \delta.
    \end{align*}
    On the other hand, 
    \begin{align*}
        \frac{\Ex_{\wh{\mu}}\insquare{f(d_1,\not\succ)}}{\Ex_{\wh{\mu}}\insquare{f(d_2,\not\succ)}}
        = \frac{\Theta(\phi+\delta)}{1\pm\Theta(\phi+\delta)}
        = \Theta(\phi+\delta) < 1.
    \end{align*}
\end{example}

        Next, we prove \cref{lem:swap_order_preservation}.
        To show that $F$ \emph{$(\beta,\gamma)$ order preserving between $\mu$ and $\widehat{\mu}$}, we need to show that 
        for any two candidates $c,c'\in C$ belonging to the same group ($G_1$ or $G_2$) and any $S\subseteq C\setminus \{c,c'\}$,  the following holds 
        \[
            \beta \cdot \Exp_{\mu}\left[f_S(c',\succ)\right] \geq \Exp_{\mu}\left[f_S(c,\succ) \right] > 0 
            \implies 
            \gamma\cdot {\Exp_{\widehat{\mu}}\left[f_S(c', \nsucc)\right]}\geq   {\Exp_{\widehat{\mu}}\left[f_S(c, \nsucc)\right]}.
            \yesnum\label{eq:implication_swap}
        \]
        At a high level, we will prove the above implication by bounded $\Exp_{\widehat{\mu}}\left[f_S(c, \nsucc)\mid \succ\right]$ with a $1\pm O(\phi t)$ multiplicative factor of $f_S(c, \succ)$ (for any $\succ\ \in \prefs{C}$ and $c\in C$) and then taking an expectation with respect to $\succ\sim \mu$.
        Concretely, we prove the following lemma.
        \begin{lemma}\label{lem:swap_simple}
            Suppose $t\geq 1$, $\phi\in (0,t^{-1})$, and $\succ\ \in \prefs{C}$.
            Let $\lambda\coloneqq \Ex_{\mu}\insquare{\frac{1}{Z(\succ)}}\cdot O\inparen{\frac{t\tau\phi}{\rho(1-\phi)}}$.
            For any $c\in C$, and $S\subseteq C\setminus \inbrace{c}$,  it holds that 
                \[
                    \Ex_\mu\insquare{f_S(c, \succ)}\cdot (1 - \lambda)
                     \leq 
                    \Exp_{\widehat{\mu}}\left[f_S(c, \nsucc) \right] 
                     \leq 
                    \Ex_\mu\insquare{f_S(c, \succ)}\cdot (1 + \lambda).
                \]
        \end{lemma}
        \cref{lem:swap_order_preservation} follows from the above result by taking an expectation with respect to $\succ\sim \mu$.
        \begin{proof}[Proof of \cref{lem:swap_order_preservation} assuming \cref{lem:swap_simple}]
            Fix any two candidates $c,c'\in C$ belonging to the same group ($G_1$ or $G_2$) and any $S\subseteq C\setminus \{c,c'\}$.
            To prove \cref{lem:swap_order_preservation} it suffices to prove \cref{eq:implication_swap}.
            Suppose the following is true 
            \[
                \beta \cdot \Exp_{\mu}\left[f_S(c',\succ)\right] \geq \Exp_{\mu}\left[f_S(c,\succ) \right] > 0.\yesnum\label{eq:precondition_swap}
            \]
            (If this is not true, \cref{eq:implication_swap} vacuously holds.)
            Using \cref{lem:swap_simple}, for any $\succ\ \in \prefs{C}$, we have the following inequalities 

            \begin{align*}
                \Ex_\mu\insquare{f_S(c', \succ)} \cdot (1 - \lambda)
                \quad \leq \quad 
                \Exp_{\widehat{\mu}}\left[f_S(c', \nsucc) \right] 
                \quad &\leq \quad 
                \Ex_\mu\insquare{f_S(c', \succ)} \cdot (1 + \lambda),\yesnum\label{eq:swap_order_eq1}\\ 
                \Ex_\mu\insquare{f_S(c, \succ)}\cdot (1 - \lambda)
                \quad \leq \quad 
                \Exp_{\widehat{\mu}}\left[f_S(c, \nsucc) \right] 
                \quad &\leq \quad 
                \Ex_\mu\insquare{f_S(c, \succ)}\cdot (1 + \lambda).\yesnum\label{eq:swap_order_eq2}
            \end{align*}
            Now, we are ready to complete the proof 
            \begin{align*}
                {\Exp_{\widehat{\mu}}\left[f_S(c', \nsucc)\right]}\ \ 
                &\Stackrel{\eqref{eq:swap_order_eq1}}{\geq} \ \  
                    (1 - \lambda)\cdot \Ex_\mu\insquare{f_S(c', \succ)}\\
                &\Stackrel{\eqref{eq:precondition_swap}}{\geq} \ \   \frac{1 - \lambda}{\beta} \cdot \Ex_\mu\insquare{f_S(c, \succ)}\\
                &\Stackrel{\eqref{eq:swap_order_eq2}}{\geq} \ \   \frac{1 - \lambda}{(1+\lambda)\cdot\beta} \cdot \Ex_{\wh{\mu}}\insquare{f_S(c, \nsucc)}.
            \end{align*}
            \cref{lem:swap_order_preservation} follows by choosing $r\coloneqq \frac{4\lambda}{1+3\lambda}=\Ex_{\mu}\insquare{\frac{1}{Z(\succ)}}\cdot O\inparen{\frac{t\tau\phi}{\rho(1-\phi)}}$,
            \[
                \beta=1-r,
                \quad \text{and}\quad 
                \gamma = \frac{\beta(1+\lambda)}{1-\lambda} = 1-\frac{r}{2}.
            \]
        \end{proof}
        It remains to prove \cref{lem:swap_simple}.
        
        \paragraph{Notation.}
        Fix $c\in C$ and $S\subseteq C\setminus \inbrace{c}$.
        For each $j\in [m]$, let $i_{\succ,j}$ be the $j$-th candidate in $\succ$.
        For each $j\in [m]$, define $\tau_{j,S}(f)$ to be the marginal score of $i_{\succ,j}$ with respect to $S$, i.e., 
        $$\tau_{j, S}(f)\coloneqq f_S(\inbrace{i_{\succ,j}},\succ)$$ 
        which is independent of $\succ\ \in \prefs{C}$.
        Fix any draw $\succ$ from $\mu$.
        For each $1\leq \ell \leq m$, let $A(\ell)\subseteq [m]\setminus [\ell]$ be the set of indices among $[m]\setminus [\ell]$ where $G_1$ candidates appear in $\succ$.
        For each $1\leq \ell \leq m$, let $B(\ell)\subseteq  [\ell]$ be the set of indices among $ [\ell]$ where $G_2$ candidates appear in $\succ$.
        
        We divide the proof into two cases depending on the group of $c.$

        \paragraph{Case A ($c\in G_1$):}
            Suppose $c\in G_1$.
            Let 
            \[
                j\coloneqq \pos_\succ(c).
            \]
            We will first consider the case where there is only one swap, i.e., $t=1$, and later generalize to multiple swaps.
            We can express $\Exp_{\widehat{\mu}}\left[f_S(c, \nsucc)\mid \succ\right]$ as follows.
            \begin{align*}
                &\Exp_{\widehat{\mu}}\left[f_S(c, \nsucc)\mid \succ\right]\\
                    &\quad= 
                    \sum_{\ell\in [j]}\Pr\insquare{\text{$c$ is swapped from $j$ to $\ell$}}\cdot  \tau_{\ell,S} 
                    +\inparen{1-\sum_{\ell\in [j]}\Pr\insquare{\text{$c$ is swapped from $j$ to $\ell$}} } \tau_{j,S}
                    \tag{The positions of candidates in $G_1$ only reduces, and the definition of $\tau_{\ell,S}$}\\             
                    &\quad= 
                    \sum_{\ell\in B\inparen{j}}\Pr\insquare{\text{$c$ is swapped from $j$ to $\ell$}}\cdot  \tau_{\ell,S} 
                    +\inparen{1-\sum_{\ell\in B\inparen{j}}\Pr\insquare{\text{$c$ is swapped from $j$ to $\ell$}} } \tau_{j,S}
                    \tag{Candidates in $G_1$ only swap positions with candidates in $G_2$ in one swap}\\
                    &\quad= 
                    \frac{1}{Z(\succ)} \sum_{\ell\in B\inparen{j}} \phi^{j-\ell} \cdot \tau_{\ell,S} 
                    +\inparen{1-\frac{1}{Z(\succ)} \sum_{\ell\in B\inparen{j}} \phi^{j-\ell}} \tau_{j,S}.
                    \tagnum{By construction in the swapping-based bias model; \cref{def:swapping}}\customlabel{eq:swapping_expression}{\theequation}
            \end{align*}
            
            We can upper bound the above expression as follows 
            \begin{align*}
                &\Exp_{\widehat{\mu}}\left[f_S(c, \nsucc)\mid \succ\right]\\
                    &\quad = \frac{1}{Z(\succ)} \sum_{\ell\in B\inparen{j}} \phi^{j-\ell} \cdot \tau_{\ell,S} 
                    +\inparen{1-\frac{1}{Z(\succ)} \sum_{\ell\in B\inparen{j}} \phi^{j-\ell}} \tau_{j,S}\\
                    &\quad \leq 
                    \frac{1}{Z(\succ)} \inparen{\tau_{1,S} \cdot \phi^{j-1} + \tau_{2,S} \cdot \phi^{j-2} +
                    \dots + \tau_{j-1, S} \cdot \phi }
                    +\inparen{1-\frac{1}{Z(\succ)} \cdot \inparen{\phi+\phi^2+\dots+\phi^{j-1}} } \tau_{j,S}\\
                    &\quad \leq 
                    \frac{1}{Z(\succ)} \inparen{\tau_{1,S} \cdot \phi^{j-1} + \tau_{2,S} \cdot \phi^{j-2} +
                    \dots + \tau_{j-1, S} \cdot \phi }
                    +\tau_{j,S} \tag{Using that $\phi,Z(\succ)\geq 0$}\\
                    &\quad \leq 
                    \frac{\tau_{1,S}}{Z(\succ)} \cdot \inparen{\phi + \phi^2 + \dots + \phi^{j-1}}
                    + \tau_{j,S} \tag{Using that $\tau_{1,S}\geq \tau_{2,S}\geq \dots\geq \tau_m$}\\
                    &\quad \leq 
                    \frac{\tau}{Z(\succ)} \cdot \frac{\phi}{1-\phi}
                    + \tau_{j,S} \tag{Using that $\phi\geq 0$ and $\tau_{1,S}\leq \tau$}\\
                    &\quad = 
                    f_S(c,\succ)+\frac{\tau}{Z(\succ)} \cdot \frac{\phi}{1-\phi}.
            \end{align*}
            Taking the expectation over $\succ\ \sim \mu$, we get that 
            \begin{align*}
                \Exp_{\widehat{\mu}}\left[f_S(c, \nsucc) \right] \ \ 
                    &\Stackrel{}{\leq}\ \  \Ex_\mu\insquare{f_S(c,\succ)} +\Ex_\mu\insquare{\frac{1}{Z(\succ)}} \cdot \frac{\tau\phi}{1-\phi}\\
                    &\Stackrel{\eqref{eq:lbddd}}{\leq}\ \ \Ex_\mu\insquare{f_S(c,\nsucc)} \cdot \inparen{1+ \Ex_\mu\insquare{\frac{1}{Z(\succ)}}\cdot \frac{\tau\phi}{\rho(1-\phi)}}.
            \end{align*}
            Since $c\in G_1$, the position of candidates in $G_1$ only reduces, which in turn increases the score, it follows that  $\Exp_{\widehat{\mu}}\left[f_S(c, \nsucc)\mid \succ\right]\geq \tau_{j,S} = f_S(c,\nsucc).$
                   Taking the expectation over $\succ\ \sim \mu$, we get that 
            \begin{align*}
                \Exp_{\widehat{\mu}}\left[f_S(c, \nsucc)\right] 
                \geq \Ex_\mu\insquare{f_S(c,\nsucc)}.
            \end{align*}

        \paragraph{Case B ($c\in G_2$):}
            Suppose $c\in G_2$.
            Let 
            \[
                j\coloneqq \pos_\succ(c).
            \]
            We will first consider the case where there is only one swap, i.e., $t=1$, and later generalize to multiple swaps.
            We can express $\Exp_{\widehat{\mu}}\left[f_S(c, \nsucc)\mid \succ\right]$ as follows.
            \begin{align*}
                &\Exp_{\widehat{\mu}}\left[f_S(c, \nsucc)\mid \succ\right]\\ 
                    &\quad = 
                    \sum_{\ell\in [m]\setminus[j]}\Pr\insquare{\text{$c$ is swapped from $j$ to $\ell$}}\cdot  \tau_{\ell,S} 
                    +\inparen{1-\sum_{\ell\in [m]\setminus[j]}\Pr\insquare{\text{$c$ is swapped from $j$ to $\ell$}} } \tau_{j,S}
                    \tag{The positions of candidates $c\in G_2$ only increases, and the definition of $\tau_{\ell,S}$}\\
                    &\quad = 
                    \sum_{\ell\in A\inparen{j}}\Pr\insquare{\text{$c$ is swapped from $j$ to $\ell$}}\cdot  \tau_{\ell,S} 
                    +\inparen{1-\sum_{\ell\in A\inparen{j}}\Pr\insquare{\text{$c$ is swapped from $j$ to $\ell$}} } \tau_{j,S}
                    \tag{Candidates in $G_2$ only swap positions with candidates in $G_1$ in one swap}\\
                    &\quad = 
                    \frac{1}{Z(\succ)} \sum_{\ell\in A\inparen{j}} \phi^{\ell-j} \cdot \tau_{\ell,S} 
                    +\inparen{1-\frac{1}{Z(\succ)} \sum_{\ell\in A\inparen{j}} \phi^{\ell-j}} \tau_{j,S}.
                    \tagnum{By construction in the swapping-based bias model; \cref{def:swapping}} 
            \end{align*}
            We can lower bound the above expression as follows 
            \begin{align*}
                \Exp_{\widehat{\mu}}\left[f_S(c, \nsucc)\mid \succ\right] 
                    &= \frac{1}{Z(\succ)} \sum_{\ell\in A\inparen{j}} \phi^{\ell-j} \cdot \tau_{\ell,S} 
                    +\inparen{1-\frac{1}{Z(\succ)} \sum_{\ell\in A\inparen{j}} \phi^{\ell-j}} \tau_{j,S}\\
                    &\geq  \inparen{1-\frac{1}{Z(\succ)} \sum_{\ell\in A\inparen{j}} \phi^{\ell-j}} \tau_{j,S} \tag{Using that $\phi,\tau_{\ell,S},Z(\succ)\geq 0$ for all $\ell\in [m]$}\\
                    &\geq  \inparen{1-\frac{1}{Z(\succ)} \inparen{\phi + \phi^2 + \dots }} \tau_{j,S} 
                    \tag{Using that $\phi,\tau_{\ell,S},Z(\succ)\geq 0$ for all $\ell\in [m]$}\\
                    &=  \inparen{1-\frac{1}{Z(\succ)}\cdot \frac{\phi}{1-\phi}} \tau_{j,S}.
            \end{align*}
            Taking the expectation over $\succ\ \sim \mu$, we get that 
            \begin{align*}
                \Exp_{\widehat{\mu}}\left[f_S(c, \nsucc) \right] \ \ 
                    &\Stackrel{}{\geq}\ \  \Ex_\mu\insquare{\tau_{j,S}} - \frac{\phi}{1-\phi} \cdot \Ex_\mu\insquare{\frac{\tau_{j,S}}{Z(\succ)}}\\
                    &\Stackrel{}{\geq}\ \  \Ex_\mu\insquare{\tau_{j,S}} - \frac{\phi}{1-\phi} \cdot \Ex_\mu\insquare{\frac{\tau}{Z(\succ)}} \tag{Using that $\tau_{j,S}\leq \tau$}\\ 
                    &\Stackrel{}{\geq}\ \  \Ex_\mu\insquare{f_S(c,\succ)} - \frac{\phi\tau}{1-\phi} \cdot \Ex_\mu\insquare{\frac{1}{Z(\succ)}}\tag{Using the definition of $j$}\\
                    &\Stackrel{\eqref{eq:lbddd}}{\geq}\ \  
                    \Ex_\mu\insquare{f_S(c,\succ)}\inparen{1 - \frac{\phi\tau}{\rho(1-\phi)} \cdot \Ex_\mu\insquare{\frac{1}{Z(\succ)}}}.
            \end{align*}
            Since $c\in G_2$, the position of candidates in $G_2$ only increases, and increasing the position does not increase the score, it follows that $\Exp_{\widehat{\mu}}\left[f_S(c, \nsucc)\mid \succ\right]\leq \tau_{j,S} = f_S(c,\nsucc).$
            Taking the expectation over $\succ\ \sim \mu$, we get that 
            \begin{align*}
                \Exp_{\widehat{\mu}}\left[f_S(c, \nsucc)\right] 
                \leq 
                \Ex_\mu\insquare{f_S(c,\nsucc)}  
            \end{align*}

        \paragraph{Completing the proof.}
            Thus, across both cases, the following inequalities hold when $t=1$:
            \begin{align*}
                & \quad \inparen{1-\Ex_\mu\insquare{\frac{1}{Z(\succ)}}\cdot \frac{\tau\phi}{\rho(1-\phi)}} 
                \cdot \Ex_\mu\insquare{f_S(c,\nsucc)} \\
                \leq & \quad
                \Exp_{\widehat{\mu}}\left[f_S(c, \nsucc)\right] 
                \leq 
                \inparen{1+\Ex_\mu\insquare{\frac{1}{Z(\succ)}}\cdot\frac{\tau\phi}{\rho(1-\phi)}} 
                \cdot\Ex_\mu\insquare{f_S(c,\nsucc)}.
            \end{align*} 
            Consider a draw of $\nsucc$, say $\nsucc_1$, obtained after performing one swap on $\succ$.
            Replacing $\succ$ by $\nsucc_1$ in the above proof, we get bounds on utility with the swapping-based model with $t=2$ swaps.
            Repeating this argument $t$ times, when $\phi\in (0,t^{-1})$, we get the following bounds.
            \begin{align*}
                & \quad \inparen{1-\Ex_\mu\insquare{\frac{1}{Z(\succ)}} \cdot \frac{O(t\tau\phi)}{\rho(1-\phi)}} 
                \cdot \Ex_\mu\insquare{f_S(c,\nsucc)} \\
                \leq & \quad
                \Exp_{\widehat{\mu}}\left[f_S(c, \nsucc)\right] 
                \leq 
                \inparen{1+\Ex_\mu\insquare{\frac{1}{Z(\succ)}}\cdot \frac{O(t\tau\phi)}{\rho(1-\phi)}} 
                \cdot\Ex_\mu\insquare{f_S(c,\nsucc)}.
            \end{align*}

    \subsubsection{Bounding $\alpha$ for the Swapping-Based Model}
        \label{sec:proofof:lem:swap_alpha}
        In this section, we bound the parameter $\alpha$ for any multiwinner scoring function when the biased generative model is the swapping-based model (\cref{def:swapping}) and the latent generative model is $\mu$ is any generative model that satisfies the following condition (for some parameter $\rho>0$)
        \[
            \min_{c\in C}\Ex_\mu\insquare{f_{C\setminus \inbrace{c}}(c,\succ{})}\geq \rho.
            \yesnum\label{eq:condition_on_latent_generating_model}
        \]
        Concretely, we prove the following bound on $\alpha.$
        \begin{lemma}[\textbf{Bounds on $\alpha$ for the swapping-based model}]\label{lem:swap_alpha}
            Let $F = \sum_{v\in V} f(\cdot ,\succ_v)$ be a latent multiwinner score function. 
            Suppose $\mu$ be a generative model such that $F$ is order-preserving with respect to $\mu$ (\cref{def:bias}) and that satisfies \cref{eq:condition_on_latent_generating_model} for some $\rho>0$.
            For any numbers $\phi\in (0,t^{-1})$ and the generative model $\wh{\mu}$ in \cref{def:swapping} with parameters $\mu$ and $\phi$,
            $\alpha$ satisfies the following bound 
            \[
                \alpha\in 
                \alpha_0 \cdot \inparen{1\pm \Ex_{\mu}\insquare{\frac{1}{Z(\succ)}}\cdot \frac{O(t\phi)}{1-\phi} \cdot \frac{\tau}{\rho}}.
            \]
            Where $\alpha_0\coloneqq (\tau(f))^{-1}\cdot \min_{c\in M} \Ex_{\mu}\insquare{f_{C\setminus \inbrace{c}}(c,\not\succ)}$ and $Z(\succ)$ is the normalizing constant corresponding to preference list $\succ$, as defined in \cref{def:swapping}.
        \end{lemma}
        \noindent 
        Similar to \cref{lem:swap_order_preservation}, the above lemma also does not fix a specific generative model of latent preference lists $\mu$.
        The parameter $\alpha$ depends on $\mu$ via the terms $\alpha_0$ and $\Ex_{\mu}\insquare{\frac{1}{Z(\succ)}}$.
        As discussed in \cref{sub:general_model}, in general, we expect $\Ex_{\mu}\insquare{\frac{1}{Z(\succ)}}$ to be of the order of $\Omega\inparen{m\phi}$.
        If this holds, then \cref{lem:swap_alpha}, implies that $\alpha\in \alpha_0\inparen{1 \pm O\inparen{m^{-1}t}}.$
        The proof of \cref{lem:swap_alpha} appears in Section~\ref{sec:proofof:lem:swap_alpha}.
        \begin{proof}[Proof of \cref{lem:swap_alpha}]
            \cref{lem:swap_alpha} follows as a corollary of \cref{lem:swap_simple}, which was proved in \cref{sec:proofof:lem:swap_order_preservation}.
            Consider any $c\in M$, selecting the set $S$ in \cref{lem:swap_simple} as $C\setminus \inbrace{c}$, we get that 
            \[
                \Ex_\mu\insquare{f_{C\setminus \inbrace{c}}(c, \succ)}\cdot (1 - \lambda)
                     \leq 
                    \Exp_{\widehat{\mu}}\left[f_{C\setminus \inbrace{c}}(c, \nsucc) \right] 
                     \leq 
                    \Ex_\mu\insquare{f_{C\setminus \inbrace{c}}(c, \succ)}\cdot (1 + \lambda).
            \]
            where $\lambda\coloneqq \Ex_{\mu}\insquare{\frac{1}{Z(\succ)}}\cdot O\inparen{\frac{t\phi}{1-\phi}\cdot \frac{\tau}{\rho}}$.
            The lemma follows by taking the minimum overall $c\in M$.
        \end{proof}

\section{Impossibility Results for \Cref{problem:representational_optimal}}
\label{sec:impossible}

In this section, we provide a necessary condition for \Cref{problem:representational_optimal}.
Roughly speaking, we show that if $n$ is not sufficiently large, representational constraints cannot recover an (approximate) optimal solution.
In combination with our algorithmic result, we observe that under the utility-based model, the number of voters necessary for SNTV or $\ell_1$-CC to find a close-to-optimal committee can be much larger than for Borda (\Cref{remark:robustness}). 

Again, we let $F = \sum_{v\in V} f(\cdot ,\succ_v)$ be a latent multiwinner score function voting, and let $(\mu,\widehat{\mu})$ be a generative model defined in \Cref{def:bias,def:latent}.
We propose the following definition.

\begin{definition}[\bf{Contribution of $G_2\cap S^\star$}]
    \label{def:contribution}
    Let $0\leq \ell\leq k$ be an integer.
    Define 
    $$
    r_{\ell} \coloneqq \Ex_{\mu}\left[\max_{S\subseteq C\setminus (G_2\cap S^\star): |S|=k, |S\cap G_2| = \ell} \frac{F(S)}{\OPT}\right]$$
    to be the expected maximum ratio between the scores of an optimal committee $S$ in $C\setminus (G_2\cap S^\star)$ with $|S\cap G_2| = \ell$ and $S^\star$.
\end{definition}

\noindent
$1-r_{\ell}$ measures the contribution of candidates in $G_2\cap S^\star$ to $F$ for all subsets $S$ with $|S\cap G_2| = \ell$.
If $r_{\ell}$ is close to 1, we can safely ignore candidates in $G_2\cap S^\star$ and can still find a near-optimal solution from the remaining candidates, even with constraint $\calK(\ell)$.
Otherwise, if $r_{\ell}$ is not close to 1, we may need the candidates $G_2\cap S^\star$ for a high-score committee.
Then if all candidates in $G_2$ are indistinguishable under $\widehat{\mu}$, it is likely that $\widehat{S}_{\ell}$ loses a lot.

\begin{remark}[\bf{Discussion on the scale of $r_{\ell}$}]
\label{remark:rl}

Usually, $r_{\ell}$ becomes smaller as the difference $|\ell - |G_2\cap S^\star||$ becomes larger.
For instance, for SNTV under \Cref{ex:utility_generative} with $\mu$ being the uniform distribution on interval $[0,1]$, suppose $\omega_c = 1$ for all $c\in S^\star$ and $\omega_c = 0.5$ for all $c\in C\setminus S^\star$.
We can verify that $\Ex_{\mu}[F(c)]\geq 2\Ex_{\mu}[F(c')]$ for any $c\in S^\star$ and $c'\in C\setminus S^\star$, which implies that $$r_{\ell} \leq 1-\frac{|\ell - |G_2\cap S^\star||}{2k}$$ when $n,m$ are sufficiently large.
This observation matches the intuition that a representational constraint $\calK(\ell)$ with $\ell \approx |G_2\cap S^\star|$ may be better for debiasing, also observed in~\cite{celis2020interventions,mehrotra2022selection}.

Another consequence of this example is that we find $F(\widehat{S}) \ll \OPT$ is possible.
If the bias parameter $\theta < 0.5$, we know that $\widehat{S}$ is likely to contain only candidates in $G_1$, which results in $$F(\widehat{S}) \leq r_0\cdot \OPT \leq (1-\frac{|G_2\cap S^\star|}{2k})\cdot \OPT,$$ i.e., $\widehat{S}$ is far from optimal.
\end{remark}

\smallskip 
\noindent
For each $j\in [m]$, let $i_{\succ,j}$ be the $j$-th candidate in $\succ$.
For each $j\in [m]$, define $\tau_j(f)$ to be the score  of $i_{\succ,j}$, i.e., $\tau_j(f)\coloneqq f(\inbrace{i_{\succ,j}},\succ)$ for all $\succ\in \prefs{C}$.
Let $\succ\ \in \prefs{C}$ be a fixed preference list and define $\tau_{\min}(f) \coloneqq \min\inbrace{\tau_j(f)\colon j\in [m],\ \tau_j(f)>0}$ to be the smallest non-zero candidate value of $f$.
Our main impossibility result is summarized as follows.

\begin{theorem}[\bf{Impossibility result for \Cref{problem:representational_optimal}}]
    \label{thm:main_impossibility}
    Let $F\colon 2^C\to \R_{\geq 0}$ be a multiwinner score function.
    Let $\mu, \widehat{\mu}$ be generative models of latent/biased preference lists such that every $\succ_v$/$\nsucc_v$ is i.i.d. drawn from $\mu$/$\widehat{\mu}$. 
    Let $$\zeta \coloneqq \max_{c\in G_2} \Ex_{\widehat{\mu}}\left[f(c, \nsucc)\right].$$
    If $\zeta < \tau_{\min}(f)$, $k = o(|G_2|)$, and $n = o\inparen{\frac{\tau_{\min}(f)}{m \zeta}}$, with probability at least 0.9, $$F(\widehat{S}_{\ell}) \leq (r_{\ell}+o(1))\cdot \OPT$$ holds for every $0\leq \ell\leq k$.
\end{theorem}

\begin{proof}

By the definition of $s(f)$ and $\zeta$, we know that for every $c\in G_2$,
\[
\Pr_{\widehat{\mu}}[f(c,\nsucc) > 0] \leq \frac{\zeta}{s(f)}.
\]
Consequently, we have that
\[
\Pr_{\widehat{\mu}}[\widehat{F}(c) = 0, \forall c\in G_2] \geq \inparen{1-\frac{\zeta}{s(f)}}^{m n}.
\]
Then since $n=o\inparen{\frac{s(f)}{m \zeta}}$, we have
\[
\Pr_{\widehat{\mu}}[\widehat{F}(c) = 0, \forall c\in G_2] \geq 1-o(1).
\]
Conditioned on the event that $\widehat{F}(c) = 0$ holds for all $c\in G_2$, $\widehat{S}_{\ell}$ exactly selects $\ell$ candidates from $G_2$ since $G_2$ does not contribute to $\widehat{F}$. 
Moreover, we can not distinguish candidates in $G_2$ and can only randomly select $\ell$ candidates from $G_2$ in $\widehat{S}_{\ell}$, which results in $o(1)$ probability to select any candidate from $G_2\cap S^\star$ since $|G_2\cap S^\star|\leq k = o(|G_2|)$.
Overall, with probability at least $1-o(1)$, we have that 
\[
F(\widehat{S}_{\ell}) \leq \max_{S\subseteq C\setminus (G_2\cap S^\star): |S|=k, |S\cap G_2| = \ell} F(S).
\]
Thus, we conclude the theorem by the definition of $r_{\ell}$.

\end{proof}

\noindent
Observe that the above theorem considers a more general generative model $(\mu, \widehat{\mu})$ that does not require order preserving properties in \Cref{def:bias,def:latent}.
Roughly, \Cref{thm:main_impossibility} indicates that the required number of voters is $n = \Omega(\frac{\tau_{\min}(f)}{m \zeta})$ which is non-trivial when $\frac{\tau_{\min}(f)}{\zeta} \gg m$.
Note that $\tau_{\min}(f) \leq \tau_1(f)$ always holds but usually has at most a $\poly(m)$ gap.
For instance, $\tau_{\min}(f) = \tau_1(f)$ for SNTV and Bloc, and $\tau_{\min}(f) = \frac{\tau_1(f)}{m-1}$ for Borda and $\ell_1$-CC.
Also, note that $\zeta$ is comparable to $\min_{c\in M} \Ex_{\wh{\mu}}\insquare{f_{C\setminus \{c\}}(c,\not\succ)}$, specifically, $\zeta \leq \min_{c\in M} \Ex_{\wh{\mu}}\insquare{f_{C\setminus \{c\}}(c,\not\succ)}$ for SNTV and $\ell_1$-CC.
Thus, the scale of $\frac{\tau_{\min}(f)}{\zeta}$ and $\frac{1}{\alpha}$ in \Cref{def:smooth} could be the same order, e.g., $\frac{\tau_{\min}(f)}{\zeta} \approx \Omega(\frac{1}{\alpha\cdot \poly(m)})$, which leads to a required number of voters $n = \Omega(\frac{1}{\alpha\cdot \poly(m)})$.
Combining with the bound on $\alpha$ in Section~\ref{sec:applications_of_result}, this required voter number is interesting; see the following comparison.

\begin{remark}[\bf{Comparison of robustness between different rules}]
\label{remark:robustness}
Specifically, under the utility-based generative model, we have that $n = \Omega(\frac{\theta^{-m+1}}{\poly(m)})$ for SNTV and $\ell_1$-CC and $n = \Omega(\frac{\theta^{-1}}{\poly(m)})$ for Borda to achieve a near-optimal solution $\widehat{S}_{\ell}$ with a representational constraint, say $F(\widehat{S}_{\ell}) \approx \OPT$.
Combining with \Cref{cor:algorithmic}, we know that the required voting number of \Cref{problem:representational_optimal} for SNTV and $\ell_1$-CC is at least $\Omega(\theta^{-m+1})$, which is much larger than the sufficient voting number $O(\theta^{-2} \cdot \poly(m))$ for Borda.
Thus, we may conclude that Borda is ``more robust'' than SNTV  and $\ell_1$-CC under the utility-based model.
\end{remark}

\newcommand{\hw}{\ensuremath{\wh{w}}} 

\section{Case Study: Utility-Based Generative Model of Latent and Biased Preferences}\label{sec:case_study}

    In this section, we present a self-contained discussion of our main result (\cref{thm:main_algorithmic}) within the utility-based model of preferences (\cref{ex:utility_generative,def:utility_bias}).
    In the utility-based model, each candidate $c$ has an \textit{intrinsic} utility $\omega_c\geq 0$.
    The true or \textit{latent} utility of candidate $c$ for voter $v$ is %
    \[  
        w_{v,c} =  \eta_{v,c} \cdot \omega_c, \yesnum
    \]
    where $\eta_{v,c}$ is a random variable drawn uniformly at random from $[0,1]$ and independent of $\eta_{v',c'}$ for any $v'\neq v$ and $c'\neq c$.
    These latent utilities, in turn, define the latent preference list of voters: for each voter $v$, $\succ_v$ is the list of candidates $c$ arranged in decreasing order of $w_{v,c}$. %
    The voters, however, do not observe their latent utilities for the candidates.
    Instead, they observe a (possibly) biased version of these latent utilities.
    These observed utilities are modeled by a bias parameter $\theta\in [0,1]$:
    for each candidate $c$ and voter $v$, the observed utility of $c$ for $v$ is
    \[
        \hw_{v,c} \coloneqq \begin{cases}
            w_{v,c} & \text{if } c\in G_1,\\
            \theta\cdot w_{v,c} & \text{if } c\in G_2.
        \end{cases}
        \yesnum\label{eq:obs_utility}
    \]

    \paragraph{Greedy is optimal without bias.} %
    Recall that our goal is to study the latent quality of solutions produced by algorithms that satisfy representational constraints in the presence of bias.
    For simplicity, suppose that the exact expected observed, biased marginal contributions $\Ex_{\wh{\mu}}\insquare{f_S(c,\nsucc)}$ are known for each $c\in C$ and $S\subseteq C$ (we relax this assumption later in this section).
    Consider the greedy algorithm that, in each iteration $t\in [k]$, selects the candidate $c$ that maximizes the expected marginal contribution to the set $S_t$ selected so far (i.e., $\arg\max_{c\in C\setminus S_t} \Ex_{\wh{\mu}}\insquare{f_{S_t}(c,\succ)}$).
    Without bias (i.e., $\theta=1$), it can be shown that this algorithm outputs the set $S_k$ that has the optimal latent utility among sets of size $k$ (i.e., $\Ex_{{\mu}}\insquare{f(S_k, \succ)} = \max_{S\subseteq C\colon \abs{S}=k} \Ex_{{\mu}}\insquare{f(S, \succ)}$).
    Intuitively, this is because of the following invariant: for any $c,c'\in C$
    \[
        \text{if $\omega_c > \omega_{c'}$,}\quad 
        \text{then,}\quad 
        \text{for any } S\subseteq C\setminus\inbrace{c,c'},\quad 
        \Ex_{{\mu}}\insquare{f_S(c,\succ)} > \Ex_{{\mu}}\insquare{f_S(c',\succ)}.
        \yesnum\label{eq:invariant:new_section}
    \]
    This invariant itself holds because of the facts that: (1) if $\omega_c > \omega_{c'}$ then with probability strictly larger than $\frac{1}{2}$, $c$ appears before $c'$ in $\succ_v$ (for any $v$) and (2) $f$ is domination sensitive (\cref{def:score}).

    \paragraph{Constrained-greedy is optimal with bias.} %
        However, when there is bias (i.e., $\theta<1$), the above invariant does not hold and the greedy algorithm may achieve a low latent utility.
        Nevertheless, because the latent utilities of all candidates in the same group are uniformly reduced (\cref{eq:obs_utility}), a group-wise version of the invariant holds:
            for any two candidates $c,c'$ in the \textit{same} group ($G_1$ or $G_2$) %
            for any $c,c'\in C$
        \[
            \text{if $\omega_c > \omega_{c'}$,}\quad 
            \text{then,}\quad 
            \text{for any } S\subseteq C\setminus\inbrace{c,c'},\quad 
            \Ex_{\wh{\mu}}\insquare{f_S(c,\nsucc)} > \Ex_{\wh{\mu}}\insquare{f_S(c',\nsucc)}.
            \yesnum\label{eq:invariant:groupwise}
        \]
        This property enables us to design an algorithm that, given $k$ and the expected biased marginal contributions utilities of all candidates (i.e., $\Ex_{\wh{\mu}}\insquare{f_S(c,\nsucc)}$ for all $c\in C$ and $C\subseteq S$) outputs a size-$k$ subset satisfying representational constraints that has optimal latent utility.
        Let $S^\star$ be any size-$k$ subset maximizing the expected latent utility.
        Let $\ell=\sabs{S^\star\cap G_2}.$
        It can be shown that the algorithm that first greedily selects $\ell$ candidates from $G_2$ and, then, greedily selects $k-\ell$ candidates from $G_1$ achieves optimal expected latent utility (\cref{algo:main_algorithm}).
        This is because (1) for any $S$, the top $t$ candidates in $G_1$ (respectively $G_2$) by expected observed scores $\Ex_{\wh{\mu}}\insquare{f_S(c,\nsucc)}$ is the same as the top $t$ candidates in $G_1$ (respectively $G_2$) by expected latent scores $\Ex_{{\mu}}\insquare{f_S(c,\succ)}$ and (2) $S^\star$ consists of the top $k-\ell$ (respectively $\ell$) candidates from $G_1$ (respectively $G_2$) by expected latent scores.
        Where the first fact is implied by \cref{eq:invariant:new_section,eq:invariant:groupwise} and the second fact holds because of \cref{eq:invariant:new_section}. %
        
    \sloppy
    Thus, if one has access to the expected marginal contributions of the candidates $\inbrace{\Ex_{\wh{\mu}}\insquare{f_S(c,\nsucc)}\colon c\in C\setminus S}$ (for all $S$) and $\ell=\sabs{S^\star\cap G_2}$, then one can execute the aforementioned algorithm, which outputs a size-$k$ subset satisfying the representational constraints specified by $\ell$ and maximizing the expected latent utility. %
    As mentioned before in \cref{sec:algorithmic}, $\ell$ can be estimated under natural assumptions on $\omega$. %
    Meanwhile, it remains to discuss how one has access to accurate \textit{expected} observed marginal contributions. 
    
    %\vspace{-0.1in}    

    \paragraph{Effect of $n$ on the effectiveness of representational constraints.}
        In \cref{problem:representational_optimal}, the input is $n$ samples of observed preference lists $\inbrace{\nsucc_v}_v$ generated as described above.
        For any $S\subseteq C$ and $c\in C\setminus S$, a natural estimate of $\Ex_{\wh{\mu}}\insquare{f_S(c,\nsucc)}$ is the sample mean $\frac{1}{n}\sum_{v\in V} f_S(c,\nsucc_v)$.
        One can prove an additive concentration bound for the resulting estimate (e.g., see \cref{lem:concentration} for a similar result). For any $\delta>0$, it holds that:
        \[
            \Pr\insquare{
                \forall_{c\in C},\ \ \forall_{S\subseteq C\setminus \inbrace{c}\colon \abs{S}=k-1}
                \abs{\frac{1}{n}\sum_{v\in V} f_S(c,\nsucc_v) - \Ex_{\wh{\mu}}\insquare{f_S(c,\nsucc)}}
                \geq 
                \errm{}
            }
            \leq \delta.\yesnum\label{eq:conc:explanation}
        \]
        Note that, for the above bound to be meaningful, the expected observed utilities $\Ex_{\wh{\mu}}\insquare{f_S(c,\nsucc)}$ must be large compared to the error term $\errm{}$.
        We parameterize the scale of expected observed utilities  using a parameter $\alpha^\circ$:\footnote{We use the superscript in $\alpha^\circ$ to differentiate it from $\alpha$ in \cref{def:smooth}.}
        \[
            \alpha^\circ(\omega,\theta,f) = \frac{1}{\tau_1(f)} \min_{c\in C} \Ex_{\wh{\mu}}\insquare{f_{C\setminus \{c\}}(c,\not\succ)},
        \]
        where $\tau_1(f)$ is a normalizing constant, defined as the maximum possible expected score $\Exp_{\mu}\left[f(c,\succ)\right]$ of a candidate. 
        Instead of taking the minimum over each candidate $c\in C$, one can show that it suffices to take a minimum over only certain candidates. We do so in \cref{def:smooth} and \cref{thm:main_algorithmic}.

        For the utility-based model, a lower bound on $\alpha^\circ$ is sufficient to ensure that, when $n$ is large enough, representational constraints achieve near-optimal latent utility.
        \begin{theorem}[\textbf{Informal version of \cref{thm:main_algorithmic} specialized to then utility-based model}]
            \label{thm:main_algorithmic:appendix}
            Let $F\colon 2^C\to \R_{\geq 0}$ be a multiwinner score function and $\mu, \widehat{\mu}$ be generative models corresponding to the utility-based model specified by $\omega$ and $\theta$. %
        For any $0<\eps,\delta<1$,
        if
        \[
            n \geq 
                \frac{
                    \poly(m)\cdot \log\inparen{1/\delta}
                }{
                    \poly\inparen{\eps\cdot \alpha^\circ(\omega,\theta,f)}
                },
        \]
        there is an algorithm that given $\ell{=}\abs{S^\star{\cap}G_2}$ and observed preferences $\inbrace{\nsucc_v}_{v\in V}$, outputs a size-$k$ subset $S{\in} \calK(\ell)$ such that 
        \[ \Pr_{\mu,\wh{\mu}}\insquare{F(S) \geq \inparen{1-\eps} \cdot  \OPT} \geq 1-\delta.\]
        \end{theorem}

    \paragraph{Generalizations.}
        To derive \cref{thm:main_algorithmic:appendix}, we used two main properties: (1) in the absence of bias, there is a greedy algorithm that achieves the optimal latent utility, and (2) that the order of candidates in any one group ($G_1$ or $G_2$) by their expected latent utilities is the same as their order by expected observed utilities (\cref{eq:invariant:new_section,eq:invariant:groupwise}).
        Neither of these conditions may be true beyond the utility-based model.

        \smallskip
        \noindent\textit{Order-preservation with respect to $\mu$ ensures that the greedy algorithm is optimal.} In more general models, the first condition in order-preservation with respect to $\mu$ (\cref{def:latent}) ensures that the greedy algorithm achieves the optimal latent utility in the absence of bias.
        The second condition in order-preservation with respect to $\mu$ (\cref{def:latent}) enables us to show that, even the constrained version of the greedy algorithm (which first selects $k-\ell$ candidates from $G_1$ and then selects $\ell$ candidates from $G_2$) achieves the optimal latent utility.

        \smallskip
        \noindent\textit{Order-preservation between $\mu$ and $\wh{\mu}$ bounds distance between orderings of candidates by observed and latent preferences.}
            Beyond the utility-based model, we may not be able to guarantee that the order of candidates in each group ($G_1$ and $G_2$) by their expected latent utilities is the same as their order by expected observed utilities.
            Indeed, this is not true in the swapping-based model, where there is constant $\beta$ such that if two candidates $c,c'$ in the same group have 
            latent scores $\Ex_{{\mu}}\insquare{f_S(c,\succ)}$ and $\Ex_{{\mu}}\insquare{f_S(c',\succ)}$ within a multiplicative factor $\beta$ (i.e., $\beta \Ex_{{\mu}}\insquare{f_S(c',\succ)}\leq \Ex_{{\mu}}\insquare{f_S(c,\succ)}\leq \frac{1}{\beta} \Ex_{{\mu}}\insquare{f_S(c',\succ)}$), then it is possible that (see \cref{ex:counter_example_opII}) 
            \[
                \Ex_{{\mu}}\insquare{f_S(c,\succ)} > \Ex_{{\mu}}\insquare{f_S(c',\succ)},
                \quad \text{but},\quad
                \Ex_{\wh{\mu}}\insquare{f_S(c,\nsucc)} < \Ex_{\wh{\mu}}\insquare{f_S(c',\nsucc)}.
            \]
            Order-preservation between $\mu$ and $\wh{\mu}$ (\cref{def:bias}) bounds the ``distance'' between the orderings of candidates in each group ($G_1$ and $G_2$) by their expected latent utilities and their order by expected observed utilities: it requires the relative order of two candidates $c,c'$ in the same group to be the same if their latent utilities are at least a factor of $\beta$ apart, i.e., if either $ \Ex_{{\mu}}\insquare{f_S(c',\succ)}\leq \beta\cdot \Ex_{{\mu}}\insquare{f_S(c,\succ)}$ or $\beta \cdot \Ex_{{\mu}}\insquare{f_S(c,\succ)}\geq \Ex_{{\mu}}\insquare{f_S(c',\succ)}$. 
        
%\vspace{-0.1in}    
\section{Tool to Study Smoothness and Effectiveness of Representational Constraints with Novel Bias Models}
    \label{sec:simulations_synthetic}
    
    In this section, we illustrate how one could use the code as a tool for  preliminary studies of the effectiveness of representational constraints with respect to novel bias models and multiwinner score functions, for which theoretical bounds may not be readily available.\footnote{The code for simulations is available at:
    \url{https://github.com/AnayMehrotra/Selection-with-Multiple-Rankings-with-Bias}}

    %\vspace{-0.2in}    
    \subsection{Implementation Details}\label{sec:implementation}
        
        The code takes as input oracles that (1) evaluate a multiwinner score function $F$ and (2) sample from generative models $(\mu,\wh{\mu})$.
        These oracles are specified by (python) functions.
        The code also takes $m$, $n$, and $k$ as input.
        First, for the specified $m$ and $k$, it outputs estimates $(\wt{\alpha},\wt{\beta},\wt{\gamma})\in [0,1]$ along with corresponding confidence intervals, which follow from a concentration inequality (\cref{lem:concentration}).
        This allows for theoretical estimates of the capabilities of representational constraints using our main result \cref{thm:main_algorithmic}.
        (Note that $\alpha,\beta,\gamma$ are independent of $n$ and $\ell$). 
        {Concretely, we numerically estimate $(\alpha,\beta,\gamma)$ as follows:
            since  the empirical averages are concentrated around the expectations $\mathbb{E}_\mu[f_S(c,\succ)]$ and $\mathbb{E}_{\hat{\mu}}[f_S(c,\not\succ)]$, we use this to compute $\alpha$ from its definition (\cref{def:smooth}).
            As $\beta$ and $\gamma$, given any $0\leq \beta\leq 1$, we  compute $0\leq \gamma\leq 1$ that satisfies the condition in \cref{def:bias} via binary search.}
        Second, given values of $n$, $m$, and $k$, the code estimates the fraction of the optimal score recovered by representational constraints for the given $F$ with respect to the given $(\mu,\wh{\mu})$.

    \subsection{Illustration of the Code: Models of Latent Preferences by  \cite{DBLP:conf/atal/SzufaFSST20} and the Swapping-Based Bias Model}\label{sec:illustration}
        
        In this section, to illustrate the use cases of our code, we analyze the effectiveness of the representational constraints with (1) the family of generative models proposed by \cite{DBLP:conf/atal/SzufaFSST20} and (2) using the swapping-based bias model (\cref{def:swapping}).

        Note that for any pair of generative models $(\mu,\wh{\mu})$ one can study the effectiveness of representational constraints with a multiwinner score function $F$ using our theoretical results, by computing the smoothness parameters for \cref{def:smooth}.
        Concretely, for the swapping-based model, one can use the bounds in \cref{lem:swap_alpha,lem:swap_order_preservation}.
        The usefulness of the code is in doing preliminary analysis on novel bias models, for which theoretical bounds may not be available.

        \paragraph{Setup.} 
        We vary the generative model $\mu$ across a subset of generative models provided by \cite{DBLP:conf/atal/SzufaFSST20}; namely, we consider the Single-Peaked model by Conitzer, the Mallows model, the Polya-Eggenberger Urn model, and the Impartial culture model.
        We fix $\wh{\mu}$ to be the generative model corresponding to the swapping-based bias model (\cref{def:swapping}).
        We fix the number of voters to be $m=50$, the size of the committee to be $k=10$, and vary the number of voters $n\in \inbrace{25, 50, 100}$.
        We fix the groups $G_1$ and $G_2$ to have equal size (i.e., $\abs{G_1}=\abs{G_2}$).
        We select the partitioning into groups uniformly at random.
        The fraction of possible executed swaps is specified by a parameter $\lambda\in [0,1]$.
         
        In each simulation, we fix a generative model $\mu$, a value of $n\in \inbrace{25,50,100}$, and a voting rule $F$ from SNTV and Borda.
        For fixed $\mu$, $n$, and $F$, we vary $\lambda$ over $[0,1]$.
        Given a $\mu$, $n$, $F$, and $\lambda$, we draw latent preferences $\inbrace{\succ_v\mid v\in V}$ i.i.d. from $\mu$.
        For each $v\in V$, we compute the maximum number of swaps $t_{\max}(v)$ that can be performed on the preference list $\succ_v$ before all candidates in $G_1$ are placed before all candidates in $G_2$.
        Concretely, $t_{\max}(v)$ is the Kendall-Tau distance between preference lists $\succ_v$ and $\succ^\star_v$, where $\succ^\star_v$ is the unique preference list that (1) ranks all candidates in $G_1$ before any candidate in $G_2$ and (2) satisfies for any $c,c'$ in the same group ($G_1$ or $G_2$) $c\succ_v c'$ if and only if $c\succ^\star_v c'$.

        In the swapping-based model, we arbitrarily fix $\phi=0.5$ for illustration.
        $\wh{\mu}$ is the generative model defined by the swapping-based biased model specified by $\phi=0.5$
        and which, for each $v\in V$, performs $t$ swaps, where 
        \[
            t=\lambda\cdot t_{\max}(v).
        \] 
        Recall that $\phi$ controls the average difference in the positions of swapped candidates.
        When $\phi$ is close to 0, with high probability, all swapped candidates are ``neighbors'' in the preference lists.
        At the other extreme, when $\phi$ is close to 1, candidates who are ``far'' in the preference lists are also swapped.
        
    \begin{figure}[ht!]
        \centering
        \subfigure[$n=25$, $m=50$ and $\mu$ is the Single-Peaked generative model by Conitzer]{
            \includegraphics[width=0.3\linewidth, trim={0cm 0cm 0cm 0cm},clip]{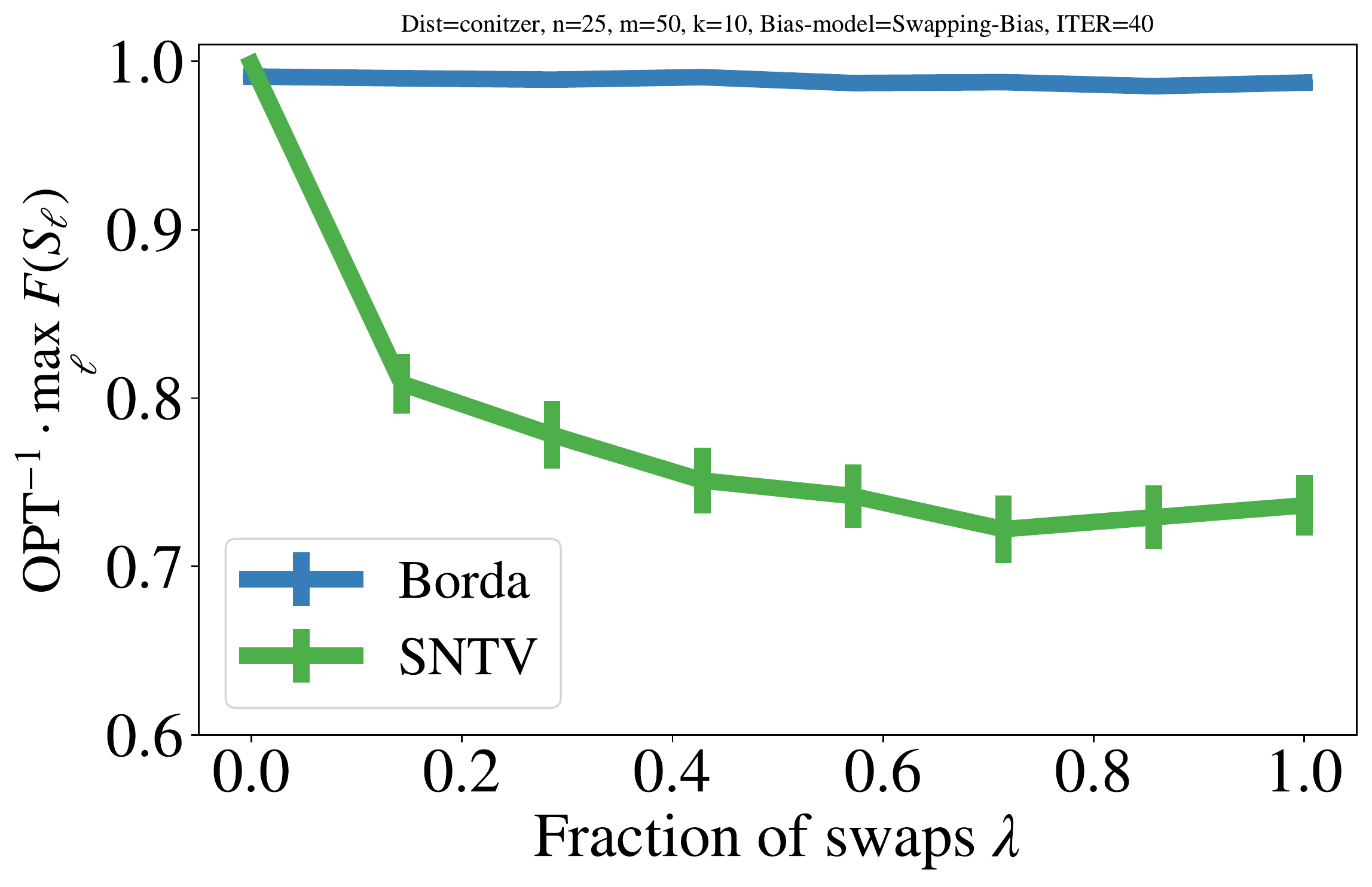}
        }
        \subfigure[$n=50$, $m=50$ and $\mu$ is the Single-Peaked generative model by Conitzer]{
            \includegraphics[width=0.3\linewidth, trim={0cm 0cm 0cm 0cm},clip]{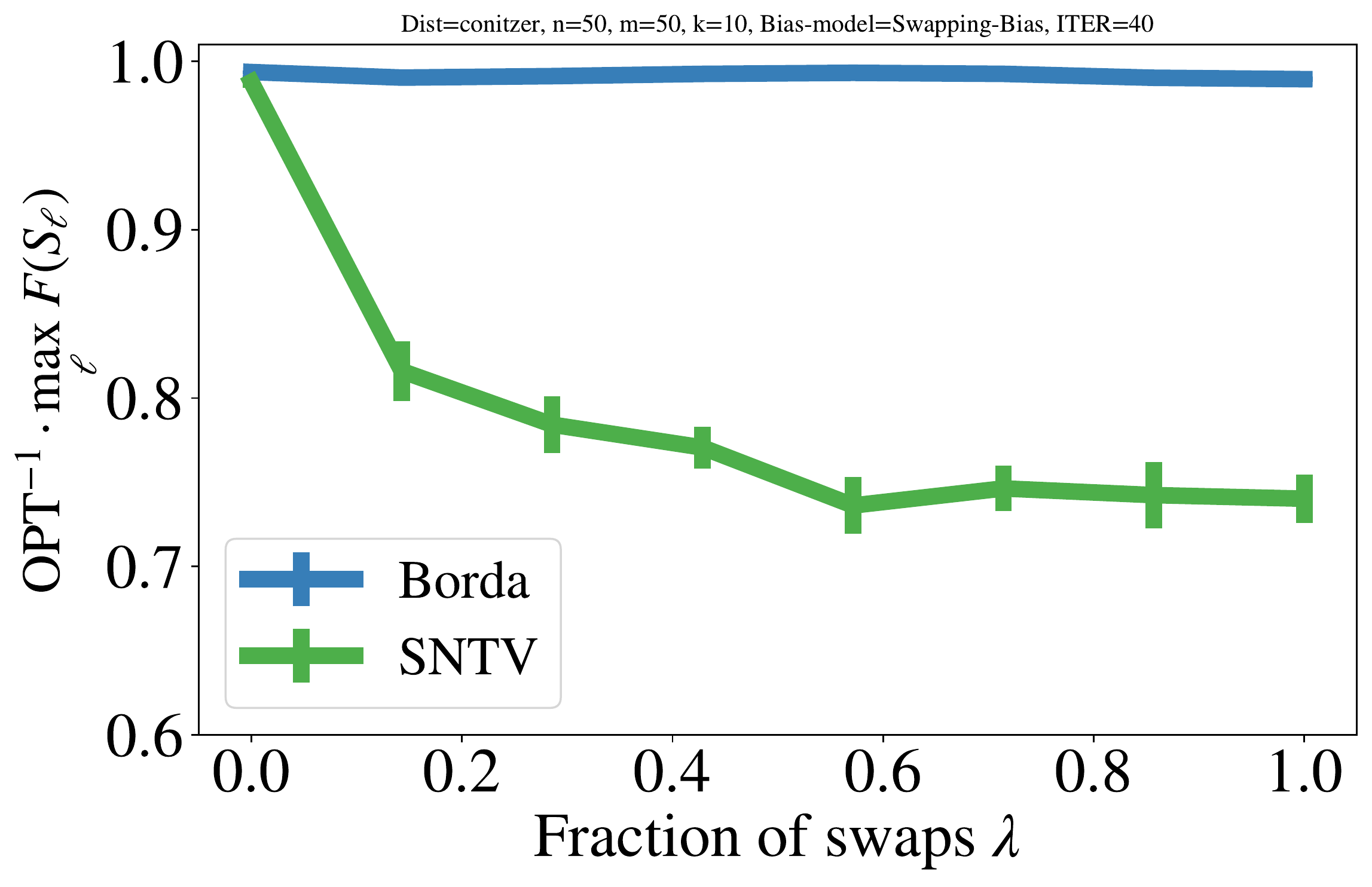}
        } %
        \subfigure[$n=100$, $m=50$ and $\mu$ is the Single-Peaked generative model by Conitzer]{
            \includegraphics[width=0.3\linewidth, trim={0cm 0cm 0cm 0cm},clip]{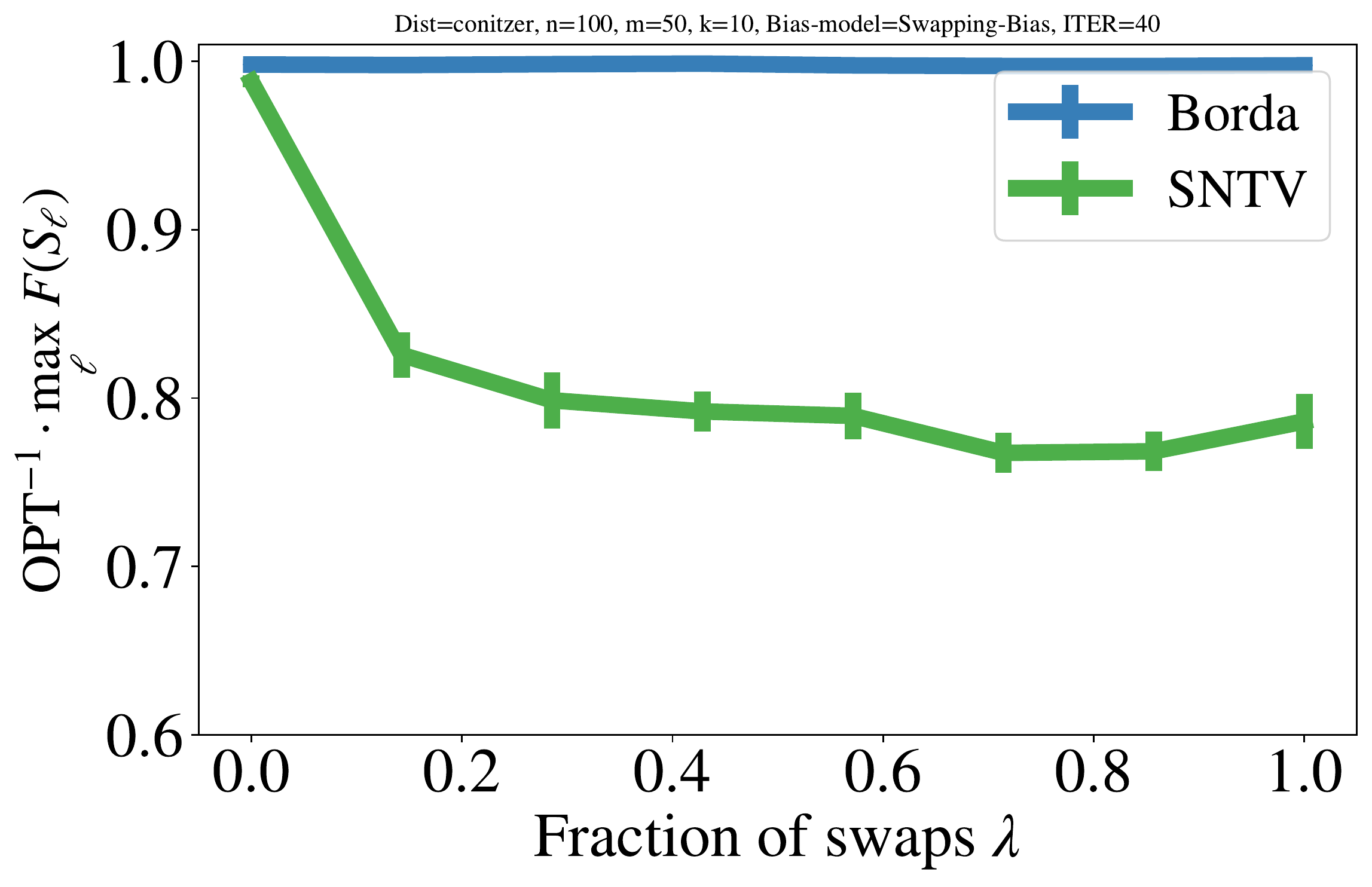}
        }\par %\vspace{-0.1in}
        \subfigure[$n=25$, $m=50$ and $\mu$ is the Mallows generative model]{
            \includegraphics[width=0.3\linewidth, trim={0cm 0cm 0cm 0cm},clip]{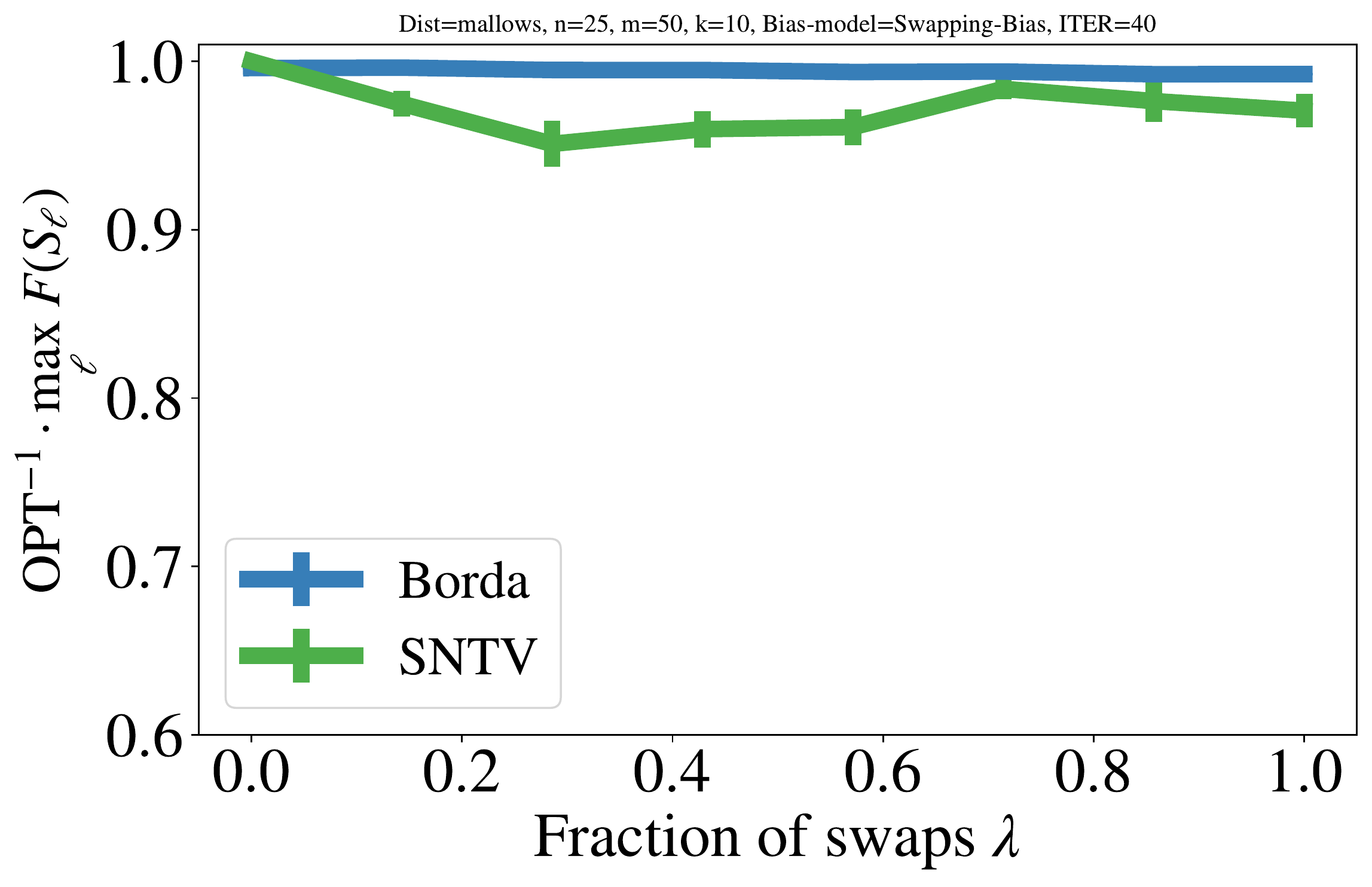}
        }
        \subfigure[$n=50$, $m=50$ and $\mu$ is the Mallows generative model]{
            \includegraphics[width=0.3\linewidth, trim={0cm 0cm 0cm 0cm},clip]{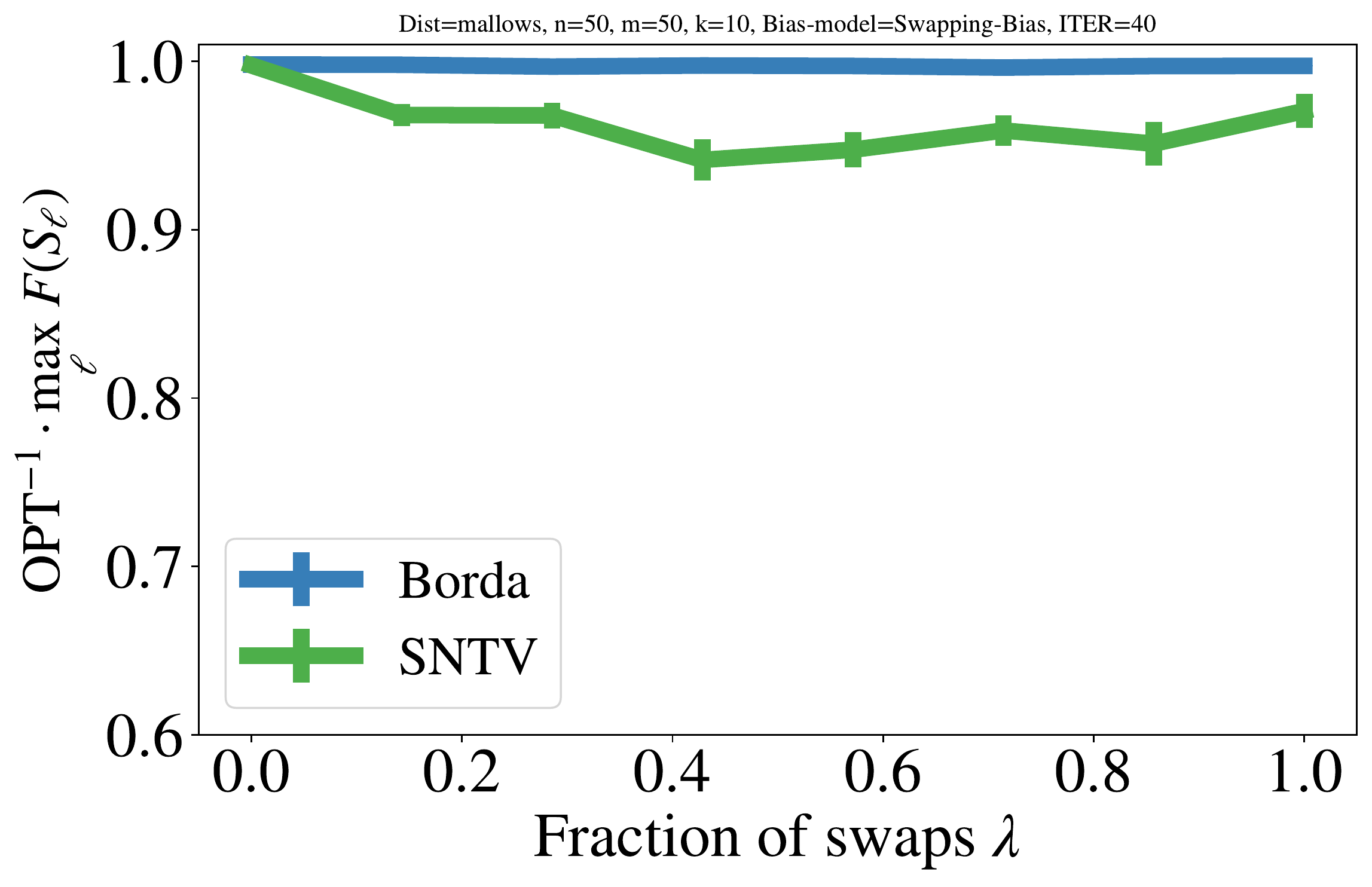}
        }
        \subfigure[$n=100$, $m=50$ and $\mu$ is the Mallows generative model]{
            \includegraphics[width=0.3\linewidth, trim={0cm 0cm 0cm 0cm},clip]{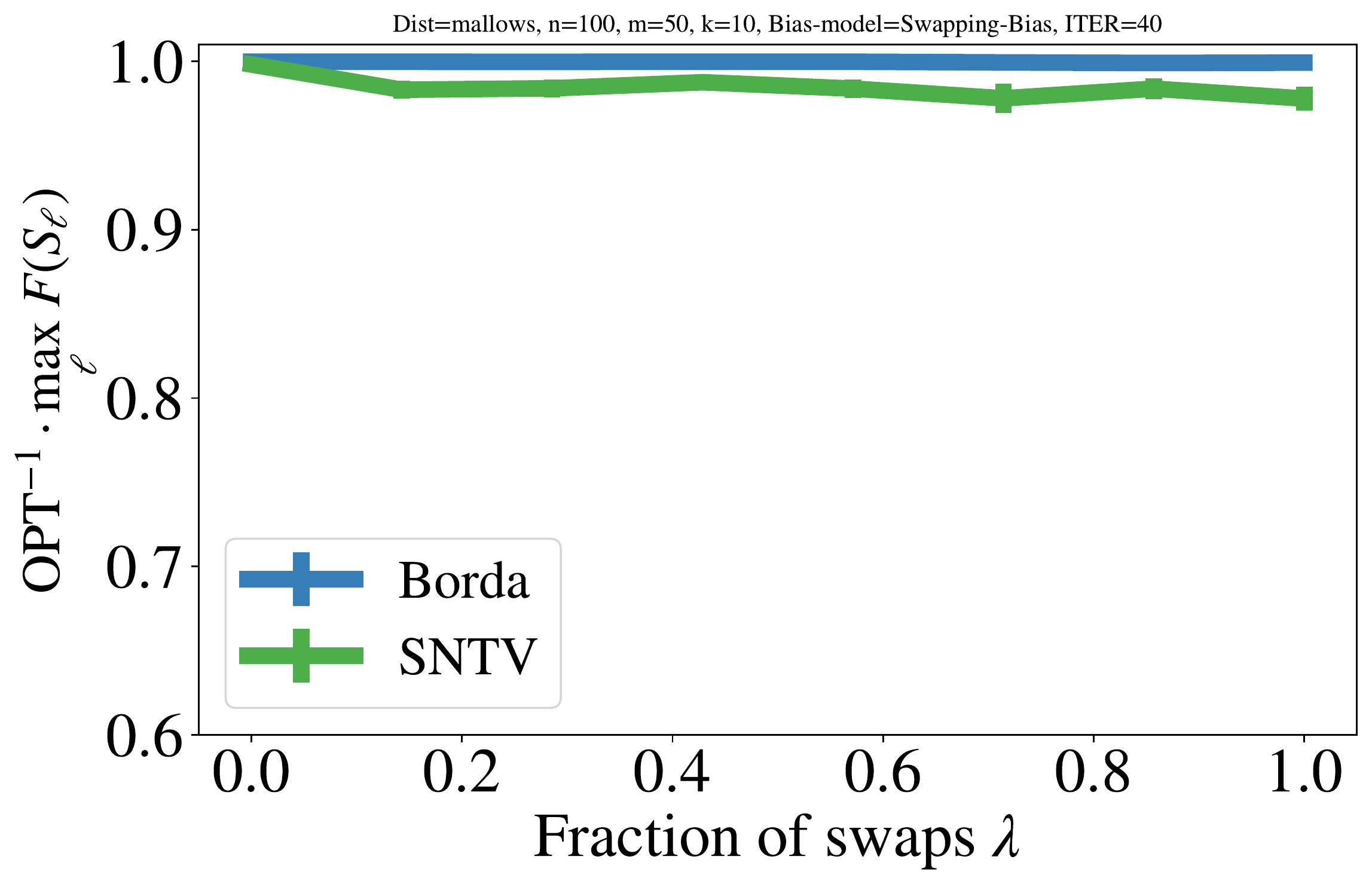}
        }\par %\vspace{-0.1in}
        \subfigure[$n=25$, $m=50$ and $\mu$ is the Polya-Eggenberger Urn generative model]{
            \includegraphics[width=0.3\linewidth, trim={0cm 0cm 0cm 0cm},clip]{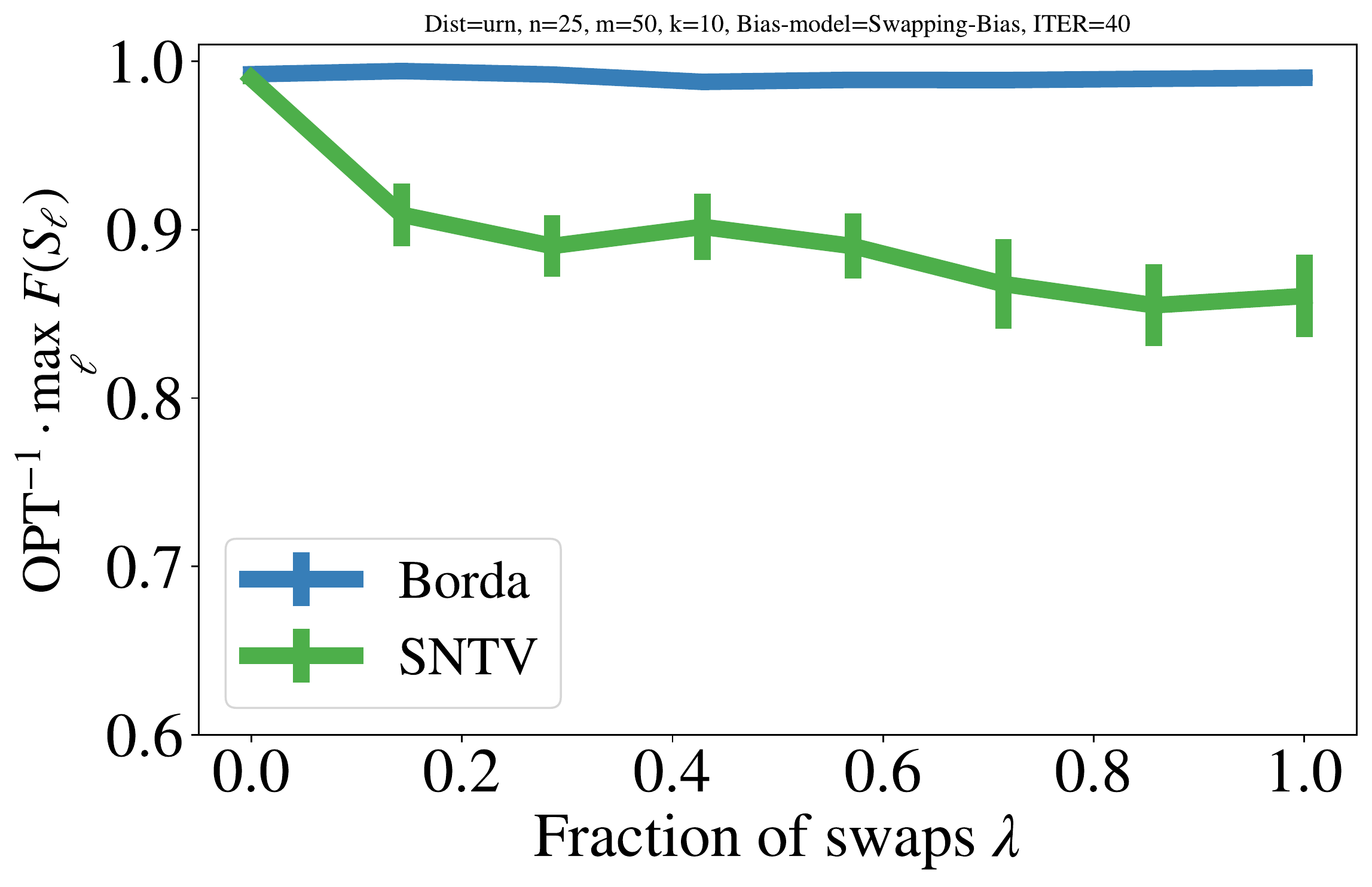}
        }
        \subfigure[$n=50$, $m=50$ and $\mu$ is the Polya-Eggenberger Urn generative model]{
            \includegraphics[width=0.3\linewidth, trim={0cm 0cm 0cm 0cm},clip]{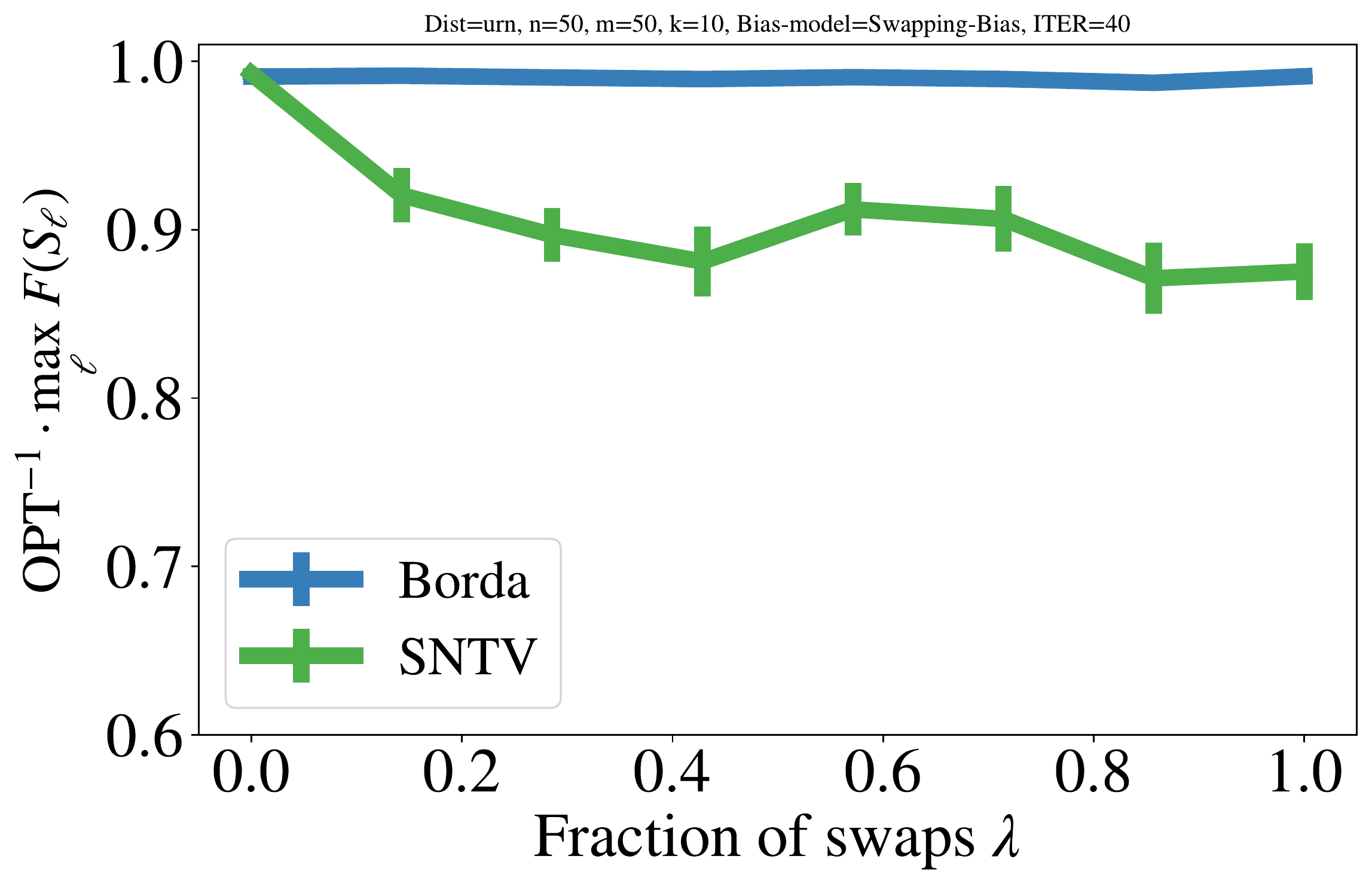}
        }
        \subfigure[$n=100$, $m=50$ and $\mu$ is the Polya-Eggenberger Urn generative model]{
            \includegraphics[width=0.3\linewidth, trim={0cm 0cm 0cm 0cm},clip]{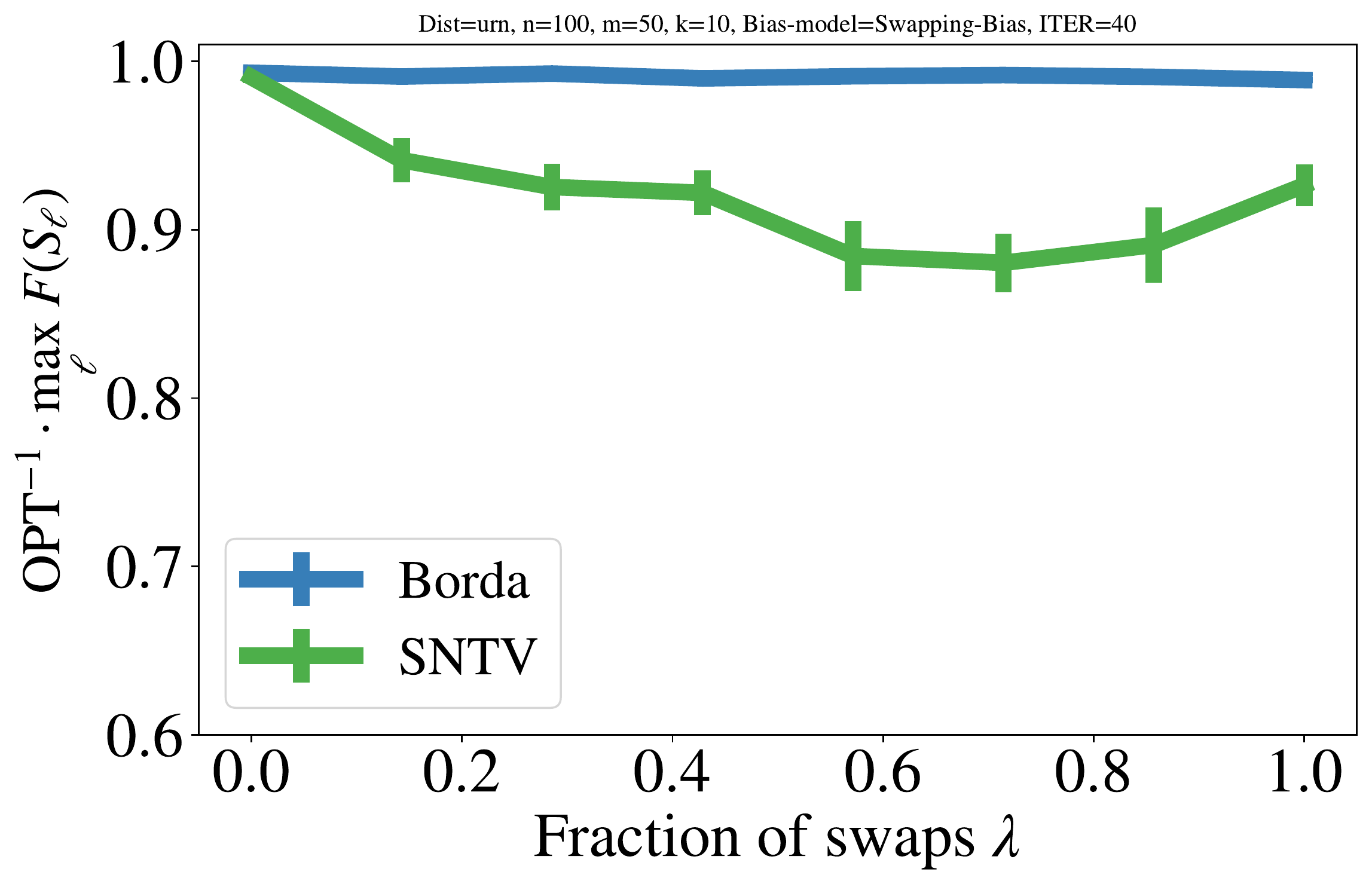}
        }\par %\vspace{-0.1in}
        \subfigure[$n=25$, $m=50$ and $\mu$ is the Mallows generative model]{
            \includegraphics[width=0.3\linewidth, trim={0cm 0cm 0cm 0cm},clip]{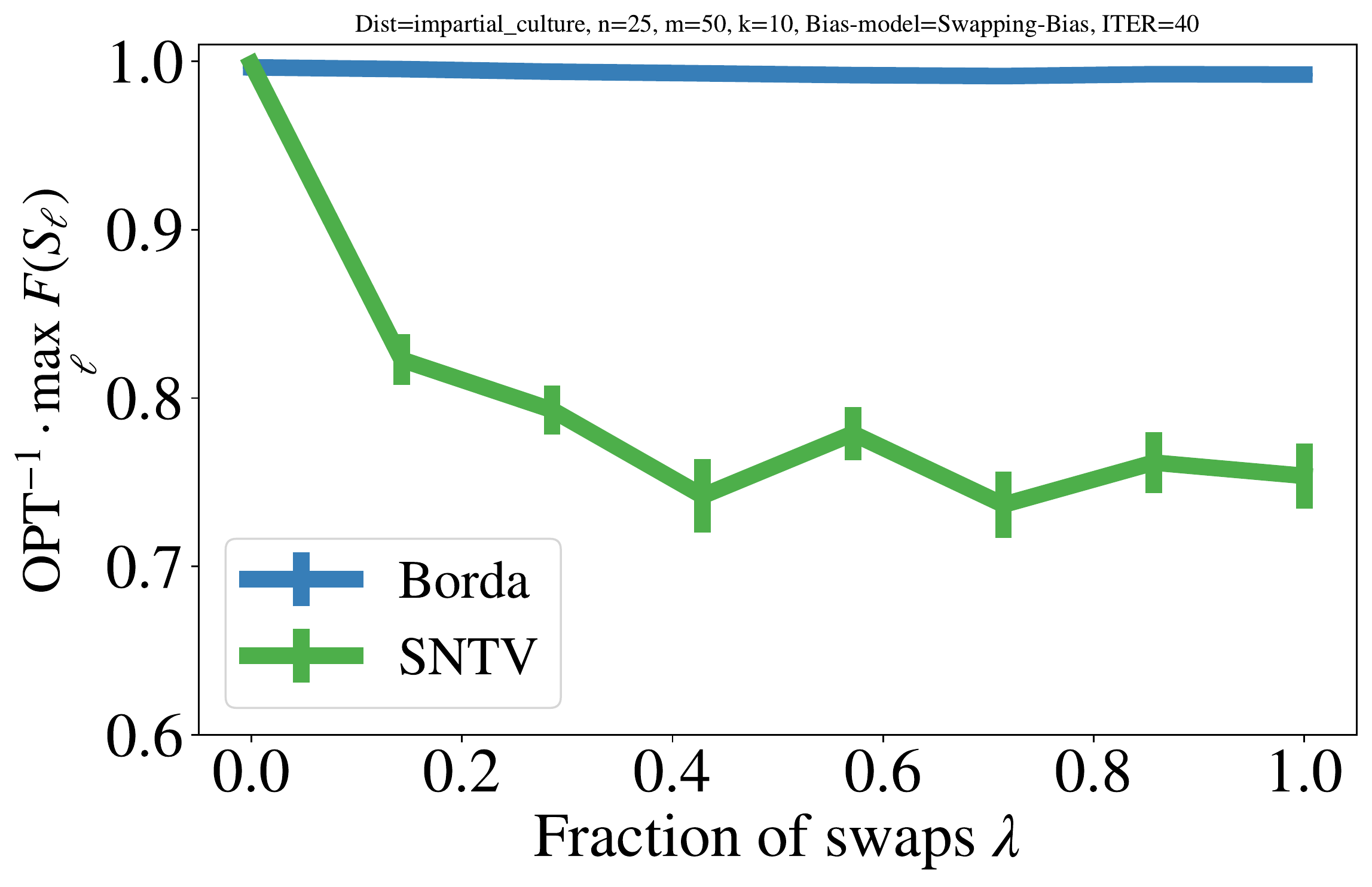}
        }
        \subfigure[$n=50$, $m=50$ and $\mu$ is the Impartial culture generative model]{
            \includegraphics[width=0.3\linewidth, trim={0cm 0cm 0cm 0cm},clip]{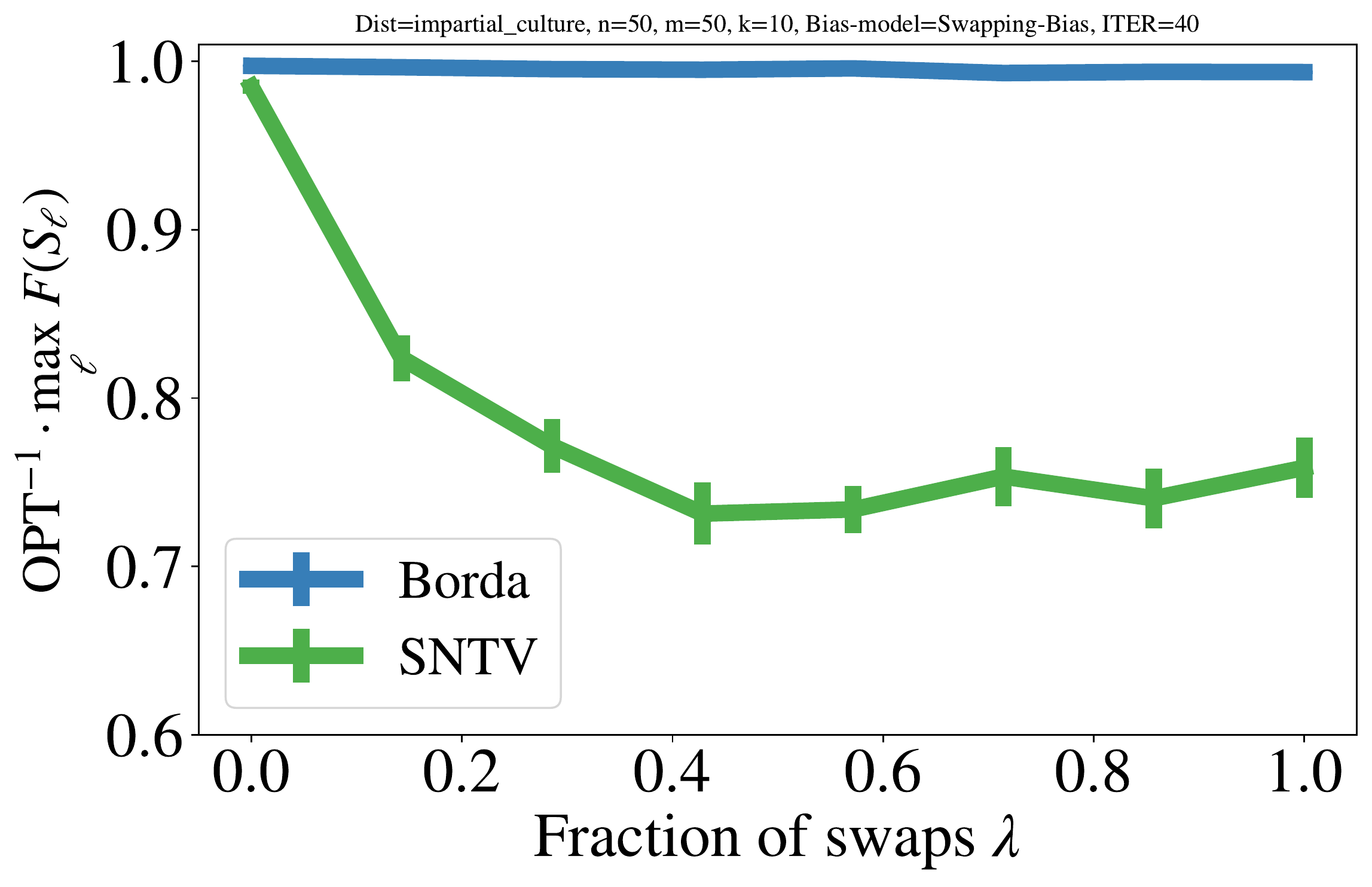}
        }
        \subfigure[$n=100$, $m=50$ and $\mu$ is the Impartial culture generative model]{
            \includegraphics[width=0.3\linewidth, trim={0cm 0cm 0cm 0cm},clip]{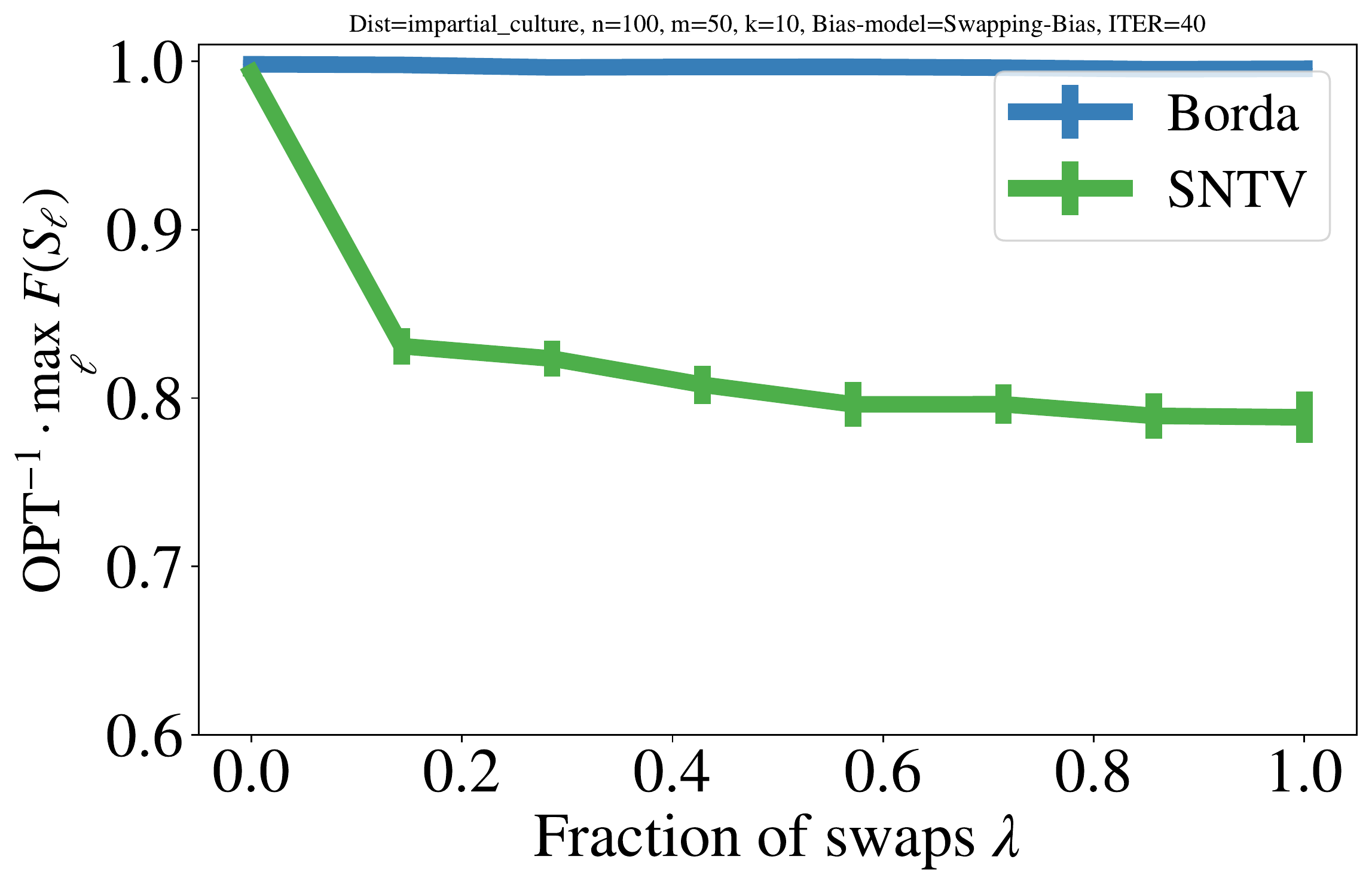}
        }\par  
        \caption{
            \textit{Simulations results with different families of generative models $\mu$:}
            The plots show the fraction of the score recovered by representational constraints with different preference aggregation functions $F$ and generative models $\mu$, under the swapping-based bias model with $\phi=0.5$ (\cref{def:swapping}).
            In all of these simulations, the number of candidates is $m=50$ and the size of the output committee is $k=10$.
            The number of candidates $n$ and the generative model $\mu$ vary and are specified with the sub-figures.
            The $y$-axis shows the fraction of the optimal score recovered by representational constraints.
            The $x$-axis shows the number of swaps $t$ allowed in the swapping-based model.
        }
        \label{fig:synthetic}
    \end{figure}
       
    \paragraph{Observations.} 
        The results appear in \cref{fig:synthetic}.
        We observe that across all choices of $\mu$, $n$, and the fraction of swaps $\lambda$: representational constraints recover a higher fraction of the optimal score
        with the Borda rule compared to the SNTV rule: for Borda, the fraction recovered is $>0.99$ across all simulations, whereas, for SNTV, it can be as low as $0.75$
        Further, for the SNTV rule, the fraction of the optimal score recovered by representational constraints increases with $n$ (for a given $\lambda$). %
        For the Borda rule, since the fraction of the optimal score recovered by representational constraints is already larger than $0.99$ differences across $n$ are small.

        These observations align with our theoretical results:
        From \cref{cor:algorithmic}, we expect representational constraints to have a higher effectiveness with the Borda rule compared to the SNTV rule.
        Similarly, from \cref{thm:main_algorithmic,thm:main_impossibility}, we expect the effectiveness of the representational constraints to increase with $n$.

\section{Conclusion and Future Work}
\label{sec:conclusion}

In this paper, we have investigated the effectiveness of representational constraints in the presence of bias, in a variant of the subset selection problem with rankings of items as input. 
To convert  multiple rankings into scores for subsets, we leverage ideas from multiwinner voting.
Extending a line of research on (re-)designing algorithms when the inputs might suffer from biases  \cite{kleinberg2018selection}, we demonstrate that representational constraints continue to have the power to improve the latent quality of the solution in the  setting of multiple  inputs and submodular (instead of modular) score functions. 
Our work brings out differences in the effectiveness of representational constraints depending on the (sub)modular score function used and can be used to guide the choice of multiwinner score functions for subset selection;
we further provide a tool to enable the latter. 
To carry out this analysis, we  develop a notion of smoothness of a submodular function that might be of further interest.

In addition to the already covered submodular multiwinner score functions, it would be interesting to extend our work to sequential rules such as STV, greedy CC, and greedy Monroe \cite{faliszewski2017multiwinner,DBLP:series/sbis/LacknerS23}.
Designing interventions for debiasing outcomes of these rules seems challenging 
as  it is not clear how representational constraints can be implemented here \cite{DBLP:conf/ijcai/CelisHV18}. 
As already discussed in the introduction, instead of assuming that voters rank the items,  each voter could also provide a numerical (utility) score for each candidate \cite{DBLP:journals/jacm/KleinbergPR04,DBLP:journals/ai/SkowronFL16}.
The class of submodular functions (\Cref{def:score}) we study contains some functions relevant in this setting: if all voters assign the same utility to their $i$th most preferred candidate for every $i$.
On a more general note, studying whether representational constraints continue to be effective for further classes of (non-separable) submodular functions is an interesting direction but is likely to require a new approach, as our smoothness definition and bias model assume that the function is separable.

{Finally, we remark that we study only one aspect of real-world selection problems and other aspects must be carefully considered to avoid possible negative impacts of intervention constraints.
Indeed, intervention constraints can help if they are used appropriately but may potentially also have negative consequences, e.g., imposing intervention constraints for one disadvantaged group may harm a different not-considered disadvantaged group.}

\section*{Acknowledgements}
    {This project is supported in part by NSF Awards (CCF-2112665 and IIS-2045951).
    NB is supported by the DFG project ComSoc-MPMS (NI 369/22).}

\printbibliography

\appendix
\addtocontents{toc}{\protect\setcounter{tocdepth}{2}}

\newpage

\section{Examples of Multiwinner Score Functions}
    \label{sec:example}
        
        We list two well-known families multiwinner score functions that are specific cases of \Cref{def:score}: the committee scoring functions and approval-based functions.
        
        \paragraph{Committee scoring functions.}
        A committee scoring function awards a score to each committee as
        \[
        \cs(S) = \sum_{v\in V} g(\pos_{\succ_{v}}(S))
        \]
        for some function $g: [m]^k \rightarrow \R_{\geq 0}$.
        The followings are some examples of committee scoring rules.
        
        \begin{example}[\bf{Examples of committee scoring functions}] 
        	Let $s\in \R_{\geq 0}^m$ be a vector. We have the following examples of committee scoring functions.
        	\begin{itemize}[itemsep=0pt,leftmargin=18pt]
        		\item If $s_1 = 1$ and $s_i = 0$ for $i\geq 2$, and $g(i_1,\ldots,i_k) = \sum_{l\in [k]}s_{i_l}$, we call $\score$ the~SNTV~rule.
        		\item If $s_i = m-i$ for $i\in [m]$, and $g(i_1,\ldots,i_k) = \sum_{l\in [k]}s_{i_l}$, we call $\score$ the Borda rule.
        		\item If $s_i = m-i$ for $i\in [m]$, and $g(i_1,\ldots,i_k) = \max_{l\in [k]}s_{i_l}$, we call $\score$ the $\ell_1$-CC rule. 
        	\end{itemize}
        \end{example}
         
        \paragraph{Approval-based functions.}
        Suppose each voter $v\in V$ approves a subset $A_v$ of $m'$ candidates (that are the first $m'\in [m]$ candidates in $\succ_v$; note that $m'$ is the same for all voters).
        An approval-based function awards a score to each committee as
        \[
        \app(S) = \sum\nolimits_{v\in V} g(|S\cap A_v|),
        \]
        where $g: [k]\rightarrow \R_{\geq 0}$ is a non-decreasing concave function.
        The above definition captures the class of OWA-rules, which are parameterized by an OWA-vector $\lambda=(\lambda_1,\dots, \lambda_k)$ with $\lambda_1\geq \dots \geq \lambda_k \geq 0$ and correspond to $g(i)=\sum_{j=1}^i \lambda_j$.
        Among others, the class of OWA-functions contains the popular PAV ($\lambda=(1,\frac{1}{2},\dots, \frac{1}{k})$) and $\ell_{\rm min}$-CC ($\lambda=(1,0,\dots, 0)$) rules. 
        We note that our main result \cref{thm:main_algorithmic} can also be extended to approval-based rules where the size of the approval set can be different for different voters.
\end{document}